\documentclass[aps,prx,twocolumn,superscriptaddress,floatfix,10pt]{revtex4-2}
\usepackage[american]{babel}
\usepackage[normalem]{ulem}
\usepackage{mathpazo}
\usepackage{lmodern}
\usepackage{physics}
\usepackage{mathtools}
\usepackage{amsmath}
\usepackage{amssymb}
\usepackage{amsfonts}
\usepackage{amsthm}
\usepackage{graphicx}
\usepackage{braket}
\usepackage{subfigure}
\usepackage{dsfont}
\usepackage{indentfirst}
\usepackage{xcolor}
\usepackage{mathrsfs}
\usepackage{algpseudocode}
\usepackage[ruled,vlined,linesnumbered]{algorithm2e}
\usepackage[unicode=true,bookmarks=false,breaklinks=false,pdfborder={0 0 1},backref=false,colorlinks=false]{hyperref}
\usepackage{bbm}
\usepackage{multirow}
\usepackage{hhline}
\usepackage{xspace}

\newcommand{\quotes}[1]{``#1''}
\newcommand{\eg}{\emph{e.g.}\xspace}
\newcommand{\ie}{\emph{i.e.}\xspace}

\newcommand{\identity}{\mathbbm{1}}
\newcommand{\thetabf}{\boldsymbol{\theta}}

\newcommand{\etat}{\eta_{t}}
\newcommand{\vecu}{\boldsymbol{u}}
\newcommand{\vecv}{\boldsymbol{v}}
\newcommand{\vecx}{\mathbf{x}}

\newcommand{\veca}{\mathbf{a}}
\newcommand{\normx}{\norm{\vecx}}
\newcommand{\sphere}{\mathbb{S}^{d-1}}
\newcommand{\euclidean}{\mathbb{R}^{d}}
\newcommand{\metric}{\boldsymbol{g}}
\newcommand{\loss}{\mathcal{L}}
\newcommand{\partialbf}{\boldsymbol{\partial}}
\newcommand{\jacobian}{\mathbf{J}}

\newcommand{\cnot}{$\operatorname{CNOT}~$}
\newcommand{\xxz}{$\operatorname{XXZ}~$}
\newcommand{\tfim}{$\operatorname{TFIM}$}
\newcommand{\hwe}{$\operatorname{HWE}$}
\newcommand{\he}{$\operatorname{HE}$}
\newcommand{\heop}{$\operatorname{HE \, OP}$}
\newcommand{\qol}{$\operatorname{QOL}$}

\DeclareMathOperator{\Var}{Var}
\DeclareMathOperator{\expmap}{ExpMap}
\DeclareMathOperator{\Grad}{grad}
\DeclareMathOperator{\Time}{TIME}

\newcommand{\chem}[1]{\ensuremath{\operatorname{#1}}}

\definecolor{darkGreen}{RGB}{0,110,0}

\theoremstyle{definition}
\newtheorem{theorem}{Theorem}[section]

\newcommand{\TII}{\affiliation{Quantum Research Center, Technology Innovation Institute, Abu Dhabi, UAE}}
\newcommand{\LIPSIX}{\affiliation{Laboratoire d'Informatique de Paris 6, CNRS, Sorbonne Universit\'{e}, Paris, France}}
\newcommand{\HKUSTGZ}{\affiliation{Thrust of Artificial Intelligence, Information Hub, The Hong Kong University of Science and Technology (Guangzhou), Guangzhou 511453, China}}


\begin{document}

\title{Quantum optimization with exact geodesic transport}

\author{Andr\'{e} J. Ferreira-Martins}
\email{andre.jfm.sci@gmail.com}
\TII
\LIPSIX

\author{Renato M. S. Farias}
\TII

\author{Giancarlo Camilo}
\TII

\author{Thiago O. Maciel}
\TII

\author{Allan Tosta}
\TII

\author{Ruge Lin}
\HKUSTGZ

\author{Abdulla Alhajri}
\TII

\author{Tobias Haug}
\TII

\author{Leandro Aolita}
\TII

\begin{abstract}
We introduce an architecture for variational quantum algorithms that can be efficiently trained via parameter updates along exact geodesics on the Riemannian state manifold. 
This features a parameter-optimal circuit ansatz which supersedes known quantum natural gradient methods by removing expensive estimations of the metric tensor and provably reducing gradient estimation costs by $62.5\%$.
Moreover, the framework also naturally incorporates conjugate gradients as a built-in feature, giving an accelerated descent method with convergence guarantees that we dub \emph{exact geodesic transport with conjugate gradients}.
Numerical benchmarks against state-of-the-art variational methods for ground-state preparation of molecular Hamiltonians or $1$-dimensional spin chains (both with and without particle-number conservation) up to $n=16$ qubits show reductions of over one order of magnitude in the number of optimization steps, with global convergence even for degenerate cases and competitive quantum-resource scalings.
In addition, we perform proof-of-principle demonstrations on \textrm{IonQ}'s \texttt{Forte} quantum processor, showcasing deployment of pre-trained circuits for the $\chem{H_{3}^{+}}$ molecule and experimental training for $\chem{H_{2}}$.
Our work enables quantum machine learning applications with shorter training runtime, with implications at the interface of quantum simulation, differential geometry, and optimal control theory.
\end{abstract}

\maketitle
\section{Introduction}
\label{sec:introduction}

Variational quantum algorithms (VQAs) are hybrid quantum-classical methods
with a variety of promising applications using current noisy and future fault-tolerant devices~\cite{Cerezo2021, Bharti2022}. 
The idea is to use a classical computer to iteratively optimize the parameters of a parametrized quantum circuit (PQC) to minimize a target cost function of the circuit's output state. 
However, this classical-quantum optimization loop represents a major challenge in practice, largely due to the large number of iterations required and the presence of flat optimization landscapes~\cite{Larocca2025}. 
The \emph{quantum natural gradient} (QNG)~\cite{Stokes2020, Koczor2022} is a first-order Riemannian optimization method~\cite{absil2008, sato2021} that improves gradient-based VQAs optimization, aiming at faster convergence by accounting for the curvature of the space of quantum states produced by a PQC architecture. 
The QNG moves in the steepest descent direction determined by the metric tensor $\metric$ of the space. 
Further attempts to improve the QNG descent include second-order geodesic corrections~\cite{Halla2025secondorder} and conjugate gradient methods~\cite{Halla2025conjugate}, the latter including also a memory from the previous step. 
However, given that for most PQC ansatze the metric tensor is not a priori available, implementing the QNG in practice comes with a significant overhead of empirically estimating it on quantum hardware~\cite{Haug2022}. 
For a generic PQC with $M$ free parameters, this requires $\order{M^{2}}$ measurements at each optimization step or even more when higher-order corrections are added (\eg $\order{M^{3}}$ measurements are needed to estimate generic second-order geometric quantities~\cite{Halla2025secondorder}). 
Attempts to lower the cost of estimating the metric to $\order{M}$ rely on block-diagonal approximations of $\metric$~\cite{Yao2022, Meyer2023, Qi2024, Fitzek2024, DellAnna2025, GomezLurbe2025}, though it has been shown that they may lead to suboptimal optimization due to loss of information on parameter correlation~\cite{Wierichs2020}.
Further attempts use stochastic approximations of $\metric$ to reduce the cost to $\mathcal{O}(1)$ 
\cite{Straaten2021, Gacon2021, Gacon2023, Halla2025fisher}.
However, numerical evidence suggests that their approximations deteriorate the performance when compared to the QNG with the exact metric tensor. 
To date, it is unknown how to benefit from Riemannian optimization for training VQAs in practice.

In this work, we develop VQAs based on a specific class of PQCs whose circuit parameters are the hyperspherical coordinates of the vector of quantum amplitudes. 
In these coordinates, the metric tensor $\metric$ is a diagonal matrix with a closed-form expression, which completely removes the need for empirical estimations of it.
Moreover, this leads to analytical solutions for geodesic paths that allow for parameter-update rules via \emph{Exact Geodesic Transport} (EGT), see Fig. \ref{fig:schematics}.
Under the assumption of small step-sizes (learning rates), the QNG~\cite{Stokes2020} and its second-order correction~\cite{Halla2025secondorder} correspond respectively to only the first and second-order Taylor series approximations to the exact geodesic paths leveraged by EGT. 
Moreover, we also propose \emph{Exact Geodesic Transport with Conjugate Gradients} (EGT-CG), an enhanced optimizer with provable global convergence guarantees as a built-in feature and faster convergence observed in all analyzed examples. 
Additionally, we show how to estimate the gradients required at each iteration using  $62.5\%$ fewer quantum resources than with the standard parameter shift rule (PSR)~\cite{Schuld2019, Anselmetti2021, Kottmann2021}.

We apply our optimizer to ground-state optimization (GSO) for molecular and spin-chain Hamiltonians, and benchmark it against state-of-the-art optimizers.
Numerical simulations for molecules with up to $n=14$ spin orbitals show that EGT-CG consistently displays a superior performance in the number of iterations required for convergence (with up to $20.5$ times speedups over the Adam optimizer), even in the challenging case of degenerate ground states. 
We also benchmark the scheme for $1$-dimensional chains of up to $n=16$ spins,  with (XXZ model) and without (transverse-field Ising model) particle-number conservation, showing over an order of magnitude faster convergence to chemical accuracy than, \emph{e.g.}, the Adam optimizer.  
In the case of XXZ, we show that this is possible with no overhead in the total CNOT-gate count.  
Moreover, we perform demonstrations on the \texttt{Forte} ion-trap quantum processor from \textrm{IonQ}~\cite{IonQ}: 
For the $\chem{H_{3}^{+}}$ molecule, we prepare the ground state at $6$ bond lengths using a classically pre-trained circuit. 
In turn, for the $\chem{H}_{2}$ molecule, we perform $3$ full steps of the EGT-CG optimization loop, including experimental gradient estimation, reaching the target ground-state energy up to a $\sim 4\%$ deviation. 
Such precision, even under hardware noise, is possible thanks to an error mitigation functionality based on particle-number preservation intrinsic to our scheme. 
Finally, on a more technical level, our framework also includes a mechanism to avoid metric singularities via parameter re-initializations and, in the case of molecules, allows an efficient state initialization via a problem-informed warm start state with high overlap with the Hartree state that yields even faster convergence and avoids barren plateaus~\cite{Larocca2025}.

\section{Results}
\label{sec:results}
\subsection{Framework}
\label{ssec:framework}

We use the standard setup of variational quantum algorithms~\cite{Cerezo2021}: one is given an $n$-qubit parametrized quantum circuit $U_{\thetabf}$ (the \emph{VQA ansatz}) with a set of parameters $\thetabf$ that prepares a quantum state $\ket{\psi(\thetabf)} = U_{\thetabf} \, \ket{0^{n}}$, where the superscripts represent a bit value that gets repeated (\eg $0^{3}1^{2} \equiv 00011$). 
The goal is to find the parameters $\thetabf^{\star}$ that minimize a given loss function $\loss(\thetabf)$ of interest. 
The end-to-end algorithm is a classical-quantum method in which the loss function value and gradient at any specific $\thetabf$ are queried using the quantum device, while a classical computer updates $\thetabf$ iteratively in search of $\thetabf^{\star}$. 
Here, we focus on ground-state optimization, where $\loss(\thetabf) \coloneq \bra{\psi(\thetabf)} H \ket{\psi(\thetabf)}$ is the expectation value of a Hamiltonian $H$, although we emphasize our methods generalize to arbitrary $\loss(\thetabf)$. 

Given $U_{\thetabf}$, the set of all reachable quantum states defines a hypersurface on the Hilbert space.
Whenever this surface is a manifold, we can define a Riemannian metric $\metric$ by taking the real part of the quantum geometric tensor~\cite{Liu2020quantum, Stokes2020, Meyer2021}
\begin{align}\label{eq:metric_components}
    g_{j\ell}(\thetabf) = \Re{\braket{\partial_{\theta_{j}}\psi \, \vert \, \partial_{\theta_{\ell}}\psi} - \braket{\partial_{\theta_{j}}\psi \, \vert \, \psi} \!\! \braket{\psi \, \vert \, \partial_{\theta_{\ell}}\psi}}\,,
\end{align}
where $\partial_{\theta_{j}} \coloneqq \partial / \partial \theta_{j}$ denotes partial derivatives with respect to a variable $\theta_{j}$. 
Eq. \eqref{eq:metric_components} is also known as the quantum Fisher information metric, or Fubini-Study metric. 
From $\metric$, various geometrical properties can be derived, such as geodesics and the exponential map describing parallel transport along geodesics (see App. \ref{sec:differential_geometry}). 

\begin{figure}[!t]
    \includegraphics[width=\columnwidth]{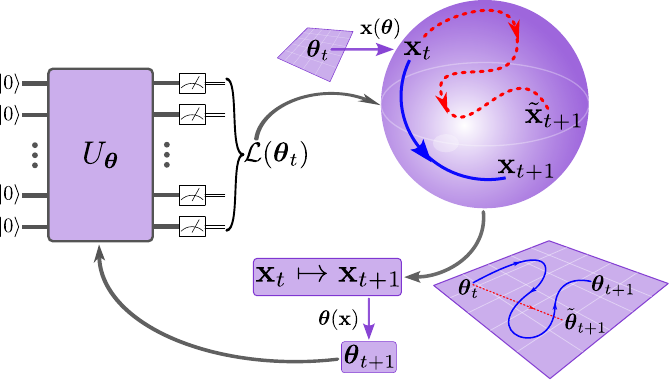}
    \caption{
        \textbf{Schematic illustration of our method.} 
        VQA pipeline with standard gradient descent (GD, red dashed curves) update rule \emph{vs.} the proposed descent with exact geodesic transport (EGT, blue solid curves). 
        In GD, the search for a minimum of the loss function $\loss(\thetabf)$ is done via steps along straight lines on a flat parameter space, whereby the parameters $\thetabf$ are updated directly.
        In contrast, EGT implements an amplitude-based update rule, based on paths along geodesics on the curved manifold that contains the state vector $\ket{\psi(\thetabf)}$ output by the variational circuit $U_{\thetabf}$.
        For our circuit ansatz, such paths define great-circle arcs on a hypersphere.
        Both spaces are related via the (exact) coordinate transformation $\thetabf = \thetabf(\vecx)$ between circuit parameters $\thetabf$ and the amplitudes $\vecx$ of $\ket{\psi(\thetabf)}$ in the computational basis, and its inverse $\vecx(\thetabf)$.
        This allows one to perform the amplitudes update via EGT, then recover the updated parameters $\thetabf_{t+1}$ from the updated amplitudes $\vecx_{t+1}$, as in Eq. \ref{eq:update_exact}.
    }
    \label{fig:schematics}
\end{figure}

Here, we consider as VQA ansatze the exact amplitude encoders introduced in Ref.~\cite{Farias2025}. 
In it, they construct a $\thetabf$-parameterized quantum circuit that prepares an arbitrary state $\ket{\psi(\thetabf)} = \sum_{j=1}^{d} \, x_{j}(\thetabf) \, \ket{b_{j}}\,$, with $\vecx \coloneqq \{x_{j}(\thetabf)\}_{j \in [d]}$ being an $\ell_{2}$-normalized complex amplitude vector, in the span of any subset $B$ of $d \leq 2^n$ computational basis states of a $n$-qubit system.
For simplicity, we focus on the case of real-valued amplitudes (for the complex case, see App. \ref{sec:differential_geometry}), since, for all the examples analyzed, the Hamiltonian $H$ has real-valued components in the computational basis.
In this case, the circuit consists of $M=d-1$ parameterized gates whose angles $\thetabf \equiv \thetabf(\vecx)$ are the hyperspherical coordinates of the amplitude vector $\vecx$ on the $(d-1)$-dimensional sphere (see Ref.~\cite{Farias2025} for details). 
As a result, using these ansatze, we guarantee full expressivity, \ie, overparametrization~\cite{Larocca2023, Haug2021capacity, Haug2024} in the subspace spanned by $B$, with the minimum possible number of parameters.
In addition, we fully characterize the absence of (probabilistic)  barren plateaus~\cite{Larocca2025} for $\mathrm{poly}(n)$-dimensional (sub)spaces in the optimization landscape. App. \ref{app:barren-plateau} shows the variance of the loss function $\loss$ and its (natural) gradient components.

The manifold spanned by the set of all states $\ket{\psi(\thetabf)}$ reachable by this ansatz is diffeomorphic to the ($d-1$)-dimensional sphere $\sphere$, whose metric $\metric$ is diagonal in the coordinate basis, with components $g_{11} = 1$ and $g_{jj} = \prod_{\ell = 1}^{j-1} \, \sin^{2}(\theta_{\ell})$ for $j \in [2, \, d - 1]$.
Consequently, the exponential map on $\sphere$, for a starting point $\vecx$ with tangent vector $\vecv$, is given by~\cite{absil2008, sato2021}
\begin{align}
    \expmap_{\vecx}(\eta \, \vecv) = \cos(\eta \, \norm{\vecv}) \, \vecx + \sin(\eta \, \norm{\vecv}) \, \frac{\vecv}{\norm{\vecv}} \,,
    \label{eq:expmap}
\end{align} 
where $\eta$ is an affine parameter, and $\norm{\vecv} \coloneqq \sqrt{\braket{\vecv, \, \vecv}_ {\vecx}}$, with $\braket{\vecu, \, \vecv}_{\vecx} \coloneq  \vecu^{T} \vecv$ being the induced metric in the tangent space of $\vecx \in \sphere$. 
See App.\ref{sec:differential_geometry} for derivation.

\subsection{VQAs with exact geodesic transport}
\label{ssec:EGT}

We propose a gradient-based parameter optimization with exact geodesic descent that removes the issue of metric estimation on hardware, solving the main practical 
bottlenecks of previous QNG approaches, while replacing low-order approximations of the geodesic flow by the exact procedure.

Our Exact Geodesic Transport (EGT) parameter-update rule, illustrated in Fig. \ref{fig:schematics} reads
\begin{align}
    \vecx_{t+1} = \mathrm{ExpMap}_{\vecx_{t}}(\etat\,\vecv_{t}) \, , \quad \thetabf_{t+1} = \thetabf(\vecx_{t+1}) \, ,
    \label{eq:update_exact}
\end{align} 
where $\etat$ is the learning rate at step $t$, and $\vecv_{t} \coloneq - \jacobian(\thetabf_{t}) \, \metric^{-1}(\thetabf_{t}) \, (\partialbf_{\thetabf}\loss)_{\thetabf_{t}}$ (see App. \ref{app:riemannian_opt_wolfe}).
Here, $\thetabf(\vecx)$ is the standard hyperspherical-coordinate transformation (see Ref. \cite[Eq. $(7)$]{Farias2025}), $\metric^{-1}(\thetabf_{t}) \, (\partialbf_{\thetabf}\loss)_{\thetabf_{t}}$ is the \emph{natural gradient} of $\loss(\thetabf)$, $\jacobian(\thetabf)$ is the Jacobian of the transformation $\thetabf(\vecx)$, and $\partialbf_{\thetabf} \coloneqq \{\partial_{\theta_{j}}\}_{j\in[\abs{\thetabf}]}$ is used for the vector of all partial derivatives over $\abs{\thetabf}$ variables. 
We denote $(\partialbf_{\thetabf} \, \loss)_{\thetabf^{\prime}}$ as the derivative of the function $ \loss $ w.r.t. $\thetabf$ and evaluated at $\thetabf = \thetabf^{\prime}$. 
For $\etat \, \norm{\vecv_{t}} \ll 1$, first- and second-order Taylor approximations of Eq. \eqref{eq:update_exact} recover, up to a coordinate transformation, the standard QNG update rule~\cite{Stokes2020} and its second-order correction~\cite{Halla2025secondorder}, respectively. 
In addition, in App. \ref{app:relation} we show how Eq. \eqref{eq:update_exact} corresponds to an approximation of quantum imaginary-time evolution (QITE) at each step $t$, which becomes exact whenever $\loss$ is the fidelity between $\ket{\psi(\thetabf)}$ and a pure state.

Having analytical access to geometric quantities on the sphere allows us to improve the EGT update rule in Eq. \eqref{eq:update_exact} using the Riemannian conjugate gradient (CG) method in~\cite{sato2022}.
The resulting Exact Geodesic Transport with Conjugate Gradient (EGT-CG) optimizer takes the same form as Eq. \eqref{eq:update_exact}, replacing $\vecv_{t}$ by a new tangent vector $\vecu_{t}$ that carries a memory from the previous step, given recursively by 
\begin{align}
\vecu_{0} \coloneq \vecv_{0} \, , \quad \vecu_{t} \coloneq \vecv_{t} + \beta_{t} \, \mathcal{T}_{\etat \vecu_{t}}(\vecu_{t-1}) \, ,
\label{eq:egt-cg-update}
\end{align}
where $\beta_{t}\geq 0$ and the expression $\mathcal{T}_{\etat \vecu_{t}}(\vecu_{t-1}) \coloneq \cos(\etat \, \norm{\vecu_{t-1}}) \, \vecu_{t-1} - \sin(\etat \, \norm{\vecu_{t-1}}) \, \norm{\vecu_{t-1}} \, \vecx_{t-1}$ is the exact vector transport of the previous direction $\vecu_{t-1}$ from the tangent space at $\vecx_{t-1}$ to the tangent space at $\vecx_{t}$ (see App. \ref{sec:differential_geometry}). 
The choice of $\etat$ and $\beta_{t}$ dictates the performance and convergence guarantees of the optimizer, with $\beta_{t} = 0$ for all $t$ recovering the EGT update rule in Eq. \eqref{eq:update_exact}. 
For the GSO problem, EGT-CG allows convergence guarantees by scheduling $\etat$ using the so-called \emph{strong Wolfe conditions}~\cite{wolfe1969,wolfe1971},
and $\beta_{t}$ using the hybrid conjugate gradient, as discussed in detail in App. \ref{app:riemannian_opt_wolfe}. 
We compare these convergence guarantees to heuristic choices of $\etat$ and $\beta_{t}$ using Bayesian optimization~\cite{Movckus1975, Wang2023} in App. \ref{app:global_convergence}, showing that assuring Wolfe conditions can lead to significantly better performance in cases where the ground state is degenerate, whilst using fewer loss function evaluations. 
In Methods \ref{ssec:ansatz_opt_features}, we show extra benefits and features when using EGT-CG in conjunction with the ansatz in Ref.~\cite{Farias2025}.

\subsection{Numerical simulations}
\label{ssec:results_numerical}

We numerically benchmark the performance of the EGT-CG optimizer against numerous optimizers and learning rate schedulers (referred to as \emph{optimization schemes}) on the GSO problem. 
We show results for the electronic structure problem of $14$ different molecules, and the performance scaling with system size using the \xxz spin-chain Hamiltonian.

\subsubsection{EGT-CG for the electronic structure problem}
\label{sssec:electronic_structure}
\begin{figure*}[!t]
    \includegraphics[width=\textwidth]{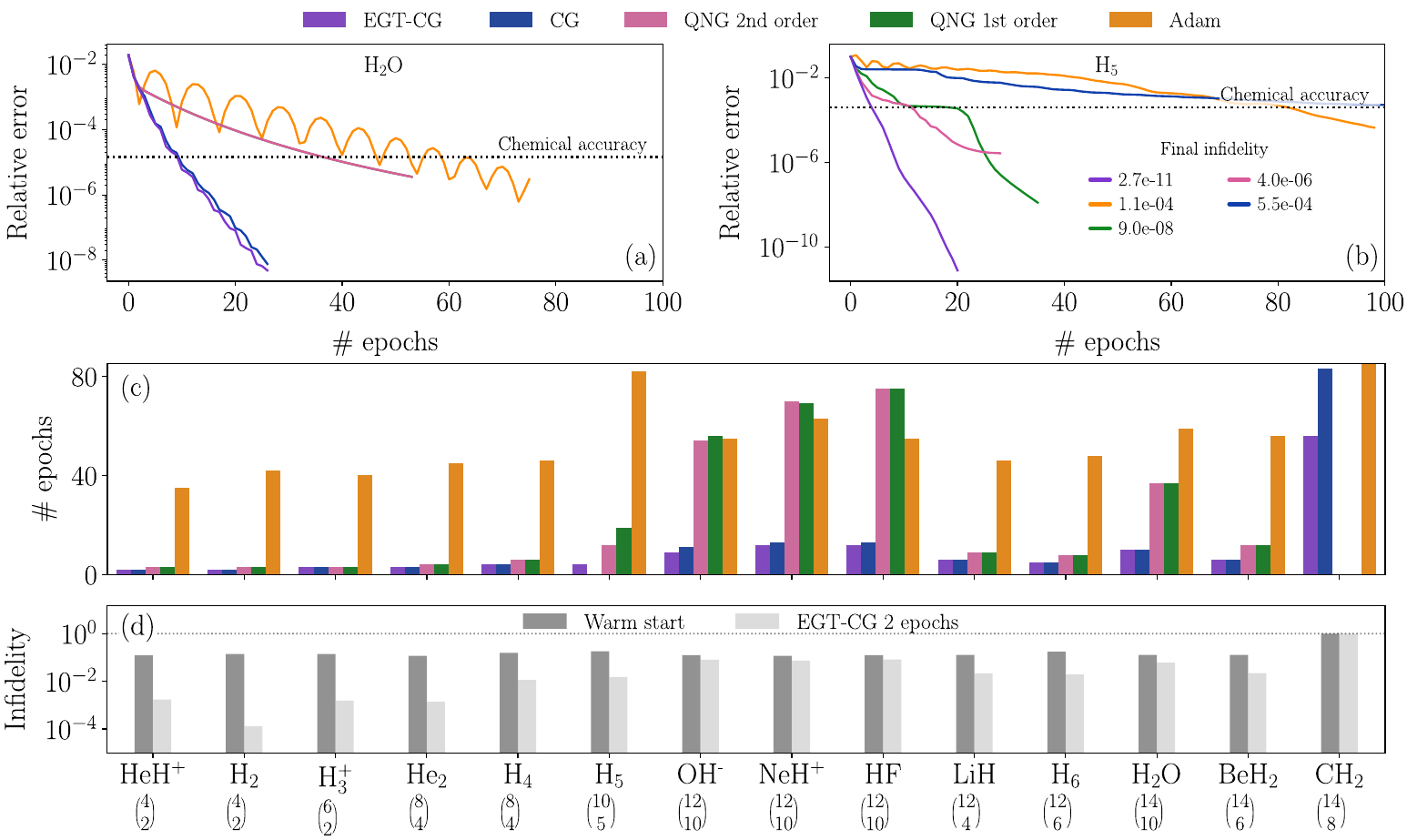}
    \caption{
        \textbf{Ground-state optimization for molecules.}
        Performance of different optimization schemes using $\operatorname{HWE}_{k}$ as VQA ansatz and the warm start in Eq. \eqref{eq:warm_start} with $\alpha=0.9$.
        The different optimizers are: exact geodesic transport with conjugate gradients (EGT-CG, purple), conjugate gradient method with flat-space gradients (CG, blue), quantum natural gradient (QNG) of first (green) and second (pink) orders, and the standard Adam optimizer (orange). 
        See Methods \ref{ssec:details_numerical} for the respective learning rate schedulers. 
        (\textrm{a}) Relative ground-state energy error for \chem{H_{2}O} \emph{vs.} number of epochs.
        Using the Jordan-Wigner transformation and the STO-3G basis set, the ground state is represented by $n = 14$ spin-orbitals as the active space with $k = 10$ electrons. 
        EGT-CG and CG display a similar performance, both comfortably beating the other methods ($\eta = 0.005$ for Adam here). 
        (\textrm{b}) Relative ground-state energy error for \chem{H_{5}}, represented by $n = 10$ spin-orbitals and $k = 5$ electrons. 
        The color code is the same as in (\textrm{a}), and here Adam has $\eta = 0.05$.
        The infidelity $1 - F_{\text{final}}$ of the state prepared by each scheme at the end of the optimization relative to the true ground state is displayed in the inset. 
        (\textrm{c}) Number of epochs to achieve chemical accuracy for $14$ molecules and different optimizers. 
        For \chem{H_{5}} and \chem{CH_{2}}, some optimization schemes did not converge to chemical accuracy, in which case no bar is shown. 
        \chem{CH_{2}} with Adam reached chemical accuracy in 130 epochs ($y$-axis was limited for better presentation).
        (\textrm{d}) Infidelities relative to the true ground state for the initial state $\ket{\psi_\text{warm}}$ with $\alpha=0.9$ and for the state after 2 epochs of EGT-CG, for the same molecules of (\textrm{c}). 
        Below each molecule label in the $x$-axis, the size of the corresponding parametrized subspace is indicated by $\binom{n}{k}$. 
    }
    \label{fig:molecules_comparison}
\end{figure*}

Given a molecule having $k$ electrons and a basis set describing its orbitals, the Jordan-Wigner transform~\cite{Jordan1928, Tranter2018, Tilly2022} maps its fermionic Hamiltonian into a qubit Hamiltonian $H$ where each qubit corresponds to one of the $n$ spin-orbitals.
The ground state of the qubitized Hamiltonian has real-valued amplitudes and is supported on the $\binom{n}{k}$-dimensional subspace of Hamming-weight $k$ (HW-$k$) computational basis states. Hereinafter, we denote by $\operatorname{HWE}_{k}$ the VQA ansatz in Ref.~\cite{Farias2025} when selecting the amplitudes in fixed-Hamming weight subspaces.
$\operatorname{HWE}_{k}$ also allows to initialize the circuit parameters $\thetabf$ 
to produce arbitrary warm start states. For this problem, we initialize the optimization in $\ket{\psi_\text{warm}}$ (see Methods \ref{ssec:ansatz_opt_features} for details), which has high fidelity with the Hartree state, known to provide a rough approximation to the true ground state of some molecules~\cite{Szabo1996}. 

Fig. \ref{fig:molecules_comparison}(a) shows our results for the \chem{H_{2}O} molecule.
The EGT-CG and the CG optimizers with flat-space gradients achieve chemical accuracy after $10$ epochs, which is $3.7$ times faster than the QNG optimizers of first and second orders, and $4.7$ times faster than the best Adam optimizer tested.
From these results, one might be tempted to conclude that the performance advantage of EGT-CG is mainly attributed to the CG method and not to the properties of the curved manifold.  
However, Fig. \ref{fig:molecules_comparison}(b) shows numerical evidence that this is not the case: we repeat the analysis for the \chem{H_{5}} molecule --- a free radical with a degenerate ground state at the equilibrium bond length --- and observe that EGT-CG achieves chemical accuracy after only $4$ epochs while the CG optimizer reaches a plateau near chemical accuracy without ever crossing it. 
There, EGT-CG is also $3$ and $4.75$ faster than the first and second order QNG methods, respectively, and $20.5$ faster than the Adam optimizer.
In the inset of Fig. \ref{fig:molecules_comparison}(b), we also display the infidelities of the final state achieved by each optimization scheme, relative to the true ground states found by direct diagonalization.
EGT-CG achieves by far the best accuracy, both in terms of ground state fidelity as well as its energy.

The analysis above was repeated for $12$ other molecules.
All displayed similar patterns to those observed in Figs. \ref{fig:molecules_comparison}(a)-(b). 
Fig. \ref{fig:molecules_comparison}(c) shows a comparison of the number of epochs needed by each optimizer to achieve chemical accuracy in the GSO problem for each molecule. 
EGT-CG has the best performance overall, consistently reaching chemical accuracy in fewer epochs than the other optimizers. 
For $12$ of the $14$ molecules, the CG optimizer with flat-space gradients performs comparably to EGT-CG. 
However, for those with a degenerate ground state --- namely, \chem{H_{5}} and \chem{CH_{2}} ---, either the CG or the QNG optimizers fail to achieve chemical accuracy while the EGT-CG still shows the fastest convergence. 
This highlights the practical importance of the global convergence guarantees of EGT-CG. 

In Fig. \ref{fig:molecules_comparison}(d), we plot the infidelity of the warm start, $\ket{\psi_\text{warm}}$, relative to the true ground state found by numerical diagonalization.
The fidelities where calculated as $\sum_{j=1}^s | \langle \psi_{\text{warm}} | \alpha_{j} \rangle |^2$, where $\{ \ket{\alpha_j}\}_{j=1}^s$ are the eigenstates with degeneracy $s$ ($s=2$ and $s=3$ respectively for $\chem{H_{5}}$ and $\chem{CH_{2}}$, and $s=1$ for all others). 
The small infidelities highlight the fact that $\ket{\psi_\text{warm}}$ is indeed adequate for all molecules considered except for $\chem{CH_{2}}$, for which the Hartree state is nearly orthogonal to the ground state subspace (see App. \ref{app:opt_curves_ch2} for a discussion). 
Notice that global convergence is guaranteed by EGT-CG with strong Wolfe conditions regardless of the initialization, though warm starts can only reduce the number of epochs to achieve optimality.  
Moreover, we verify that performing only $2$ epochs of EGT-CG already allows us to move significantly in the direction of the ground state (see the light-gray bars in Fig. \ref{fig:molecules_comparison}(d)). 

\subsubsection{Ground-state optimization for $1$d spin chains}
\label{ssec:xxz}

\begin{figure}[!t]
    \includegraphics[width=\columnwidth]{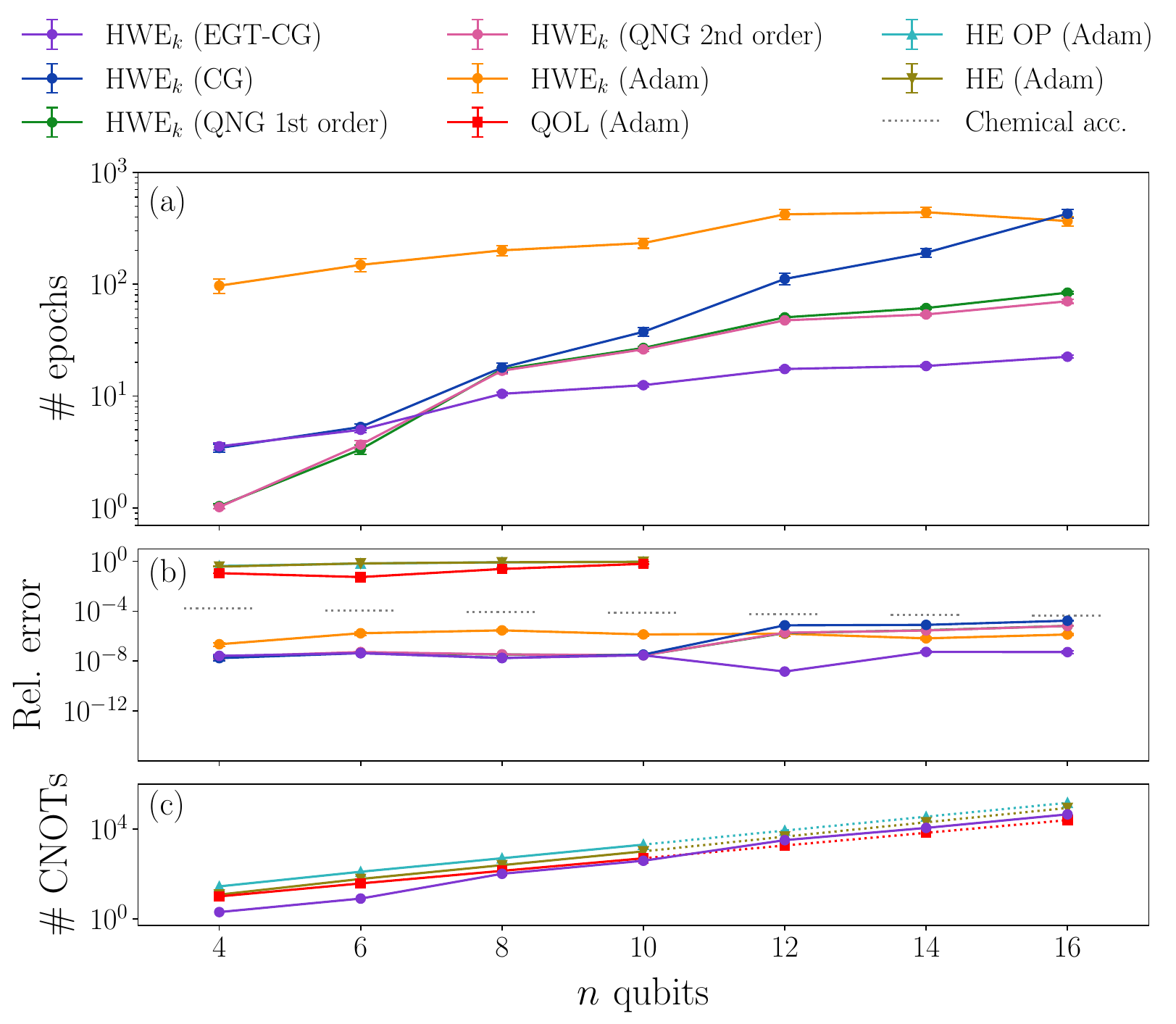}
    \caption{
        \textbf{Ground-state optimization for the XXZ model.}
        Performance of different variational circuit ansatze and optimization schemes as a function of system size~$n$.
        We compare the Hamming-weight encoder ($\operatorname{HWE}_{k}$) with all the considered optimization schemes against the quantum orthogonal layer (QOL), hardware-efficient (HE), and overparametrized HE (HE OP) ansatze optimized with Adam.
        Each point is averaged over $50$ Haar-random initializations, with error bars indicating the $95\%$ confidence interval.
        (\textrm{a}) Average number of epochs required to reach chemical accuracy in energy estimation; missing points indicate failure to reach chemical accuracy.
        (\textrm{b}) Relative energy estimation error at the end of training.  
        Chemical accuracy for each system size is plotted as a dotted black line.
        (\textrm{c}) Number of CNOT gates per ansatz; for $n>10$, CNOT counts are shown for scaling comparison only. 
        For further details on optimization schemes and quantum resources, see the Methods \ref{ssec:details_numerical} and App. \ref{app:quantum_resources}.
        }
    \label{fig:xxz_comparison}
\end{figure}

We then apply the EGT-CG optimizer to the GSO problem of the one-dimensional \xxz Hamiltonian $H_{\operatorname{XXZ}} \coloneqq \sum_{j=0}^{n-1} \, \left( X_{j} \, X_{j+1} + Y_{j} \, Y_{j+1} + \Delta \, Z_{j} \, Z_{j+1} \right)$ with closed-boundary condition. 
Here $\Delta$ is the anisotropy strength, and $\{X_{j}, \, Y_{j}, \, Z_{j}\}$ are the usual single-qubit Pauli operators acting on the $j$th qubit. 
For simplicity, we focus on the case $\Delta = 1/2$ and investigate $n$ even in the interval 
$n \in [4, \, 16]$. 
The ground state is known to have real-valued amplitudes supported on the \emph{half-filling} subspace, \ie, the $\binom{n}{n/2}$-dimensional subspace~\cite{Gomez1996}.
Apart from $\operatorname{HWE}_{k}$, here we also tried different ansatze (see Methods \ref{ssec:details_numerical} for details) with the Adam optimizer, using Haar-random initialization as no warm state is useful. 
Due to cyclic translational and parity symmetries in the system, only $d = \order{n^{-1} \binom{n}{n/2}}$ unique amplitudes in absolute value contribute to the ground state~\cite{Gomez1996}, leading to the possibility of restricting the circuit implementation when using $\operatorname{HWE}_{k}$ as ansatz.
Figure \ref{fig:xxz_comparison}(a) shows the average number of optimization epochs taken by each ansatz and optimization scheme to reach chemical accuracy as a function of the system size $n$.
Several ansatze fail to reach chemical accuracy across the tested system sizes.
This is further quantified in Fig. \ref{fig:xxz_comparison}(b), which shows the relative energy error at the end of training. 
Among all methods considered, EGT-CG combined with the $\operatorname{HWE}_{k}$ encoder consistently yields the best performance, achieving chemical accuracy up to 25 times faster than all competing optimization schemes and exhibiting a milder growth in the number of required epochs as the system size increases. 
These results show that EGT-CG has favorable scalability when applied to symmetry-adapted ansätze. 
In turn, Fig. \ref{fig:xxz_comparison}(c) compares the quantum resources required by the different circuit architectures in terms of CNOT counts.
The quantum orthogonal layer (QOL) is the only circuit ansatz exhibiting a more favorable CNOT scaling than the $\operatorname{HWE}_{k}$ ansatz used by EGT-CG.
However, it fails to converge to chemical accuracy even at small system sizes, while EGT-CG does. 

To end up with, our scheme is not restricted to optimizations on fixed particle-number subspaces.
For the case of the transverse-field Ising model, which does not obey this symmetry, we showed that EGT-CG still outperforms all the other schemes, reaching chemical accuracy up to 60 times faster (see App. \ref{app:tfim} for details).

\subsection{Experimental demonstrations on trapped ions}
\label{ssec:results_experimental}

\begin{figure}[!t]
    \includegraphics[width=\columnwidth]{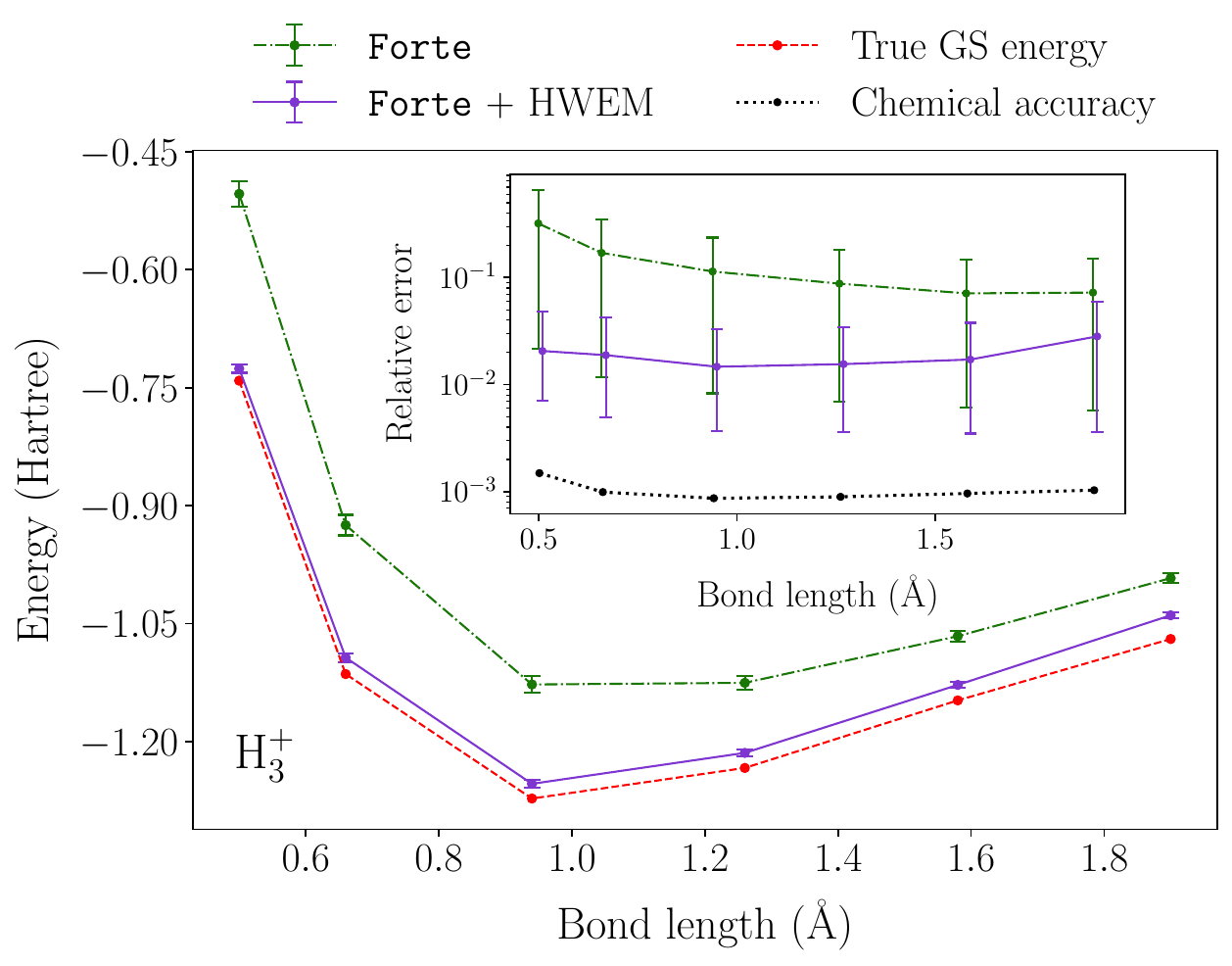}
    \caption{
        \textbf{Ground state energy of $\chem{H_{3}^{+}}$ on IonQ \texttt{Forte}.}
        (\emph{main}) Energy estimates (in Hartrees) for the \chem{H_{3}^{+}} molecule as a function of bond length (in Angstrom).
        The optimal circuit parameters were trained classically.
        The dashed red curve indicates the true ground-state energies.
        The dashed green curve is the result of the deployment on IonQ's \texttt{Forte} quantum processor.
        The solid purple curve shows the hardware results after debiasing and Hamming-weight-preservation error mitigation (HWEM) for $k = 2$.
        Error bars indicate the $95\%$ confidence interval for the mean.
        (\emph{inset}) Energy estimation error relative to the exact ground state energy as a function of the bond length.
        The color scheme is the same, with chemical accuracy as a dotted black line.
    }
    \label{fig:h3+_qhw}
\end{figure}

Lastly, we performed two proof-of-concept experimental demonstrations of the EGT-CG optimizer for molecular ground state optimizations on the \textrm{IonQ} \texttt{Forte} ion-trap quantum processor~\cite{IonQ}. 

The first is a quantum hardware deployment of classically-trained circuits that prepare the ground state of the $\chem{H_{3}^{+}}$ molecule at $6$ different bond lengths (see Methods \ref{ssec:details_experimental} for details). 
The main panel in Fig. \ref{fig:h3+_qhw} shows the estimated ground-state energy as a function of the bond length, while the inset shows the relative error against the true ground energy. 
For each measurement basis, we collected $10^4$ samples. 
The results using IonQ's default \emph{debiasing} error mitigation technique~\cite{debiasing} exhibit $10$-$30\%$ relative error in the estimated energies, which we attribute to circuit and readout noise in the device. 
We then perform, on top of debiasing, a Hamming-weight-based error mitigation (HWEM) by post-selecting the measurements of Pauli-$Z$ observables that have fixed Hamming weight $k = 2$. 
This reduces the relative errors by roughly an order of magnitude to $2$-$4\%$ with approximately $13\%$ of the measured bitstrings discarded on average, 
which we attribute to the fact that Pauli-$Z$ observables contribute the most to the \chem{H_{3}^{+}} Hamiltonian. 
However, the demonstration fails to reach chemical accuracy in both cases.

\begin{figure}[!t]
    \includegraphics[width=\columnwidth]{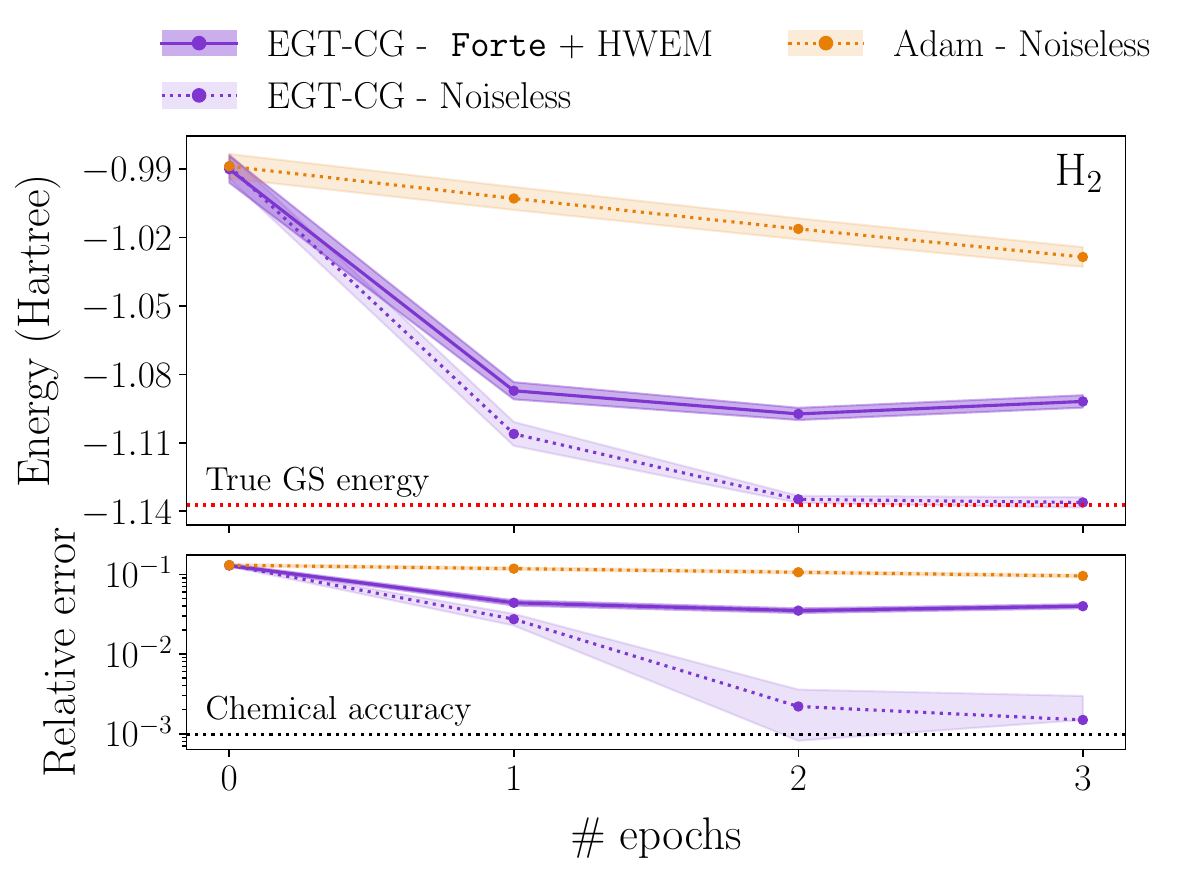}
    \caption{
        \textbf{EGT-CG optimization for the ground state of the $\chem{H_{2}}$ molecule on IonQ \texttt{Forte}.}
        (\emph{top}) Energy (in Hartree) during ground-state optimization for the $\chem{H_{2}}$ molecule, as a function of the number of epochs.
        In dark purple, we present the quantum hardware results for the EGT-CG optimizer with debiasing and HW $k=2$ error mitigation (HWEM). 
        We used $2 \times 10^{4}$ shots for each measurement basis required to estimate the loss function $\loss$ and the gradient $\partialbf_{\thetabf} \loss$ at the first epoch, and $4 \times 10^{4}$ shots at the last two.
        In light purple, we show the noiseless simulation of the same process, only with $2 \times 10^4$ shots in all epochs.
        In orange, we show a noiseless simulation of the iteration using the Adam optimizer with fixed learning rate $\eta = 10^{-2}$, and $2 \times 10^4$ shots in all epochs.
        The shaded regions represent the $95\%$ confidence interval for the mean.
        The true ground state energy is plotted as a dotted red line.
        (\emph{bottom}) Energy estimation error relative to the true ground state energy as a function of epochs.
        Chemical accuracy is plotted as a dotted black line.
        }
    \label{fig:h2_qhw}
\end{figure}

Second, we validate the end-to-end EGT-CG pipeline on the IonQ hardware by executing $3$ steps of optimization for the ground state optimization of the \chem{H_{2}} molecule (see Methods \ref{ssec:details_experimental} for details). 
As before, we mitigate the errors using debiasing and the Hamming-weight-based post-selection. 
The optimization curves are displayed in Fig. \ref{fig:h2_qhw}, with comparisons against two noiseless, finite-precision classical simulations using both EGT-CG and the Adam optimizers. 
We observe that the experimental results drop to and stabilize at a relative error of $\sim5\%$ after $3$ epochs, while the Adam optimizer (even in the absence of noise) only reaches a relative error around $10\%$ in the same number of steps.

\section{Discussion}
\label{sec:discussion}

We introduced \emph{exact geodesic transport with conjugate gradients} (EGT-CG), a framework to optimize circuit parameters along geodesic paths.
This opens up new possibilities for quantum machine learning theory at the nexus with differential geometry and optimal control theory, which may in turn lead to a new generation of practical VQAs.
Importantly, we emphasize that our framework is fully versatile, extending beyond particle-number-preserving Hamiltonians and non-particle-number-preserving models as well as, for instance, non-linear loss functions~\cite{Sciorilli2025}. 
We also highlight the importance of using the strong Wolfe conditions to schedule the learning rate, which guarantees that the optimizer is globally convergent (\ie, the optimization reaches a stationary point regardless of the initialization) even in cases where barren plateaus are present.

An important question is the relevance of our framework for quantum advantages with VQAs.
For that, a necessary condition is that the circuit ansatz is hard to simulate classically, which may be achieved as follows: assume we have a classically intractable state $V \, \ket{1^{k}0^{n-k}}$, where $V$ is an arbitrary, fixed unitary (\ie, independent of the optimization parameters) with polynomial circuit depth, and $k = \order{\log(n)}$.
This state may, for instance, be a known approximate solution to a given problem. 
Then, one can improve on that by variationally optimizing the state $\ket{\psi_{\mathrm{qa}}(\thetabf)} = V \, U_{\thetabf} \, \ket{1^{k}0^{n-k}}$ over $\thetabf$, where $U_{\thetabf}$ is our $\operatorname{HWE}_k$ ansatz with polynomial circuit depth.
Notably, since the metric is invariant under left multiplication by a fixed unitary, $V \, U_{\thetabf}$ has also an analytic metric and is therefore directly amenable to EGT-CG.

Alternatively, our scheme can also be used as an efficient primitive for input-state preparation for quantum algorithms.
For example, one can use EGT-CG for restricted ground-state optimization over a polynomially large subspace of an exponentially large target space and then use the output state as a warm start for quantum imaginary time evolution~\cite{Motta2019, Silva2023, Tosta2024, Silva2024}, eigenvalue thresholding~\cite{Lin2022, Wan2022, Tosta2024}, or quantum phase estimation~\cite{Berry2024}. 
This is particularly appealing for molecular Hamiltonians, where good warm-start states are known --- such as spin-projected matrix-product states~\cite{Li2017spin} --- but whose preparation requires an extensive number of auxiliary qubits~\cite{Berry2024}; 
whereas our method requires none. 
These are some of the questions for future exploration. 

\section{Methods}
\label{sec:methods}
\subsection{Ansatz/optimizer features}
\label{ssec:ansatz_opt_features}

The VQA ansatz of Ref.~\cite{Farias2025} is based on an analytic map $\vecx: [0, \, \pi]^{d-2} \cross [0, \, 2\pi) \subset \mathbb{R}^{d-1}\rightarrow\sphere$ (the hyperspherical coordinates map) from the circuit parameters $\thetabf$ to the state amplitudes with components $x_{1}(\thetabf)\coloneq\cos(\theta_{1})$, $x_{j}(\thetabf)\coloneqq \cos(\theta_{j})\prod_{\ell\in[j-1]} \, \sin(\theta_{\ell})$ for all $1<j<d$, and $x_{d}(\thetabf)\coloneqq \prod_{\ell\in[d-1]} \, \sin(\theta_{\ell})$. 
This simple form gives rise to a number of additional features that we briefly highlight below. 

\vspace{.2cm}

\paragraph{Efficient gradient estimation. }
The gradient at any step $t$ of the optimization is given by $\vecv_t \coloneq - \jacobian(\thetabf_t) \,   \metric^{-1}(\thetabf_t) \, (\partialbf_{\thetabf}\loss)_{\thetabf_t}$ with $\thetabf_t=\thetabf(\vecx_t)$. 
This can be estimated by first estimating $(\partialbf_{\thetabf}\loss)_{\thetabf_t}$ on hardware and then classically computing the action of $\jacobian(\thetabf_t) \,   \metric^{-1}(\thetabf_t)$ on it, which can be done recursively in $\mathcal{O}(d)$ space (see App. \ref{app:jacobian_regularization}). 
The exact map between $\vecx$ and $\thetabf$ allows us to reduce the required quantum resources to estimate $(\partialbf_{\thetabf}\loss)_{\thetabf_t}$ to $3/8$ of those using the standard parameter-shift rule (PSR). 
This follows from the following identity
\begin{align}
(\partial_{\theta_{\ell}} \loss)_{\thetabf_{t}} = g_{\ell\ell}^{1/2} \, \left[2 \, \loss_{\phi_{\ell}}(\thetabf_{t}) - \loss_{\varphi_{\ell}}(\thetabf_{t}) -  \loss_{\psi}(\thetabf_{t})\right] \, ,
\end{align}
where $\loss_\xi(\thetabf_{t})\coloneq \bra{\xi(\thetabf_{t})} H \ket{\xi(\thetabf_{t})}$ and $\ket{\phi_{\ell}(\thetabf_{t})}$, $\ket{\varphi_{\ell}(\thetabf_{t})}$ are states that can be prepared using the same circuit ansatz structure (see App. \ref{app:hwe_gradient} for details). 

\vspace{.2cm}

\paragraph{Efficient warm start.}
For the electronic structure problem, we initialize the optimization such that the corresponding state has high fidelity with the Hartree state $\ket{1^{k} 0^{n-k}}$, which is known to provide a rough approximation to the true ground state of some molecules~\cite{Szabo1996}. 
We propose, as a warm start, the superposition state
\begin{align}
    \ket{\psi_\text{warm}} = \sqrt{\alpha} \, \ket{1^{k} 0^{n-k}} + \sqrt{\frac{1- \alpha}{\binom{n}{k} - 1}} \, \sum_{\ket{b} \in B_{k}^{\prime}} \, \ket{b} \,,
    \label{eq:warm_start}
\end{align}
where $\alpha \in (0, 1)$ is a free parameter ($\alpha = 0.9$ was used in all the simulations), $B_{k}$ is the set of all $n$-qubit computational basis states of HW-$k$, and $B_{k}^{\prime} \coloneq B_{k} \setminus \{\ket{1^{k} 0^{n-k}}\}$. 
This state has fidelity $\alpha$ with the Hartree state, is supported over all the HW-$k$ basis states, and can be prepared deterministically using the $\operatorname{HWE}_{k}$ ansatz, but not with other standard VQA ansatze. 
The motivation for \eqref{eq:warm_start} is two-fold: while the high fidelity naturally reduces the overall number of optimization steps, populating all directions yields a sizable gradient that is beneficial for the gradient-based optimizer~\cite{Mhiri2025}. 

Notice that, in addition to the fixed HW subspace restriction, we could enforce spin symmetries to further reduce the effective subspace in which $\ket{\psi_{\text{warm}}}$ is supported. 
However, we numerically verified that this reduction was not significant for the system sizes analyzed.
Also, spin-projected matrix product states~\cite{Li2017spin} are known to provide good warm starts as well.
However, their preparation requires circuits of challenging sizes, in particular with an extensive number of auxiliary qubits.
For example, recent estimates for industry-relevant molecules give between $4$ and $5$ times more auxiliary qubits required than system qubits themselves \cite[Table 1]{Berry2024}.
In contrast, $\operatorname{HWE}_{k}$ allows us to prepare $\ket{\psi_{\text{warm}}}$ with circuits of depth $\binom{n}{k}$ but without extra qubits. 

\vspace{.2cm}

\paragraph{Avoidance of metric singularities.} 
A common issue of standard QNG methods is the appearance of metric singularities during the optimization, namely a parameter vector $\thetabf$ for which $\det(\metric(\thetabf)) \approx 0$. 
In this case, the metric tensor becomes ill-conditioned and $\metric^{-1}(\thetabf)$ appearing in the natural gradient leads to numerical instability. 
The standard way to mitigate this via standard pre-conditioning techniques, \eg Tikhonov regularization~\cite{Hoerl1970}. 
Here, we avoid this problem using the analytic properties of the ansatz: 
if $\sin(\theta_{j}) \approx \theta_{j} \lesssim \tau$, where $\tau$ is a constant threshold, then we set $\theta_{j} \rightarrow \pi / 2$.
Due to the product structure of the metric components and the diagonal structure of the metric, this suffices to single out the parameter distorting the natural gradient.
We empirically verified that $\tau = 10^{-3}$ yielded stable training for all numerical demonstrations presented in Sec. \ref{sec:results}.

\subsection{Numerical details}
\label{ssec:details_numerical}

The molecular Hamiltonians were taken from \texttt{PennyLane}'s Quantum Dataset~\cite{Pennylane2022, Utkarsh2023}  in the STO-3G basis set and at the equilibrium bond length. All simulations were run using \texttt{Qiboml}'s~\cite{Robbiati2025, QibomlZenodo} integration of the software packages \texttt{Qibo}~\cite{Qibo2021, QiboZenodo} and \texttt{PyTorch}~\cite{Pytorch2019}.
For all the numeric experiments, we worked at infinite precision, and gradients were calculated using \texttt{PyTorch}'s backpropagation method.
When applicable, we used \texttt{Qibo}'s backend that specializes in simulations of Hamming-weight-preserving circuits. 
We used Ref.~\cite{Nogueira2014} for the Bayesian optimization.

For all simulations using $\operatorname{HWE}_{k}$ as circuit ansatz, we compared different optimization schemes consisting in the combination of optimizers and learning rate schedulers. 
All combinations were tested, and we reported the results for the scheme leading the best results (in terms of faster convergence), which are: ($i$) EGT-CG; 
($ii$) a CG optimizer with standard gradient descent~\cite{dai2001} and learning rate schedules based on the strong Wolfe conditions~\cite{sato2022};
($iii$) the original (first-order) QNG~\cite{Stokes2020} as well as its second-order correction~\cite{Halla2025secondorder}, both with learning rates chosen via per-epoch Bayesian optimization~\cite{Movckus1975, Wang2023};
and ($iv$) the \emph{Adam} optimizer~\cite{Kingma2017} with constant learning rates $\eta \in \{0.001, \, 0.005, \, 0.01, \, 0.05, \, 0.1, \, 0.5\}$. 
The performance of each scheme was assessed using the relative error of ground-state energy estimation as a function of the number of iterations (or \emph{epochs}). 
For the \xxz and Ising model simulations, we also tried different ansatze, namely:
($i$) another HW-preserving ansatz composed of \emph{quantum orthogonal layers} (\qol) of nearest- and next-nearest-neighbor connectivity \cite[Fig. $2$]{Robbiati2024}. 
The number of free parameters is $M = \binom{n}{k}-1$;
($ii$) the brickwork \emph{hardware-efficient} (\he) ansatz~\cite{Kandala2017} that parametrizes the full Hilbert space with $M = 2^{n} - 1$ parameters;
($iii$) an overparametrized \he{} ansatz, referred to as \heop, with $M = 2 \, (2^{n} - 1)$ free parameters. 
For these alternative ansatze, only Adam with the aforementioned learning rates was used, as the other optimizers cannot be efficiently implemented.

In all cases, we set a limit of $10^{3}$ epochs, and the optimization was either halted $15$ epochs after reaching chemical accuracy or if an early stop was triggered by either $20$ consecutive epochs of loss function decrease (in magnitude) smaller than $10^{-4}$, or $10$ consecutive epochs of loss function increase.

\subsection{Experimental details}
\label{ssec:details_experimental}

To estimate $\loss(\thetabf^{\star})$ for $\chem{H_{3}^{+}}$ using the quantum hardware, we performed a usual grouping of mutually commuting Pauli observables, yielding a set of $7$ measurement bases for each bond length.
The bond lengths considered were (in \AA): $\{0.5, \, 0.66, \, 0.94, \, 1.26, \, 1.58, \, 1.9\}$.
This amounts to $42$ circuits executed. For each one, we collected $10^4$ samples.
The $\chem{H_{3}^{+}}$ electronic structure is represented by $n=6$ spin orbitals and $k=2$ electrons, with its ground state supported in the $15$-dimensional Hilbert subspace spanned by 6-qubit computational basis states of fixed HW $k = 2$. 
In all considered bond lengths, the classically-trained ground states are supported in the $3$-dimensional subspace spanned by $\left\{\ket{110000}, \, \ket{001100}, \, \ket{000011} \right\}$, which in hyperspherical coordinates can be covered by only $2$ parameters using the generalized RBS (gRBS) gates defined in Ref.~\cite{Farias2025}.
To deploy the circuits on IonQ \texttt{Forte}, we transpiled the gRBS gates into IonQ's so-called \emph{\quotes{quantum information science}} gateset~\cite{qis}. 
This process resulted in each circuit for the state preparation being transpiled into single- and two-qubit gate counts of $86$ and $36$, respectively.

For $\chem{H_2}$, as shown in Fig. \ref{fig:molecules_comparison}, the Hamiltonian can be expressed in $n = 4$ qubits representing spin-orbitals occupied by $k = 2$ electrons.
At each epoch, we estimated the loss function $\loss$ and its gradient $\partialbf_{\thetabf}\loss$ using the loss-function-based measurement scheme derived in App. \ref{app:hwe_gradient}, totaling $11$ loss function estimations and $55$ measurement sets per optimization step.
Due to budget constraints, we tested the strong Wolfe conditions using noisy classical simulations based on IonQ \texttt{Forte}'s reported noise parameters.
In the first epoch, each measured observable was estimated using $2 \times 10^{4}$ samples, while $4 \times 10^{4}$ samples per observable were collected in the remaining steps.

\acknowledgments
\noindent RL was affiliated with the Technology Innovation Institute for the largest part of the development of this work.
We thank Ilia Luchnikov for insightful discussions.
\bibliography{references}

 \newcommand{\noop}[1]{}
\begin{thebibliography}{84}%
\makeatletter
\providecommand \@ifxundefined [1]{%
 \@ifx{#1\undefined}
}%
\providecommand \@ifnum [1]{%
 \ifnum #1\expandafter \@firstoftwo
 \else \expandafter \@secondoftwo
 \fi
}%
\providecommand \@ifx [1]{%
 \ifx #1\expandafter \@firstoftwo
 \else \expandafter \@secondoftwo
 \fi
}%
\providecommand \natexlab [1]{#1}%
\providecommand \enquote  [1]{``#1''}%
\providecommand \bibnamefont  [1]{#1}%
\providecommand \bibfnamefont [1]{#1}%
\providecommand \citenamefont [1]{#1}%
\providecommand \href@noop [0]{\@secondoftwo}%
\providecommand \href [0]{\begingroup \@sanitize@url \@href}%
\providecommand \@href[1]{\@@startlink{#1}\@@href}%
\providecommand \@@href[1]{\endgroup#1\@@endlink}%
\providecommand \@sanitize@url [0]{\catcode `\\12\catcode `\$12\catcode `\&12\catcode `\#12\catcode `\^12\catcode `\_12\catcode `\%12\relax}%
\providecommand \@@startlink[1]{}%
\providecommand \@@endlink[0]{}%
\providecommand \url  [0]{\begingroup\@sanitize@url \@url }%
\providecommand \@url [1]{\endgroup\@href {#1}{\urlprefix }}%
\providecommand \urlprefix  [0]{URL }%
\providecommand \Eprint [0]{\href }%
\providecommand \doibase [0]{https://doi.org/}%
\providecommand \selectlanguage [0]{\@gobble}%
\providecommand \bibinfo  [0]{\@secondoftwo}%
\providecommand \bibfield  [0]{\@secondoftwo}%
\providecommand \translation [1]{[#1]}%
\providecommand \BibitemOpen [0]{}%
\providecommand \bibitemStop [0]{}%
\providecommand \bibitemNoStop [0]{.\EOS\space}%
\providecommand \EOS [0]{\spacefactor3000\relax}%
\providecommand \BibitemShut  [1]{\csname bibitem#1\endcsname}%
\let\auto@bib@innerbib\@empty
\bibitem [{\citenamefont {Cerezo}\ \emph {et~al.}(2021)\citenamefont {Cerezo}, \citenamefont {Arrasmith}, \citenamefont {Babbush}, \citenamefont {Benjamin}, \citenamefont {Endo}, \citenamefont {Fujii}, \citenamefont {McClean}, \citenamefont {Mitarai}, \citenamefont {Yuan}, \citenamefont {Cincio},\ and\ \citenamefont {Coles}}]{Cerezo2021}%
  \BibitemOpen
  \bibfield  {author} {\bibinfo {author} {\bibfnamefont {M.}~\bibnamefont {Cerezo}}, \bibinfo {author} {\bibfnamefont {A.}~\bibnamefont {Arrasmith}}, \bibinfo {author} {\bibfnamefont {R.}~\bibnamefont {Babbush}}, \bibinfo {author} {\bibfnamefont {S.~C.}\ \bibnamefont {Benjamin}}, \bibinfo {author} {\bibfnamefont {S.}~\bibnamefont {Endo}}, \bibinfo {author} {\bibfnamefont {K.}~\bibnamefont {Fujii}}, \bibinfo {author} {\bibfnamefont {J.~R.}\ \bibnamefont {McClean}}, \bibinfo {author} {\bibfnamefont {K.}~\bibnamefont {Mitarai}}, \bibinfo {author} {\bibfnamefont {X.}~\bibnamefont {Yuan}}, \bibinfo {author} {\bibfnamefont {L.}~\bibnamefont {Cincio}},\ and\ \bibinfo {author} {\bibfnamefont {P.~J.}\ \bibnamefont {Coles}},\ }\bibfield  {title} {\bibinfo {title} {\emph{Variational quantum algorithms}},\ }\href {https://doi.org/10.1038/s42254-021-00348-9} {\bibfield  {journal} {\bibinfo  {journal} {Nature Reviews Physics}\ }\textbf {\bibinfo {volume} {3}},\ \bibinfo {pages} {625–644} (\bibinfo {year}
  {2021})}\BibitemShut {NoStop}%
\bibitem [{\citenamefont {Bharti}\ \emph {et~al.}(2022)\citenamefont {Bharti}, \citenamefont {Cervera-Lierta}, \citenamefont {Kyaw}, \citenamefont {Haug}, \citenamefont {Alperin-Lea}, \citenamefont {Anand}, \citenamefont {Degroote}, \citenamefont {Heimonen}, \citenamefont {Kottmann}, \citenamefont {Menke} \emph {et~al.}}]{Bharti2022}%
  \BibitemOpen
  \bibfield  {author} {\bibinfo {author} {\bibfnamefont {K.}~\bibnamefont {Bharti}}, \bibinfo {author} {\bibfnamefont {A.}~\bibnamefont {Cervera-Lierta}}, \bibinfo {author} {\bibfnamefont {T.~H.}\ \bibnamefont {Kyaw}}, \bibinfo {author} {\bibfnamefont {T.}~\bibnamefont {Haug}}, \bibinfo {author} {\bibfnamefont {S.}~\bibnamefont {Alperin-Lea}}, \bibinfo {author} {\bibfnamefont {A.}~\bibnamefont {Anand}}, \bibinfo {author} {\bibfnamefont {M.}~\bibnamefont {Degroote}}, \bibinfo {author} {\bibfnamefont {H.}~\bibnamefont {Heimonen}}, \bibinfo {author} {\bibfnamefont {J.~S.}\ \bibnamefont {Kottmann}}, \bibinfo {author} {\bibfnamefont {T.}~\bibnamefont {Menke}}, \emph {et~al.},\ }\bibfield  {title} {\bibinfo {title} {\emph{Noisy intermediate-scale quantum algorithms}},\ }\href {https://link.aps.org/doi/10.1103/RevModPhys.94.015004} {\bibfield  {journal} {\bibinfo  {journal} {Rev. Mod. Phys.}\ }\textbf {\bibinfo {volume} {94}},\ \bibinfo {pages} {015004} (\bibinfo {year} {2022})}\BibitemShut {NoStop}%
\bibitem [{\citenamefont {Larocca}\ \emph {et~al.}(2025)\citenamefont {Larocca}, \citenamefont {Thanasilp}, \citenamefont {Wang}, \citenamefont {Sharma}, \citenamefont {Biamonte}, \citenamefont {Coles}, \citenamefont {Cincio}, \citenamefont {McClean}, \citenamefont {Holmes},\ and\ \citenamefont {Cerezo}}]{Larocca2025}%
  \BibitemOpen
  \bibfield  {author} {\bibinfo {author} {\bibfnamefont {M.}~\bibnamefont {Larocca}}, \bibinfo {author} {\bibfnamefont {S.}~\bibnamefont {Thanasilp}}, \bibinfo {author} {\bibfnamefont {S.}~\bibnamefont {Wang}}, \bibinfo {author} {\bibfnamefont {K.}~\bibnamefont {Sharma}}, \bibinfo {author} {\bibfnamefont {J.}~\bibnamefont {Biamonte}}, \bibinfo {author} {\bibfnamefont {P.~J.}\ \bibnamefont {Coles}}, \bibinfo {author} {\bibfnamefont {L.}~\bibnamefont {Cincio}}, \bibinfo {author} {\bibfnamefont {J.~R.}\ \bibnamefont {McClean}}, \bibinfo {author} {\bibfnamefont {Z.}~\bibnamefont {Holmes}},\ and\ \bibinfo {author} {\bibfnamefont {M.}~\bibnamefont {Cerezo}},\ }\bibfield  {title} {\bibinfo {title} {\emph{Barren plateaus in variational quantum computing}},\ }\href {https://doi.org/10.1038/s42254-025-00813-9} {\bibfield  {journal} {\bibinfo  {journal} {Nature Reviews Physics}\ }\textbf {\bibinfo {volume} {7}},\ \bibinfo {pages} {174–189} (\bibinfo {year} {2025})}\BibitemShut {NoStop}%
\bibitem [{\citenamefont {Stokes}\ \emph {et~al.}(2020)\citenamefont {Stokes}, \citenamefont {Izaac}, \citenamefont {Killoran},\ and\ \citenamefont {Carleo}}]{Stokes2020}%
  \BibitemOpen
  \bibfield  {author} {\bibinfo {author} {\bibfnamefont {J.}~\bibnamefont {Stokes}}, \bibinfo {author} {\bibfnamefont {J.}~\bibnamefont {Izaac}}, \bibinfo {author} {\bibfnamefont {N.}~\bibnamefont {Killoran}},\ and\ \bibinfo {author} {\bibfnamefont {G.}~\bibnamefont {Carleo}},\ }\bibfield  {title} {\bibinfo {title} {\emph{Quantum natural gradient}},\ }\href {https://doi.org/10.22331/q-2020-05-25-269} {\bibfield  {journal} {\bibinfo  {journal} {{Quantum}}\ }\textbf {\bibinfo {volume} {4}},\ \bibinfo {pages} {269} (\bibinfo {year} {2020})}\BibitemShut {NoStop}%
\bibitem [{\citenamefont {Koczor}\ and\ \citenamefont {Benjamin}(2022)}]{Koczor2022}%
  \BibitemOpen
  \bibfield  {author} {\bibinfo {author} {\bibfnamefont {B.}~\bibnamefont {Koczor}}\ and\ \bibinfo {author} {\bibfnamefont {S.~C.}\ \bibnamefont {Benjamin}},\ }\bibfield  {title} {\bibinfo {title} {\emph{Quantum natural gradient generalized to noisy and nonunitary circuits}},\ }\href {https://link.aps.org/doi/10.1103/PhysRevA.106.062416} {\bibfield  {journal} {\bibinfo  {journal} {Phys. Rev. A}\ }\textbf {\bibinfo {volume} {106}},\ \bibinfo {pages} {062416} (\bibinfo {year} {2022})}\BibitemShut {NoStop}%
\bibitem [{\citenamefont {Absil}\ \emph {et~al.}(2008)\citenamefont {Absil}, \citenamefont {Mahony},\ and\ \citenamefont {Sepulchre}}]{absil2008}%
  \BibitemOpen
  \bibfield  {author} {\bibinfo {author} {\bibfnamefont {P.-A.}\ \bibnamefont {Absil}}, \bibinfo {author} {\bibfnamefont {R.}~\bibnamefont {Mahony}},\ and\ \bibinfo {author} {\bibfnamefont {R.}~\bibnamefont {Sepulchre}},\ }\href {http://www.jstor.org/stable/j.ctt7smmk} {\emph {\bibinfo {title} {Optimization Algorithms on Matrix Manifolds}}}\ (\bibinfo  {publisher} {Princeton University Press},\ \bibinfo {year} {2008})\BibitemShut {NoStop}%
\bibitem [{\citenamefont {Sato}(2021)}]{sato2021}%
  \BibitemOpen
  \bibfield  {author} {\bibinfo {author} {\bibfnamefont {H.}~\bibnamefont {Sato}},\ }\href {https://link.springer.com/book/10.1007/978-3-030-62391-3} {\emph {\bibinfo {title} {Riemannian optimization and its applications}}},\ Vol.\ \bibinfo {volume} {670}\ (\bibinfo  {publisher} {Springer},\ \bibinfo {year} {2021})\BibitemShut {NoStop}%
\bibitem [{\citenamefont {Halla}(2025{\natexlab{a}})}]{Halla2025secondorder}%
  \BibitemOpen
  \bibfield  {author} {\bibinfo {author} {\bibfnamefont {M.}~\bibnamefont {Halla}},\ }\bibfield  {title} {\bibinfo {title} {\emph{Quantum natural gradient with geodesic corrections for small shallow quantum circuits}},\ }\href {https://doi.org/10.1088/1402-4896/add05e} {\bibfield  {journal} {\bibinfo  {journal} {Physica Scripta}\ }\textbf {\bibinfo {volume} {100}},\ \bibinfo {pages} {055121} (\bibinfo {year} {2025}{\natexlab{a}})}\BibitemShut {NoStop}%
\bibitem [{\citenamefont {Halla}(2025{\natexlab{b}})}]{Halla2025conjugate}%
  \BibitemOpen
  \bibfield  {author} {\bibinfo {author} {\bibfnamefont {M.}~\bibnamefont {Halla}},\ }\href {https://arxiv.org/abs/2501.05847} {\bibinfo {title} {\emph{Modified conjugate quantum natural gradient}}} (\bibinfo {year} {2025}{\natexlab{b}}),\ \Eprint {https://arxiv.org/abs/2501.05847} {arXiv:2501.05847 [quant-ph]} \BibitemShut {NoStop}%
\bibitem [{\citenamefont {Haug}\ and\ \citenamefont {Kim}(2022)}]{Haug2022}%
  \BibitemOpen
  \bibfield  {author} {\bibinfo {author} {\bibfnamefont {T.}~\bibnamefont {Haug}}\ and\ \bibinfo {author} {\bibfnamefont {M.~S.}\ \bibnamefont {Kim}},\ }\bibfield  {title} {\bibinfo {title} {\emph{Natural parametrized quantum circuit}},\ }\href {http://dx.doi.org/10.1103/PhysRevA.106.052611} {\bibfield  {journal} {\bibinfo  {journal} {Physical Review A}\ }\textbf {\bibinfo {volume} {106}} (\bibinfo {year} {2022})}\BibitemShut {NoStop}%
\bibitem [{\citenamefont {Yao}\ \emph {et~al.}(2022)\citenamefont {Yao}, \citenamefont {Cussenot}, \citenamefont {Wolf},\ and\ \citenamefont {Miatto}}]{Yao2022}%
  \BibitemOpen
  \bibfield  {author} {\bibinfo {author} {\bibfnamefont {Y.}~\bibnamefont {Yao}}, \bibinfo {author} {\bibfnamefont {P.}~\bibnamefont {Cussenot}}, \bibinfo {author} {\bibfnamefont {R.~A.}\ \bibnamefont {Wolf}},\ and\ \bibinfo {author} {\bibfnamefont {F.}~\bibnamefont {Miatto}},\ }\bibfield  {title} {\bibinfo {title} {\emph{Complex natural gradient optimization for optical quantum circuit design}},\ }\href {http://dx.doi.org/10.1103/PhysRevA.105.052402} {\bibfield  {journal} {\bibinfo  {journal} {Physical Review A}\ }\textbf {\bibinfo {volume} {105}} (\bibinfo {year} {2022})}\BibitemShut {NoStop}%
\bibitem [{\citenamefont {Meyer}\ \emph {et~al.}(2023)\citenamefont {Meyer}, \citenamefont {Scherer}, \citenamefont {Plinge}, \citenamefont {Mutschler},\ and\ \citenamefont {Hartmann}}]{Meyer2023}%
  \BibitemOpen
  \bibfield  {author} {\bibinfo {author} {\bibfnamefont {N.}~\bibnamefont {Meyer}}, \bibinfo {author} {\bibfnamefont {D.~D.}\ \bibnamefont {Scherer}}, \bibinfo {author} {\bibfnamefont {A.}~\bibnamefont {Plinge}}, \bibinfo {author} {\bibfnamefont {C.}~\bibnamefont {Mutschler}},\ and\ \bibinfo {author} {\bibfnamefont {M.~J.}\ \bibnamefont {Hartmann}},\ }\bibfield  {title} {\bibinfo {title} {\emph{{Q}uantum natural policy gradients: {T}owards sample-efficient reinforcement learning}},\ }in\ \href {https://doi.org/10.1109/qce57702.2023.10181} {\emph {\bibinfo {booktitle} {2023 IEEE International Conference on Quantum Computing and Engineering (QCE)}}}\ (\bibinfo  {publisher} {IEEE},\ \bibinfo {year} {2023})\ p.\ \bibinfo {pages} {36–41}\BibitemShut {NoStop}%
\bibitem [{\citenamefont {Qi}\ and\ \citenamefont {Hsieh}(2024)}]{Qi2024}%
  \BibitemOpen
  \bibfield  {author} {\bibinfo {author} {\bibfnamefont {J.}~\bibnamefont {Qi}}\ and\ \bibinfo {author} {\bibfnamefont {M.-H.}\ \bibnamefont {Hsieh}},\ }\bibfield  {title} {\bibinfo {title} {\emph{Federated quantum natural gradient descent for quantum federated learning}},\ }in\ \href {https://doi.org/10.1016/B978-0-44-319037-7.00028-4} {\emph {\bibinfo {booktitle} {Federated Learning}}}\ (\bibinfo  {publisher} {Academic Press},\ \bibinfo {year} {2024})\ pp.\ \bibinfo {pages} {329--341}\BibitemShut {NoStop}%
\bibitem [{\citenamefont {Fitzek}\ \emph {et~al.}(2024)\citenamefont {Fitzek}, \citenamefont {Jonsson}, \citenamefont {Dobrautz},\ and\ \citenamefont {Sch{\"{a}}fer}}]{Fitzek2024}%
  \BibitemOpen
  \bibfield  {author} {\bibinfo {author} {\bibfnamefont {D.}~\bibnamefont {Fitzek}}, \bibinfo {author} {\bibfnamefont {R.~S.}\ \bibnamefont {Jonsson}}, \bibinfo {author} {\bibfnamefont {W.}~\bibnamefont {Dobrautz}},\ and\ \bibinfo {author} {\bibfnamefont {C.}~\bibnamefont {Sch{\"{a}}fer}},\ }\bibfield  {title} {\bibinfo {title} {\emph{Optimizing variational quantum algorithms with q{B}ang: {E}fficiently interweaving metric and momentum to navigate flat energy landscapes}},\ }\href {https://doi.org/10.22331/q-2024-04-09-1313} {\bibfield  {journal} {\bibinfo  {journal} {Quantum}\ }\textbf {\bibinfo {volume} {8}},\ \bibinfo {pages} {1313} (\bibinfo {year} {2024})}\BibitemShut {NoStop}%
\bibitem [{\citenamefont {Dell’Anna}\ \emph {et~al.}(2025)\citenamefont {Dell’Anna}, \citenamefont {Gómez-Lurbe}, \citenamefont {Pérez},\ and\ \citenamefont {Ercolessi}}]{DellAnna2025}%
  \BibitemOpen
  \bibfield  {author} {\bibinfo {author} {\bibfnamefont {F.}~\bibnamefont {Dell’Anna}}, \bibinfo {author} {\bibfnamefont {R.}~\bibnamefont {Gómez-Lurbe}}, \bibinfo {author} {\bibfnamefont {A.}~\bibnamefont {Pérez}},\ and\ \bibinfo {author} {\bibfnamefont {E.}~\bibnamefont {Ercolessi}},\ }\bibfield  {title} {\bibinfo {title} {\emph{Quantum natural gradient optimizer on noisy platforms: {Q}uantum approximate optimization algorithm as a case study}},\ }\bibfield  {journal} {\bibinfo  {journal} {Physical Review A}\ }\textbf {\bibinfo {volume} {112}},\ \href {https://doi.org/10.1103/62wx-tvk5} {10.1103/62wx-tvk5} (\bibinfo {year} {2025})\BibitemShut {NoStop}%
\bibitem [{\citenamefont {G\'{o}mez-Lurbe}(2025)}]{GomezLurbe2025}%
  \BibitemOpen
  \bibfield  {author} {\bibinfo {author} {\bibfnamefont {R.}~\bibnamefont {G\'{o}mez-Lurbe}},\ }\href {https://arxiv.org/abs/2505.09818} {\bibinfo {title} {\emph{Efficient protocol to estimate the quantum {F}isher information matrix for commuting-block circuits}}} (\bibinfo {year} {2025}),\ \Eprint {https://arxiv.org/abs/2505.09818} {arXiv:2505.09818 [quant-ph]} \BibitemShut {NoStop}%
\bibitem [{\citenamefont {Wierichs}\ \emph {et~al.}(2020)\citenamefont {Wierichs}, \citenamefont {Gogolin},\ and\ \citenamefont {Kastoryano}}]{Wierichs2020}%
  \BibitemOpen
  \bibfield  {author} {\bibinfo {author} {\bibfnamefont {D.}~\bibnamefont {Wierichs}}, \bibinfo {author} {\bibfnamefont {C.}~\bibnamefont {Gogolin}},\ and\ \bibinfo {author} {\bibfnamefont {M.}~\bibnamefont {Kastoryano}},\ }\bibfield  {title} {\bibinfo {title} {\emph{Avoiding local minima in variational quantum eigensolvers with the natural gradient optimizer}},\ }\href {http://dx.doi.org/10.1103/PhysRevResearch.2.043246} {\bibfield  {journal} {\bibinfo  {journal} {Physical Review Research}\ }\textbf {\bibinfo {volume} {2}} (\bibinfo {year} {2020})}\BibitemShut {NoStop}%
\bibitem [{\citenamefont {van Straaten}\ and\ \citenamefont {Koczor}(2021)}]{Straaten2021}%
  \BibitemOpen
  \bibfield  {author} {\bibinfo {author} {\bibfnamefont {B.}~\bibnamefont {van Straaten}}\ and\ \bibinfo {author} {\bibfnamefont {B.}~\bibnamefont {Koczor}},\ }\bibfield  {title} {\bibinfo {title} {\emph{Measurement cost of metric-aware variational quantum algorithms}},\ }\href {http://dx.doi.org/10.1103/PRXQuantum.2.030324} {\bibfield  {journal} {\bibinfo  {journal} {PRX Quantum}\ }\textbf {\bibinfo {volume} {2}} (\bibinfo {year} {2021})}\BibitemShut {NoStop}%
\bibitem [{\citenamefont {Gacon}\ \emph {et~al.}(2021)\citenamefont {Gacon}, \citenamefont {Zoufal}, \citenamefont {Carleo},\ and\ \citenamefont {Woerner}}]{Gacon2021}%
  \BibitemOpen
  \bibfield  {author} {\bibinfo {author} {\bibfnamefont {J.}~\bibnamefont {Gacon}}, \bibinfo {author} {\bibfnamefont {C.}~\bibnamefont {Zoufal}}, \bibinfo {author} {\bibfnamefont {G.}~\bibnamefont {Carleo}},\ and\ \bibinfo {author} {\bibfnamefont {S.}~\bibnamefont {Woerner}},\ }\bibfield  {title} {\bibinfo {title} {\emph{Simultaneous perturbation stochastic approximation of the quantum {F}isher information}},\ }\href {https://doi.org/10.22331/q-2021-10-20-567} {\bibfield  {journal} {\bibinfo  {journal} {{Quantum}}\ }\textbf {\bibinfo {volume} {5}},\ \bibinfo {pages} {567} (\bibinfo {year} {2021})}\BibitemShut {NoStop}%
\bibitem [{\citenamefont {Gacon}\ \emph {et~al.}(2023)\citenamefont {Gacon}, \citenamefont {Zoufal}, \citenamefont {Carleo},\ and\ \citenamefont {Woerner}}]{Gacon2023}%
  \BibitemOpen
  \bibfield  {author} {\bibinfo {author} {\bibfnamefont {J.}~\bibnamefont {Gacon}}, \bibinfo {author} {\bibfnamefont {C.}~\bibnamefont {Zoufal}}, \bibinfo {author} {\bibfnamefont {G.}~\bibnamefont {Carleo}},\ and\ \bibinfo {author} {\bibfnamefont {S.}~\bibnamefont {Woerner}},\ }\bibfield  {title} {\bibinfo {title} {\emph{Stochastic approximation of variational quantum imaginary time evolution}},\ }in\ \href {http://dx.doi.org/10.1109/QCE57702.2023.10367741} {\emph {\bibinfo {booktitle} {2023 IEEE International Conference on Quantum Computing and Engineering (QCE)}}}\ (\bibinfo  {publisher} {IEEE},\ \bibinfo {year} {2023})\ p.\ \bibinfo {pages} {129–139}\BibitemShut {NoStop}%
\bibitem [{\citenamefont {Halla}(2025{\natexlab{c}})}]{Halla2025fisher}%
  \BibitemOpen
  \bibfield  {author} {\bibinfo {author} {\bibfnamefont {M.}~\bibnamefont {Halla}},\ }\href {https://arxiv.org/abs/2502.17231} {\bibinfo {title} {\emph{Estimation of quantum {F}isher information via {S}tein's identity in variational quantum algorithms}}} (\bibinfo {year} {2025}{\natexlab{c}}),\ \Eprint {https://arxiv.org/abs/2502.17231} {arXiv:2502.17231 [quant-ph]} \BibitemShut {NoStop}%
\bibitem [{\citenamefont {Schuld}\ \emph {et~al.}(2019)\citenamefont {Schuld}, \citenamefont {Bergholm}, \citenamefont {Gogolin}, \citenamefont {Izaac},\ and\ \citenamefont {Killoran}}]{Schuld2019}%
  \BibitemOpen
  \bibfield  {author} {\bibinfo {author} {\bibfnamefont {M.}~\bibnamefont {Schuld}}, \bibinfo {author} {\bibfnamefont {V.}~\bibnamefont {Bergholm}}, \bibinfo {author} {\bibfnamefont {C.}~\bibnamefont {Gogolin}}, \bibinfo {author} {\bibfnamefont {J.}~\bibnamefont {Izaac}},\ and\ \bibinfo {author} {\bibfnamefont {N.}~\bibnamefont {Killoran}},\ }\bibfield  {title} {\bibinfo {title} {\emph{Evaluating analytic gradients on quantum hardware}},\ }\href {https://doi.org/10.1103/PhysRevA.99.032331} {\bibfield  {journal} {\bibinfo  {journal} {Phys. Rev. A}\ }\textbf {\bibinfo {volume} {99}},\ \bibinfo {pages} {032331} (\bibinfo {year} {2019})}\BibitemShut {NoStop}%
\bibitem [{\citenamefont {Anselmetti}\ \emph {et~al.}(2021)\citenamefont {Anselmetti}, \citenamefont {Wierichs}, \citenamefont {Gogolin},\ and\ \citenamefont {Parrish}}]{Anselmetti2021}%
  \BibitemOpen
  \bibfield  {author} {\bibinfo {author} {\bibfnamefont {G.-L.~R.}\ \bibnamefont {Anselmetti}}, \bibinfo {author} {\bibfnamefont {D.}~\bibnamefont {Wierichs}}, \bibinfo {author} {\bibfnamefont {C.}~\bibnamefont {Gogolin}},\ and\ \bibinfo {author} {\bibfnamefont {R.~M.}\ \bibnamefont {Parrish}},\ }\bibfield  {title} {\bibinfo {title} {\emph{Local, expressive, quantum-number-preserving {VQE} ans\"{a}tze for fermionic systems}},\ }\href {https://doi.org/10.1088/1367-2630/ac2cb3} {\bibfield  {journal} {\bibinfo  {journal} {New Journal of Physics}\ }\textbf {\bibinfo {volume} {23}},\ \bibinfo {pages} {113010} (\bibinfo {year} {2021})}\BibitemShut {NoStop}%
\bibitem [{\citenamefont {Kottmann}\ \emph {et~al.}(2021)\citenamefont {Kottmann}, \citenamefont {Anand},\ and\ \citenamefont {Aspuru-Guzik}}]{Kottmann2021}%
  \BibitemOpen
  \bibfield  {author} {\bibinfo {author} {\bibfnamefont {J.~S.}\ \bibnamefont {Kottmann}}, \bibinfo {author} {\bibfnamefont {A.}~\bibnamefont {Anand}},\ and\ \bibinfo {author} {\bibfnamefont {A.}~\bibnamefont {Aspuru-Guzik}},\ }\bibfield  {title} {\bibinfo {title} {\emph{A feasible approach for automatically differentiable unitary coupled-cluster on quantum computers}},\ }\href {https://doi.org/10.1039/D0SC06627C} {\bibfield  {journal} {\bibinfo  {journal} {Chem. Sci.}\ }\textbf {\bibinfo {volume} {12}},\ \bibinfo {pages} {3497} (\bibinfo {year} {2021})}\BibitemShut {NoStop}%
\bibitem [{\citenamefont {{IonQ Inc.}}(2025{\natexlab{a}})}]{IonQ}%
  \BibitemOpen
  \bibfield  {author} {\bibinfo {author} {\bibnamefont {{IonQ Inc.}}},\ }\href {https://ionq.com/quantum-systems/forte} {\bibinfo {title} {\emph{{IonQ} {F}orte}}} (\bibinfo {year} {2025}{\natexlab{a}})\BibitemShut {NoStop}%
\bibitem [{\citenamefont {Liu}\ \emph {et~al.}(2020)\citenamefont {Liu}, \citenamefont {Yuan}, \citenamefont {Lu},\ and\ \citenamefont {Wang}}]{Liu2020quantum}%
  \BibitemOpen
  \bibfield  {author} {\bibinfo {author} {\bibfnamefont {J.}~\bibnamefont {Liu}}, \bibinfo {author} {\bibfnamefont {H.}~\bibnamefont {Yuan}}, \bibinfo {author} {\bibfnamefont {X.-M.}\ \bibnamefont {Lu}},\ and\ \bibinfo {author} {\bibfnamefont {X.}~\bibnamefont {Wang}},\ }\bibfield  {title} {\bibinfo {title} {\emph{Quantum {F}isher information matrix and multiparameter estimation}},\ }\href@noop {} {\bibfield  {journal} {\bibinfo  {journal} {Journal of Physics A: Mathematical and Theoretical}\ }\textbf {\bibinfo {volume} {53}},\ \bibinfo {pages} {023001} (\bibinfo {year} {2020})}\BibitemShut {NoStop}%
\bibitem [{\citenamefont {Meyer}(2021)}]{Meyer2021}%
  \BibitemOpen
  \bibfield  {author} {\bibinfo {author} {\bibfnamefont {J.~J.}\ \bibnamefont {Meyer}},\ }\bibfield  {title} {\bibinfo {title} {\emph{Fisher information in noisy intermediate-scale quantum applications}},\ }\href {http://dx.doi.org/10.22331/q-2021-09-09-539} {\bibfield  {journal} {\bibinfo  {journal} {Quantum}\ }\textbf {\bibinfo {volume} {5}},\ \bibinfo {pages} {539} (\bibinfo {year} {2021})}\BibitemShut {NoStop}%
\bibitem [{\citenamefont {Farias}\ \emph {et~al.}(2025{\natexlab{a}})\citenamefont {Farias}, \citenamefont {Maciel}, \citenamefont {Camilo}, \citenamefont {Lin}, \citenamefont {Ramos-Calderer},\ and\ \citenamefont {Aolita}}]{Farias2025}%
  \BibitemOpen
  \bibfield  {author} {\bibinfo {author} {\bibfnamefont {R.~M.~S.}\ \bibnamefont {Farias}}, \bibinfo {author} {\bibfnamefont {T.~O.}\ \bibnamefont {Maciel}}, \bibinfo {author} {\bibfnamefont {G.}~\bibnamefont {Camilo}}, \bibinfo {author} {\bibfnamefont {R.}~\bibnamefont {Lin}}, \bibinfo {author} {\bibfnamefont {S.}~\bibnamefont {Ramos-Calderer}},\ and\ \bibinfo {author} {\bibfnamefont {L.}~\bibnamefont {Aolita}},\ }\bibfield  {title} {\bibinfo {title} {\emph{Quantum encoder for fixed-Hamming-weight subspaces}},\ }\href {https://link.aps.org/doi/10.1103/PhysRevApplied.23.044014} {\bibfield  {journal} {\bibinfo  {journal} {Phys. Rev. Appl.}\ }\textbf {\bibinfo {volume} {23}},\ \bibinfo {pages} {044014} (\bibinfo {year} {2025}{\natexlab{a}})}\BibitemShut {NoStop}%
\bibitem [{\citenamefont {Larocca}\ \emph {et~al.}(2023)\citenamefont {Larocca}, \citenamefont {Ju}, \citenamefont {García-Martín}, \citenamefont {Coles},\ and\ \citenamefont {Cerezo}}]{Larocca2023}%
  \BibitemOpen
  \bibfield  {author} {\bibinfo {author} {\bibfnamefont {M.}~\bibnamefont {Larocca}}, \bibinfo {author} {\bibfnamefont {N.}~\bibnamefont {Ju}}, \bibinfo {author} {\bibfnamefont {D.}~\bibnamefont {García-Martín}}, \bibinfo {author} {\bibfnamefont {P.~J.}\ \bibnamefont {Coles}},\ and\ \bibinfo {author} {\bibfnamefont {M.}~\bibnamefont {Cerezo}},\ }\bibfield  {title} {\bibinfo {title} {\emph{Theory of overparametrization in quantum neural networks}},\ }\href {https://doi.org/10.1038/s43588-023-00467-6} {\bibfield  {journal} {\bibinfo  {journal} {Nature Computational Science}\ }\textbf {\bibinfo {volume} {3}},\ \bibinfo {pages} {542–551} (\bibinfo {year} {2023})}\BibitemShut {NoStop}%
\bibitem [{\citenamefont {Haug}\ \emph {et~al.}(2021)\citenamefont {Haug}, \citenamefont {Bharti},\ and\ \citenamefont {Kim}}]{Haug2021capacity}%
  \BibitemOpen
  \bibfield  {author} {\bibinfo {author} {\bibfnamefont {T.}~\bibnamefont {Haug}}, \bibinfo {author} {\bibfnamefont {K.}~\bibnamefont {Bharti}},\ and\ \bibinfo {author} {\bibfnamefont {M.}~\bibnamefont {Kim}},\ }\bibfield  {title} {\bibinfo {title} {\emph{Capacity and quantum geometry of parametrized quantum circuits}},\ }\href {https://doi.org/10.1103/PRXQuantum.2.040309} {\bibfield  {journal} {\bibinfo  {journal} {PRX Quantum}\ }\textbf {\bibinfo {volume} {2}},\ \bibinfo {pages} {040309} (\bibinfo {year} {2021})}\BibitemShut {NoStop}%
\bibitem [{\citenamefont {Haug}\ and\ \citenamefont {Kim}(2024)}]{Haug2024}%
  \BibitemOpen
  \bibfield  {author} {\bibinfo {author} {\bibfnamefont {T.}~\bibnamefont {Haug}}\ and\ \bibinfo {author} {\bibfnamefont {M.}~\bibnamefont {Kim}},\ }\bibfield  {title} {\bibinfo {title} {\emph{Generalization of quantum machine learning models using quantum {F}isher information metric}},\ }\href {https://dx.doi.org/10.1103/physrevlett.133.050603} {\bibfield  {journal} {\bibinfo  {journal} {Physical Review Letters}\ }\textbf {\bibinfo {volume} {133}},\ \bibinfo {pages} {050603} (\bibinfo {year} {2024})}\BibitemShut {NoStop}%
\bibitem [{\citenamefont {Sato}(2022)}]{sato2022}%
  \BibitemOpen
  \bibfield  {author} {\bibinfo {author} {\bibfnamefont {H.}~\bibnamefont {Sato}},\ }\bibfield  {title} {\bibinfo {title} {\emph{Riemannian conjugate gradient methods: {G}eneral framework and specific algorithms with convergence analyses}},\ }\href {http://dx.doi.org/10.1137/21M1464178} {\bibfield  {journal} {\bibinfo  {journal} {SIAM Journal on Optimization}\ }\textbf {\bibinfo {volume} {32}},\ \bibinfo {pages} {2690–2717} (\bibinfo {year} {2022})}\BibitemShut {NoStop}%
\bibitem [{\citenamefont {Wolfe}(1969)}]{wolfe1969}%
  \BibitemOpen
  \bibfield  {author} {\bibinfo {author} {\bibfnamefont {P.}~\bibnamefont {Wolfe}},\ }\bibfield  {title} {\bibinfo {title} {\emph{Convergence conditions for ascent methods}},\ }\href {http://www.jstor.org/stable/2028111} {\bibfield  {journal} {\bibinfo  {journal} {SIAM Review}\ }\textbf {\bibinfo {volume} {11}},\ \bibinfo {pages} {226} (\bibinfo {year} {1969})}\BibitemShut {NoStop}%
\bibitem [{\citenamefont {Wolfe}(1971)}]{wolfe1971}%
  \BibitemOpen
  \bibfield  {author} {\bibinfo {author} {\bibfnamefont {P.}~\bibnamefont {Wolfe}},\ }\bibfield  {title} {\bibinfo {title} {\emph{Convergence conditions for ascent methods. {II}: {S}ome Corrections}},\ }\href {http://www.jstor.org/stable/2028821} {\bibfield  {journal} {\bibinfo  {journal} {SIAM Review}\ }\textbf {\bibinfo {volume} {13}},\ \bibinfo {pages} {185} (\bibinfo {year} {1971})}\BibitemShut {NoStop}%
\bibitem [{\citenamefont {Mo{\v{c}}kus}(1975)}]{Movckus1975}%
  \BibitemOpen
  \bibfield  {author} {\bibinfo {author} {\bibfnamefont {J.}~\bibnamefont {Mo{\v{c}}kus}},\ }\bibinfo {title} {\emph{On {B}ayesian methods for seeking the extremum}},\ in\ \href {https://doi.org/10.1007/978-3-662-38527-2_55} {\emph {\bibinfo {booktitle} {Optimization Techniques IFIP Technical Conference: Novosibirsk, July 1--7, 1974}}}\ (\bibinfo  {publisher} {Springer Berlin Heidelber},\ \bibinfo {address} {Berlin, Heidelberg},\ \bibinfo {year} {1975})\ pp.\ \bibinfo {pages} {400--404}\BibitemShut {NoStop}%
\bibitem [{\citenamefont {Wang}\ \emph {et~al.}(2023)\citenamefont {Wang}, \citenamefont {Jin}, \citenamefont {Schmitt},\ and\ \citenamefont {Olhofer}}]{Wang2023}%
  \BibitemOpen
  \bibfield  {author} {\bibinfo {author} {\bibfnamefont {X.}~\bibnamefont {Wang}}, \bibinfo {author} {\bibfnamefont {Y.}~\bibnamefont {Jin}}, \bibinfo {author} {\bibfnamefont {S.}~\bibnamefont {Schmitt}},\ and\ \bibinfo {author} {\bibfnamefont {M.}~\bibnamefont {Olhofer}},\ }\bibfield  {title} {\bibinfo {title} {\emph{Recent advances in {B}ayesian optimization}},\ }\href {https://doi.org/10.1145/3582078} {\bibfield  {journal} {\bibinfo  {journal} {ACM Comput. Surv.}\ }\textbf {\bibinfo {volume} {55}} (\bibinfo {year} {2023})}\BibitemShut {NoStop}%
\bibitem [{\citenamefont {Jordan}\ and\ \citenamefont {Wigner}(1928)}]{Jordan1928}%
  \BibitemOpen
  \bibfield  {author} {\bibinfo {author} {\bibfnamefont {P.}~\bibnamefont {Jordan}}\ and\ \bibinfo {author} {\bibfnamefont {E.}~\bibnamefont {Wigner}},\ }\bibfield  {title} {\bibinfo {title} {\emph{{\"U}ber das Paulische {\"A}quivalenzverbot}},\ }\href {https://doi.org/10.1007/BF01331938} {\bibfield  {journal} {\bibinfo  {journal} {Zeitschrift f{\"u}r Physik}\ }\textbf {\bibinfo {volume} {47}},\ \bibinfo {pages} {631} (\bibinfo {year} {1928})}\BibitemShut {NoStop}%
\bibitem [{\citenamefont {Tranter}\ \emph {et~al.}(2018)\citenamefont {Tranter}, \citenamefont {Love}, \citenamefont {Mintert},\ and\ \citenamefont {Coveney}}]{Tranter2018}%
  \BibitemOpen
  \bibfield  {author} {\bibinfo {author} {\bibfnamefont {A.}~\bibnamefont {Tranter}}, \bibinfo {author} {\bibfnamefont {P.~J.}\ \bibnamefont {Love}}, \bibinfo {author} {\bibfnamefont {F.}~\bibnamefont {Mintert}},\ and\ \bibinfo {author} {\bibfnamefont {P.~V.}\ \bibnamefont {Coveney}},\ }\bibfield  {title} {\bibinfo {title} {\emph{A comparison of the {B}ravyi–{K}itaev and {J}ordan–{W}igner transformations for the quantum simulation of quantum chemistry}},\ }\href {https://doi.org/10.1021/acs.jctc.8b00450} {\bibfield  {journal} {\bibinfo  {journal} {Journal of {C}hemical {T}heory and {C}omputation}\ }\textbf {\bibinfo {volume} {14}},\ \bibinfo {pages} {5617–5630} (\bibinfo {year} {2018})}\BibitemShut {NoStop}%
\bibitem [{\citenamefont {Tilly}\ \emph {et~al.}(2022)\citenamefont {Tilly}, \citenamefont {Chen}, \citenamefont {Cao}, \citenamefont {Picozzi}, \citenamefont {Setia}, \citenamefont {Li}, \citenamefont {Grant}, \citenamefont {Wossnig}, \citenamefont {Rungger}, \citenamefont {Booth},\ and\ \citenamefont {Tennyson}}]{Tilly2022}%
  \BibitemOpen
  \bibfield  {author} {\bibinfo {author} {\bibfnamefont {J.}~\bibnamefont {Tilly}}, \bibinfo {author} {\bibfnamefont {H.}~\bibnamefont {Chen}}, \bibinfo {author} {\bibfnamefont {S.}~\bibnamefont {Cao}}, \bibinfo {author} {\bibfnamefont {D.}~\bibnamefont {Picozzi}}, \bibinfo {author} {\bibfnamefont {K.}~\bibnamefont {Setia}}, \bibinfo {author} {\bibfnamefont {Y.}~\bibnamefont {Li}}, \bibinfo {author} {\bibfnamefont {E.}~\bibnamefont {Grant}}, \bibinfo {author} {\bibfnamefont {L.}~\bibnamefont {Wossnig}}, \bibinfo {author} {\bibfnamefont {I.}~\bibnamefont {Rungger}}, \bibinfo {author} {\bibfnamefont {G.~H.}\ \bibnamefont {Booth}},\ and\ \bibinfo {author} {\bibfnamefont {J.}~\bibnamefont {Tennyson}},\ }\bibfield  {title} {\bibinfo {title} {\emph{The variational quantum eigensolver: {A} review of methods and best practices}},\ }\href {https://doi.org/10.1016/j.physrep.2022.08.003} {\bibfield  {journal} {\bibinfo  {journal} {Physics Reports}\ }\textbf {\bibinfo {volume} {986}},\ \bibinfo {pages} {1–128}
  (\bibinfo {year} {2022})}\BibitemShut {NoStop}%
\bibitem [{\citenamefont {Szabo}\ and\ \citenamefont {Ostlund}(1996)}]{Szabo1996}%
  \BibitemOpen
  \bibfield  {author} {\bibinfo {author} {\bibfnamefont {A.}~\bibnamefont {Szabo}}\ and\ \bibinfo {author} {\bibfnamefont {N.~S.}\ \bibnamefont {Ostlund}},\ }\href@noop {} {\emph {\bibinfo {title} {{M}odern quantum chemistry: {I}ntroduction to advanced electronic structure theory}}}\ (\bibinfo  {publisher} {Dover Publications},\ \bibinfo {address} {Mineola},\ \bibinfo {year} {1996})\BibitemShut {NoStop}%
\bibitem [{\citenamefont {Gomez}\ \emph {et~al.}(2011)\citenamefont {Gomez}, \citenamefont {Sierra},\ and\ \citenamefont {Ruiz-Altaba}}]{Gomez1996}%
  \BibitemOpen
  \bibfield  {author} {\bibinfo {author} {\bibfnamefont {C.}~\bibnamefont {Gomez}}, \bibinfo {author} {\bibfnamefont {G.}~\bibnamefont {Sierra}},\ and\ \bibinfo {author} {\bibfnamefont {M.}~\bibnamefont {Ruiz-Altaba}},\ }\href {https://doi.org/10.1017/CBO9780511628825} {\emph {\bibinfo {title} {Quantum groups in two-dimensional physics}}},\ Cambridge Monographs on Mathematical Physics\ (\bibinfo  {publisher} {Cambridge University Press},\ \bibinfo {year} {2011})\BibitemShut {NoStop}%
\bibitem [{\citenamefont {{IonQ Inc.}}(2025{\natexlab{b}})}]{debiasing}%
  \BibitemOpen
  \bibfield  {author} {\bibinfo {author} {\bibnamefont {{IonQ Inc.}}},\ }\href {https://www.ionq.com/resources/debiasing-and-sharpening} {\bibinfo {title} {\emph{Debiasing and sharpening}}} (\bibinfo {year} {2025}{\natexlab{b}})\BibitemShut {NoStop}%
\bibitem [{\citenamefont {Sciorilli}\ \emph {et~al.}(2025)\citenamefont {Sciorilli}, \citenamefont {Borges}, \citenamefont {Patti}, \citenamefont {García-Martín}, \citenamefont {Camilo}, \citenamefont {Anandkumar},\ and\ \citenamefont {Aolita}}]{Sciorilli2025}%
  \BibitemOpen
  \bibfield  {author} {\bibinfo {author} {\bibfnamefont {M.}~\bibnamefont {Sciorilli}}, \bibinfo {author} {\bibfnamefont {L.}~\bibnamefont {Borges}}, \bibinfo {author} {\bibfnamefont {T.~L.}\ \bibnamefont {Patti}}, \bibinfo {author} {\bibfnamefont {D.}~\bibnamefont {García-Martín}}, \bibinfo {author} {\bibfnamefont {G.}~\bibnamefont {Camilo}}, \bibinfo {author} {\bibfnamefont {A.}~\bibnamefont {Anandkumar}},\ and\ \bibinfo {author} {\bibfnamefont {L.}~\bibnamefont {Aolita}},\ }\bibfield  {title} {\bibinfo {title} {\emph{Towards large-scale quantum optimization solvers with few qubits}},\ }\href {http://dx.doi.org/10.1038/s41467-024-55346-z} {\bibfield  {journal} {\bibinfo  {journal} {Nature Communications}\ }\textbf {\bibinfo {volume} {16}} (\bibinfo {year} {2025})}\BibitemShut {NoStop}%
\bibitem [{\citenamefont {Motta}\ \emph {et~al.}(2019)\citenamefont {Motta}, \citenamefont {Sun}, \citenamefont {Tan}, \citenamefont {O’Rourke}, \citenamefont {Ye}, \citenamefont {Minnich}, \citenamefont {Brandão},\ and\ \citenamefont {Chan}}]{Motta2019}%
  \BibitemOpen
  \bibfield  {author} {\bibinfo {author} {\bibfnamefont {M.}~\bibnamefont {Motta}}, \bibinfo {author} {\bibfnamefont {C.}~\bibnamefont {Sun}}, \bibinfo {author} {\bibfnamefont {A.~T.~K.}\ \bibnamefont {Tan}}, \bibinfo {author} {\bibfnamefont {M.~J.}\ \bibnamefont {O’Rourke}}, \bibinfo {author} {\bibfnamefont {E.}~\bibnamefont {Ye}}, \bibinfo {author} {\bibfnamefont {A.~J.}\ \bibnamefont {Minnich}}, \bibinfo {author} {\bibfnamefont {F.~G. S.~L.}\ \bibnamefont {Brandão}},\ and\ \bibinfo {author} {\bibfnamefont {G.~K.-L.}\ \bibnamefont {Chan}},\ }\bibfield  {title} {\bibinfo {title} {\emph{Determining eigenstates and thermal states on a quantum computer using quantum imaginary time evolution}},\ }\href {https://doi.org/10.1038/s41567-019-0704-4} {\bibfield  {journal} {\bibinfo  {journal} {Nature Physics}\ }\textbf {\bibinfo {volume} {16}},\ \bibinfo {pages} {205–210} (\bibinfo {year} {2019})}\BibitemShut {NoStop}%
\bibitem [{\citenamefont {Silva}\ \emph {et~al.}(2023)\citenamefont {Silva}, \citenamefont {Taddei}, \citenamefont {Carrazza},\ and\ \citenamefont {Aolita}}]{Silva2023}%
  \BibitemOpen
  \bibfield  {author} {\bibinfo {author} {\bibfnamefont {T.~L.}\ \bibnamefont {Silva}}, \bibinfo {author} {\bibfnamefont {M.~M.}\ \bibnamefont {Taddei}}, \bibinfo {author} {\bibfnamefont {S.}~\bibnamefont {Carrazza}},\ and\ \bibinfo {author} {\bibfnamefont {L.}~\bibnamefont {Aolita}},\ }\bibfield  {title} {\bibinfo {title} {\emph{Fragmented imaginary-time evolution for early-stage quantum signal processors}},\ }\href {http://dx.doi.org/10.1038/s41598-023-45540-2} {\bibfield  {journal} {\bibinfo  {journal} {Scientific Reports}\ }\textbf {\bibinfo {volume} {13}} (\bibinfo {year} {2023})}\BibitemShut {NoStop}%
\bibitem [{\citenamefont {Tosta}\ \emph {et~al.}(2024)\citenamefont {Tosta}, \citenamefont {de~Lima~Silva}, \citenamefont {Camilo},\ and\ \citenamefont {Aolita}}]{Tosta2024}%
  \BibitemOpen
  \bibfield  {author} {\bibinfo {author} {\bibfnamefont {A.}~\bibnamefont {Tosta}}, \bibinfo {author} {\bibfnamefont {T.}~\bibnamefont {de~Lima~Silva}}, \bibinfo {author} {\bibfnamefont {G.}~\bibnamefont {Camilo}},\ and\ \bibinfo {author} {\bibfnamefont {L.}~\bibnamefont {Aolita}},\ }\bibfield  {title} {\bibinfo {title} {\emph{Randomized semi-quantum matrix processing}},\ }\href {http://dx.doi.org/10.1038/s41534-024-00883-0} {\bibfield  {journal} {\bibinfo  {journal} {npj Quantum Information}\ }\textbf {\bibinfo {volume} {10}} (\bibinfo {year} {2024})}\BibitemShut {NoStop}%
\bibitem [{\citenamefont {de~Lima~Silva}\ \emph {et~al.}(2024)\citenamefont {de~Lima~Silva}, \citenamefont {Borges},\ and\ \citenamefont {Aolita}}]{Silva2024}%
  \BibitemOpen
  \bibfield  {author} {\bibinfo {author} {\bibfnamefont {T.}~\bibnamefont {de~Lima~Silva}}, \bibinfo {author} {\bibfnamefont {L.}~\bibnamefont {Borges}},\ and\ \bibinfo {author} {\bibfnamefont {L.}~\bibnamefont {Aolita}},\ }\href {https://arxiv.org/abs/2411.17816} {\bibinfo {title} {\emph{Partition function estimation with a quantum coin toss}}} (\bibinfo {year} {2024}),\ \Eprint {https://arxiv.org/abs/2411.17816} {arXiv:2411.17816 [quant-ph]} \BibitemShut {NoStop}%
\bibitem [{\citenamefont {Lin}\ and\ \citenamefont {Tong}(2022)}]{Lin2022}%
  \BibitemOpen
  \bibfield  {author} {\bibinfo {author} {\bibfnamefont {L.}~\bibnamefont {Lin}}\ and\ \bibinfo {author} {\bibfnamefont {Y.}~\bibnamefont {Tong}},\ }\bibfield  {title} {\bibinfo {title} {\emph{Heisenberg-limited ground-state energy estimation for early fault-tolerant quantum computers}},\ }\href {https://doi.org/10.1103/PRXQuantum.3.010318} {\bibfield  {journal} {\bibinfo  {journal} {PRX Quantum}\ }\textbf {\bibinfo {volume} {3}},\ \bibinfo {pages} {010318} (\bibinfo {year} {2022})}\BibitemShut {NoStop}%
\bibitem [{\citenamefont {Wan}\ \emph {et~al.}(2022)\citenamefont {Wan}, \citenamefont {Berta},\ and\ \citenamefont {Campbell}}]{Wan2022}%
  \BibitemOpen
  \bibfield  {author} {\bibinfo {author} {\bibfnamefont {K.}~\bibnamefont {Wan}}, \bibinfo {author} {\bibfnamefont {M.}~\bibnamefont {Berta}},\ and\ \bibinfo {author} {\bibfnamefont {E.~T.}\ \bibnamefont {Campbell}},\ }\bibfield  {title} {\bibinfo {title} {\emph{Randomized quantum algorithm for statistical phase estimation}},\ }\href {http://dx.doi.org/10.1103/PhysRevLett.129.030503} {\bibfield  {journal} {\bibinfo  {journal} {Physical Review Letters}\ }\textbf {\bibinfo {volume} {129}} (\bibinfo {year} {2022})}\BibitemShut {NoStop}%
\bibitem [{\citenamefont {Berry}\ \emph {et~al.}(2024)\citenamefont {Berry}, \citenamefont {Tong}, \citenamefont {Khattar}, \citenamefont {White}, \citenamefont {Kim}, \citenamefont {Boixo}, \citenamefont {Lin}, \citenamefont {Lee}, \citenamefont {Chan}, \citenamefont {Babbush},\ and\ \citenamefont {Rubin}}]{Berry2024}%
  \BibitemOpen
  \bibfield  {author} {\bibinfo {author} {\bibfnamefont {D.~W.}\ \bibnamefont {Berry}}, \bibinfo {author} {\bibfnamefont {Y.}~\bibnamefont {Tong}}, \bibinfo {author} {\bibfnamefont {T.}~\bibnamefont {Khattar}}, \bibinfo {author} {\bibfnamefont {A.}~\bibnamefont {White}}, \bibinfo {author} {\bibfnamefont {T.~I.}\ \bibnamefont {Kim}}, \bibinfo {author} {\bibfnamefont {S.}~\bibnamefont {Boixo}}, \bibinfo {author} {\bibfnamefont {L.}~\bibnamefont {Lin}}, \bibinfo {author} {\bibfnamefont {S.}~\bibnamefont {Lee}}, \bibinfo {author} {\bibfnamefont {G.~K.-L.}\ \bibnamefont {Chan}}, \bibinfo {author} {\bibfnamefont {R.}~\bibnamefont {Babbush}},\ and\ \bibinfo {author} {\bibfnamefont {N.~C.}\ \bibnamefont {Rubin}},\ }\href {https://arxiv.org/abs/2409.11748} {\bibinfo {title} {\emph{Rapid initial state preparation for the quantum simulation of strongly correlated molecules}}} (\bibinfo {year} {2024}),\ \Eprint {https://arxiv.org/abs/2409.11748} {arXiv:2409.11748 [quant-ph]} \BibitemShut {NoStop}%
\bibitem [{\citenamefont {Li}\ and\ \citenamefont {Chan}(2017)}]{Li2017spin}%
  \BibitemOpen
  \bibfield  {author} {\bibinfo {author} {\bibfnamefont {Z.}~\bibnamefont {Li}}\ and\ \bibinfo {author} {\bibfnamefont {G.~K.-L.}\ \bibnamefont {Chan}},\ }\bibfield  {title} {\bibinfo {title} {\emph{Spin-projected matrix product states: {V}ersatile tool for strongly correlated systems}},\ }\href {https://doi.org/10.1021/acs.jctc.7b00270} {\bibfield  {journal} {\bibinfo  {journal} {Journal of Chemical Theory and Computation}\ }\textbf {\bibinfo {volume} {13}},\ \bibinfo {pages} {2681} (\bibinfo {year} {2017})}\BibitemShut {NoStop}%
\bibitem [{\citenamefont {Mhiri}\ \emph {et~al.}(2025)\citenamefont {Mhiri}, \citenamefont {Puig}, \citenamefont {Lerch}, \citenamefont {Rudolph}, \citenamefont {Chotibut}, \citenamefont {Thanasilp},\ and\ \citenamefont {Holmes}}]{Mhiri2025}%
  \BibitemOpen
  \bibfield  {author} {\bibinfo {author} {\bibfnamefont {H.}~\bibnamefont {Mhiri}}, \bibinfo {author} {\bibfnamefont {R.}~\bibnamefont {Puig}}, \bibinfo {author} {\bibfnamefont {S.}~\bibnamefont {Lerch}}, \bibinfo {author} {\bibfnamefont {M.~S.}\ \bibnamefont {Rudolph}}, \bibinfo {author} {\bibfnamefont {T.}~\bibnamefont {Chotibut}}, \bibinfo {author} {\bibfnamefont {S.}~\bibnamefont {Thanasilp}},\ and\ \bibinfo {author} {\bibfnamefont {Z.}~\bibnamefont {Holmes}},\ }\href {https://arxiv.org/abs/2502.07889} {\bibinfo {title} {\emph{A unifying account of warm start guarantees for patches of quantum landscapes}}} (\bibinfo {year} {2025}),\ \Eprint {https://arxiv.org/abs/2502.07889} {arXiv:2502.07889 [quant-ph]} \BibitemShut {NoStop}%
\bibitem [{\citenamefont {Hoerl}\ and\ \citenamefont {Kennard}(1970)}]{Hoerl1970}%
  \BibitemOpen
  \bibfield  {author} {\bibinfo {author} {\bibfnamefont {A.~E.}\ \bibnamefont {Hoerl}}\ and\ \bibinfo {author} {\bibfnamefont {R.~W.}\ \bibnamefont {Kennard}},\ }\bibfield  {title} {\bibinfo {title} {\emph{{R}idge regression: {B}iased estimation for nonorthogonal problems}},\ }\href@noop {} {\bibfield  {journal} {\bibinfo  {journal} {Technometrics}\ }\textbf {\bibinfo {volume} {12}},\ \bibinfo {pages} {55} (\bibinfo {year} {1970})}\BibitemShut {NoStop}%
\bibitem [{\citenamefont {Bergholm}\ \emph {et~al.}(2022)\citenamefont {Bergholm}, \citenamefont {Izaac}, \citenamefont {Schuld}, \citenamefont {Gogolin}, \citenamefont {Ahmed}, \citenamefont {Ajith}, \citenamefont {Alam}, \citenamefont {Alonso-Linaje}, \citenamefont {AkashNarayanan}, \citenamefont {Asadi} \emph {et~al.}}]{Pennylane2022}%
  \BibitemOpen
  \bibfield  {author} {\bibinfo {author} {\bibfnamefont {V.}~\bibnamefont {Bergholm}}, \bibinfo {author} {\bibfnamefont {J.}~\bibnamefont {Izaac}}, \bibinfo {author} {\bibfnamefont {M.}~\bibnamefont {Schuld}}, \bibinfo {author} {\bibfnamefont {C.}~\bibnamefont {Gogolin}}, \bibinfo {author} {\bibfnamefont {S.}~\bibnamefont {Ahmed}}, \bibinfo {author} {\bibfnamefont {V.}~\bibnamefont {Ajith}}, \bibinfo {author} {\bibfnamefont {M.~S.}\ \bibnamefont {Alam}}, \bibinfo {author} {\bibfnamefont {G.}~\bibnamefont {Alonso-Linaje}}, \bibinfo {author} {\bibfnamefont {B.}~\bibnamefont {AkashNarayanan}}, \bibinfo {author} {\bibfnamefont {A.}~\bibnamefont {Asadi}}, \emph {et~al.},\ }\href {https://arxiv.org/abs/1811.04968} {\bibinfo {title} {\emph{{PennyLane}: {A}utomatic differentiation of hybrid quantum-classical computations}}} (\bibinfo {year} {2022}),\ \Eprint {https://arxiv.org/abs/1811.04968} {arXiv:1811.04968 [quant-ph]} \BibitemShut {NoStop}%
\bibitem [{\citenamefont {Azad}\ and\ \citenamefont {Fomichev}(2023)}]{Utkarsh2023}%
  \BibitemOpen
  \bibfield  {author} {\bibinfo {author} {\bibfnamefont {U.}~\bibnamefont {Azad}}\ and\ \bibinfo {author} {\bibfnamefont {S.}~\bibnamefont {Fomichev}},\ }\href@noop {} {\bibinfo {title} {\emph{{PennyLane} quantum chemistry datasets}}} (\bibinfo {year} {2023})\BibitemShut {NoStop}%
\bibitem [{\citenamefont {Robbiati}\ \emph {et~al.}(2025)\citenamefont {Robbiati}, \citenamefont {Papaluca}, \citenamefont {Pasquale}, \citenamefont {Pedicillo}, \citenamefont {Farias}, \citenamefont {Sopena}, \citenamefont {Robbiano}, \citenamefont {Alramahi}, \citenamefont {Bordoni}, \citenamefont {Candido} \emph {et~al.}}]{Robbiati2025}%
  \BibitemOpen
  \bibfield  {author} {\bibinfo {author} {\bibfnamefont {M.}~\bibnamefont {Robbiati}}, \bibinfo {author} {\bibfnamefont {A.}~\bibnamefont {Papaluca}}, \bibinfo {author} {\bibfnamefont {A.}~\bibnamefont {Pasquale}}, \bibinfo {author} {\bibfnamefont {E.}~\bibnamefont {Pedicillo}}, \bibinfo {author} {\bibfnamefont {R.~M.~S.}\ \bibnamefont {Farias}}, \bibinfo {author} {\bibfnamefont {A.}~\bibnamefont {Sopena}}, \bibinfo {author} {\bibfnamefont {M.}~\bibnamefont {Robbiano}}, \bibinfo {author} {\bibfnamefont {G.}~\bibnamefont {Alramahi}}, \bibinfo {author} {\bibfnamefont {S.}~\bibnamefont {Bordoni}}, \bibinfo {author} {\bibfnamefont {A.}~\bibnamefont {Candido}}, \emph {et~al.},\ }\href {https://arxiv.org/abs/2510.11773} {\bibinfo {title} {\emph{{Qiboml}: {T}owards the orchestration of quantum-classical machine learning}}} (\bibinfo {year} {2025}),\ \Eprint {https://arxiv.org/abs/2510.11773} {arXiv:2510.11773 [quant-ph]} \BibitemShut {NoStop}%
\bibitem [{\citenamefont {Papaluca}\ \emph {et~al.}(2025)\citenamefont {Papaluca}, \citenamefont {Robbiati}, \citenamefont {Pedicillo}, \citenamefont {Farias}, \citenamefont {Laurora}, \citenamefont {Sopena}, \citenamefont {Ramahi}, \citenamefont {Pasquale}, \citenamefont {Carrazza},\ and\ \citenamefont {Candido}}]{QibomlZenodo}%
  \BibitemOpen
  \bibfield  {author} {\bibinfo {author} {\bibfnamefont {A.}~\bibnamefont {Papaluca}}, \bibinfo {author} {\bibfnamefont {M.}~\bibnamefont {Robbiati}}, \bibinfo {author} {\bibfnamefont {E.}~\bibnamefont {Pedicillo}}, \bibinfo {author} {\bibfnamefont {R.~M.~S.}\ \bibnamefont {Farias}}, \bibinfo {author} {\bibfnamefont {N.}~\bibnamefont {Laurora}}, \bibinfo {author} {\bibfnamefont {A.}~\bibnamefont {Sopena}}, \bibinfo {author} {\bibfnamefont {G.~A.}\ \bibnamefont {Ramahi}}, \bibinfo {author} {\bibfnamefont {A.}~\bibnamefont {Pasquale}}, \bibinfo {author} {\bibfnamefont {S.}~\bibnamefont {Carrazza}},\ and\ \bibinfo {author} {\bibfnamefont {A.}~\bibnamefont {Candido}},\ }\href {https://doi.org/10.5281/zenodo.17310379} {\bibinfo {title} {qiboteam/qiboml: qiboml 0.1.0}} (\bibinfo {year} {2025})\BibitemShut {NoStop}%
\bibitem [{\citenamefont {Efthymiou}\ \emph {et~al.}(2021)\citenamefont {Efthymiou}, \citenamefont {Ramos-Calderer}, \citenamefont {Bravo-Prieto}, \citenamefont {P\'{e}rez-Salinas}, \citenamefont {Garc\'{i}a-Mart\'{i}n}, \citenamefont {Garcia-Saez}, \citenamefont {Latorre},\ and\ \citenamefont {Carrazza}}]{Qibo2021}%
  \BibitemOpen
  \bibfield  {author} {\bibinfo {author} {\bibfnamefont {S.}~\bibnamefont {Efthymiou}}, \bibinfo {author} {\bibfnamefont {S.}~\bibnamefont {Ramos-Calderer}}, \bibinfo {author} {\bibfnamefont {C.}~\bibnamefont {Bravo-Prieto}}, \bibinfo {author} {\bibfnamefont {A.}~\bibnamefont {P\'{e}rez-Salinas}}, \bibinfo {author} {\bibfnamefont {D.}~\bibnamefont {Garc\'{i}a-Mart\'{i}n}}, \bibinfo {author} {\bibfnamefont {A.}~\bibnamefont {Garcia-Saez}}, \bibinfo {author} {\bibfnamefont {J.~I.}\ \bibnamefont {Latorre}},\ and\ \bibinfo {author} {\bibfnamefont {S.}~\bibnamefont {Carrazza}},\ }\bibfield  {title} {\bibinfo {title} {\emph{{Q}ibo: A framework for quantum simulation with hardware acceleration}},\ }\href {https://doi.org/10.1088/2058-9565/ac39f5} {\bibfield  {journal} {\bibinfo  {journal} {Quantum Science and Technology}\ }\textbf {\bibinfo {volume} {7}},\ \bibinfo {pages} {015018} (\bibinfo {year} {2021})}\BibitemShut {NoStop}%
\bibitem [{\citenamefont {Farias}\ \emph {et~al.}(2025{\natexlab{b}})\citenamefont {Farias}, \citenamefont {Efthymiou}, \citenamefont {Carrazza}, \citenamefont {Papaluca}, \citenamefont {Robbiati}, \citenamefont {Sopena}, \citenamefont {Bordoni}, \citenamefont {Pedicillo}, \citenamefont {Pasquale}, \citenamefont {Candido} \emph {et~al.}}]{QiboZenodo}%
  \BibitemOpen
  \bibfield  {author} {\bibinfo {author} {\bibfnamefont {R.~M.~S.}\ \bibnamefont {Farias}}, \bibinfo {author} {\bibfnamefont {S.}~\bibnamefont {Efthymiou}}, \bibinfo {author} {\bibfnamefont {S.}~\bibnamefont {Carrazza}}, \bibinfo {author} {\bibfnamefont {A.}~\bibnamefont {Papaluca}}, \bibinfo {author} {\bibfnamefont {M.}~\bibnamefont {Robbiati}}, \bibinfo {author} {\bibfnamefont {A.}~\bibnamefont {Sopena}}, \bibinfo {author} {\bibfnamefont {S.}~\bibnamefont {Bordoni}}, \bibinfo {author} {\bibfnamefont {E.}~\bibnamefont {Pedicillo}}, \bibinfo {author} {\bibfnamefont {A.}~\bibnamefont {Pasquale}}, \bibinfo {author} {\bibfnamefont {A.}~\bibnamefont {Candido}}, \emph {et~al.},\ }\href {https://doi.org/10.5281/zenodo.17822128} {\bibinfo {title} {qiboteam/qibo: qibo 0.2.23}} (\bibinfo {year} {2025}{\natexlab{b}})\BibitemShut {NoStop}%
\bibitem [{\citenamefont {Paszke}\ \emph {et~al.}(2019)\citenamefont {Paszke}, \citenamefont {Gross}, \citenamefont {Massa}, \citenamefont {Lerer}, \citenamefont {Bradbury}, \citenamefont {Chanan}, \citenamefont {Killeen}, \citenamefont {Lin}, \citenamefont {Gimelshein}, \citenamefont {Antiga} \emph {et~al.}}]{Pytorch2019}%
  \BibitemOpen
  \bibfield  {author} {\bibinfo {author} {\bibfnamefont {A.}~\bibnamefont {Paszke}}, \bibinfo {author} {\bibfnamefont {S.}~\bibnamefont {Gross}}, \bibinfo {author} {\bibfnamefont {F.}~\bibnamefont {Massa}}, \bibinfo {author} {\bibfnamefont {A.}~\bibnamefont {Lerer}}, \bibinfo {author} {\bibfnamefont {J.}~\bibnamefont {Bradbury}}, \bibinfo {author} {\bibfnamefont {G.}~\bibnamefont {Chanan}}, \bibinfo {author} {\bibfnamefont {T.}~\bibnamefont {Killeen}}, \bibinfo {author} {\bibfnamefont {Z.}~\bibnamefont {Lin}}, \bibinfo {author} {\bibfnamefont {N.}~\bibnamefont {Gimelshein}}, \bibinfo {author} {\bibfnamefont {L.}~\bibnamefont {Antiga}}, \emph {et~al.},\ }\href {https://arxiv.org/abs/1912.01703} {\bibinfo {title} {\emph{{PyTorch}: {A}n imperative style, high-performance deep learning library}}} (\bibinfo {year} {2019}),\ \Eprint {https://arxiv.org/abs/1912.01703} {arXiv:1912.01703 [cs.LG]} \BibitemShut {NoStop}%
\bibitem [{\citenamefont {Nogueira}(2014)}]{Nogueira2014}%
  \BibitemOpen
  \bibfield  {author} {\bibinfo {author} {\bibfnamefont {F.}~\bibnamefont {Nogueira}},\ }\href {https://github.com/bayesian-optimization/BayesianOptimization} {\bibinfo {title} {\emph{Bayesian optimization: {O}pen source constrained global optimization tool for {P}ython}}} (\bibinfo {year} {2014})\BibitemShut {NoStop}%
\bibitem [{\citenamefont {Dai}\ and\ \citenamefont {Yuan}(2001)}]{dai2001}%
  \BibitemOpen
  \bibfield  {author} {\bibinfo {author} {\bibfnamefont {Y.~H.}\ \bibnamefont {Dai}}\ and\ \bibinfo {author} {\bibfnamefont {Y.}~\bibnamefont {Yuan}},\ }\bibfield  {title} {\bibinfo {title} {\emph{An efficient hybrid conjugate gradient method for unconstrained optimization}},\ }\href {https://doi.org/10.1023/A:1012930416777} {\bibfield  {journal} {\bibinfo  {journal} {Annals of Operations Research}\ }\textbf {\bibinfo {volume} {103}},\ \bibinfo {pages} {33} (\bibinfo {year} {2001})}\BibitemShut {NoStop}%
\bibitem [{\citenamefont {Kingma}\ and\ \citenamefont {Ba}(2017)}]{Kingma2017}%
  \BibitemOpen
  \bibfield  {author} {\bibinfo {author} {\bibfnamefont {D.~P.}\ \bibnamefont {Kingma}}\ and\ \bibinfo {author} {\bibfnamefont {J.}~\bibnamefont {Ba}},\ }\href {https://arxiv.org/abs/1412.6980} {\bibinfo {title} {\emph{{Adam}: {A} method for stochastic optimization}}} (\bibinfo {year} {2017}),\ \Eprint {https://arxiv.org/abs/1412.6980} {arXiv:1412.6980 [cs.LG]} \BibitemShut {NoStop}%
\bibitem [{\citenamefont {Robbiati}\ \emph {et~al.}(2024)\citenamefont {Robbiati}, \citenamefont {Pedicillo}, \citenamefont {Pasquale}, \citenamefont {Li}, \citenamefont {Wright}, \citenamefont {Farias}, \citenamefont {Giang}, \citenamefont {Son}, \citenamefont {Kn\"{o}rzer}, \citenamefont {Goh} \emph {et~al.}}]{Robbiati2024}%
  \BibitemOpen
  \bibfield  {author} {\bibinfo {author} {\bibfnamefont {M.}~\bibnamefont {Robbiati}}, \bibinfo {author} {\bibfnamefont {E.}~\bibnamefont {Pedicillo}}, \bibinfo {author} {\bibfnamefont {A.}~\bibnamefont {Pasquale}}, \bibinfo {author} {\bibfnamefont {X.}~\bibnamefont {Li}}, \bibinfo {author} {\bibfnamefont {A.}~\bibnamefont {Wright}}, \bibinfo {author} {\bibfnamefont {R.~M.~S.}\ \bibnamefont {Farias}}, \bibinfo {author} {\bibfnamefont {K.~U.}\ \bibnamefont {Giang}}, \bibinfo {author} {\bibfnamefont {J.}~\bibnamefont {Son}}, \bibinfo {author} {\bibfnamefont {J.}~\bibnamefont {Kn\"{o}rzer}}, \bibinfo {author} {\bibfnamefont {S.~T.}\ \bibnamefont {Goh}}, \emph {et~al.},\ }\href {https://arxiv.org/abs/2408.03987} {\bibinfo {title} {\emph{Double-bracket quantum algorithms for high-fidelity ground state preparation}}} (\bibinfo {year} {2024}),\ \Eprint {https://arxiv.org/abs/2408.03987} {arXiv:2408.03987 [quant-ph]} \BibitemShut {NoStop}%
\bibitem [{\citenamefont {Kandala}\ \emph {et~al.}(2017)\citenamefont {Kandala}, \citenamefont {Mezzacapo}, \citenamefont {Temme}, \citenamefont {Takita}, \citenamefont {Brink}, \citenamefont {Chow},\ and\ \citenamefont {Gambetta}}]{Kandala2017}%
  \BibitemOpen
  \bibfield  {author} {\bibinfo {author} {\bibfnamefont {A.}~\bibnamefont {Kandala}}, \bibinfo {author} {\bibfnamefont {A.}~\bibnamefont {Mezzacapo}}, \bibinfo {author} {\bibfnamefont {K.}~\bibnamefont {Temme}}, \bibinfo {author} {\bibfnamefont {M.}~\bibnamefont {Takita}}, \bibinfo {author} {\bibfnamefont {M.}~\bibnamefont {Brink}}, \bibinfo {author} {\bibfnamefont {J.~M.}\ \bibnamefont {Chow}},\ and\ \bibinfo {author} {\bibfnamefont {J.~M.}\ \bibnamefont {Gambetta}},\ }\bibfield  {title} {\bibinfo {title} {\emph{Hardware-efficient variational quantum eigensolver for small molecules and quantum magnets}},\ }\href {https://doi.org/10.1038/nature23879} {\bibfield  {journal} {\bibinfo  {journal} {Nature}\ }\textbf {\bibinfo {volume} {549}},\ \bibinfo {pages} {242} (\bibinfo {year} {2017})}\BibitemShut {NoStop}%
\bibitem [{\citenamefont {{IonQ Inc.}}(2025{\natexlab{c}})}]{qis}%
  \BibitemOpen
  \bibfield  {author} {\bibinfo {author} {\bibnamefont {{IonQ Inc.}}},\ }\href {https://docs.ionq.com/api-reference/v0.3/writing-quantum-programs#supported-gates} {\bibinfo {title} {\emph{Supported gates}}} (\bibinfo {year} {2025}{\natexlab{c}})\BibitemShut {NoStop}%
\bibitem [{\citenamefont {O'Neill}(2006)}]{Oneill2006}%
  \BibitemOpen
  \bibfield  {author} {\bibinfo {author} {\bibfnamefont {B.}~\bibnamefont {O'Neill}},\ }\href {https://doi.org/10.1016/C2009-0-05241-6} {\emph {\bibinfo {title} {Elementary differential geometry}}},\ \bibinfo {edition} {{R}ev. 2nd}\ ed.\ (\bibinfo  {publisher} {Elsevier Academic Press},\ \bibinfo {year} {2006})\BibitemShut {NoStop}%
\bibitem [{\citenamefont {Blumenson}(1960)}]{Blumenson1960}%
  \BibitemOpen
  \bibfield  {author} {\bibinfo {author} {\bibfnamefont {L.~E.}\ \bibnamefont {Blumenson}},\ }\bibfield  {title} {\bibinfo {title} {\emph{A derivation of $n$-dimensional spherical coordinates}},\ }\href {https://doi.org/10.2307/2308932} {\bibfield  {journal} {\bibinfo  {journal} {The American Mathematical Monthly}\ }\textbf {\bibinfo {volume} {67}},\ \bibinfo {pages} {63} (\bibinfo {year} {1960})}\BibitemShut {NoStop}%
\bibitem [{Note1()}]{Note1}%
  \BibitemOpen
  \bibinfo {note} {Which is the same Jacobian matrix as the one found in the main text just below Eq.\ref {eq:update_exact}, but without the index $\protect \mathbf {x}$, which is dropped for convenience. For explicit formulas, see Eq.\protect \eqref {eq:jacob_matrix} in App.\ref {app:jacobian_regularization}}\BibitemShut {NoStop}%
\bibitem [{\citenamefont {Bengtsson}\ and\ \citenamefont {{\.Z}yczkowski}(2017)}]{Bengtsson2017}%
  \BibitemOpen
  \bibfield  {author} {\bibinfo {author} {\bibfnamefont {I.}~\bibnamefont {Bengtsson}}\ and\ \bibinfo {author} {\bibfnamefont {K.}~\bibnamefont {{\.Z}yczkowski}},\ }\href {https://doi.org/10.1017/CBO9780511535048} {\emph {\bibinfo {title} {{G}eometry of quantum states: {A}n introduction to quantum entanglement}}}\ (\bibinfo  {publisher} {Cambridge University Press},\ \bibinfo {year} {2017})\BibitemShut {NoStop}%
\bibitem [{\citenamefont {Lee}(2018)}]{Lee2018}%
  \BibitemOpen
  \bibfield  {author} {\bibinfo {author} {\bibfnamefont {J.~M.}\ \bibnamefont {Lee}},\ }\href {https://doi.org/10.1007/978-3-319-91755-9} {\emph {\bibinfo {title} {Introduction to {Riemannian} {Manifolds}}}},\ \bibinfo {series} {Graduate {Texts} in {Mathematics}}, Vol.\ \bibinfo {volume} {176}\ (\bibinfo  {publisher} {Springer International Publishing},\ \bibinfo {year} {2018})\BibitemShut {NoStop}%
\bibitem [{\citenamefont {Fang}(2018)}]{fang2018}%
  \BibitemOpen
  \bibfield  {author} {\bibinfo {author} {\bibfnamefont {K.~W.}\ \bibnamefont {Fang}},\ }\href {https://doi.org/10.1201/9781351077040} {\emph {\bibinfo {title} {Symmetric multivariate and related distributions}}}\ (\bibinfo  {publisher} {Chapman and Hall/CRC},\ \bibinfo {year} {2018})\BibitemShut {NoStop}%
\bibitem [{\citenamefont {Dai}\ and\ \citenamefont {Yuan}(1999)}]{dai1999}%
  \BibitemOpen
  \bibfield  {author} {\bibinfo {author} {\bibfnamefont {Y.~H.}\ \bibnamefont {Dai}}\ and\ \bibinfo {author} {\bibfnamefont {Y.}~\bibnamefont {Yuan}},\ }\bibfield  {title} {\bibinfo {title} {\emph{A nonlinear conjugate gradient method with a strong global convergence property}},\ }\href {https://doi.org/10.1137/S1052623497318992} {\bibfield  {journal} {\bibinfo  {journal} {SIAM Journal on Optimization}\ }\textbf {\bibinfo {volume} {10}},\ \bibinfo {pages} {177} (\bibinfo {year} {1999})}\BibitemShut {NoStop}%
\bibitem [{\citenamefont {Hestenes}\ and\ \citenamefont {Stiefel}(1952)}]{hestenes1952}%
  \BibitemOpen
  \bibfield  {author} {\bibinfo {author} {\bibfnamefont {M.~R.}\ \bibnamefont {Hestenes}}\ and\ \bibinfo {author} {\bibfnamefont {E.}~\bibnamefont {Stiefel}},\ }\bibfield  {title} {\bibinfo {title} {\emph{Methods of conjugate gradients for solving linear systems}},\ }\href {https://api.semanticscholar.org/CorpusID:2207234} {\bibfield  {journal} {\bibinfo  {journal} {Journal of Research of the National Bureau of Standards}\ }\textbf {\bibinfo {volume} {49}},\ \bibinfo {pages} {409} (\bibinfo {year} {1952})}\BibitemShut {NoStop}%
\bibitem [{\citenamefont {Truong}\ and\ \citenamefont {Nguyen}(2021)}]{truong2021}%
  \BibitemOpen
  \bibfield  {author} {\bibinfo {author} {\bibfnamefont {T.~T.}\ \bibnamefont {Truong}}\ and\ \bibinfo {author} {\bibfnamefont {H.-T.}\ \bibnamefont {Nguyen}},\ }\bibfield  {title} {\bibinfo {title} {\emph{Backtracking gradient descent method and some applications in large scale optimisation. {P}art 2: {A}lgorithms and experiments}},\ }\href {https://doi.org/10.1007/s00245-020-09718-8} {\bibfield  {journal} {\bibinfo  {journal} {Applied Mathematics {\&} Optimization}\ }\textbf {\bibinfo {volume} {84}},\ \bibinfo {pages} {2557} (\bibinfo {year} {2021})}\BibitemShut {NoStop}%
\bibitem [{\citenamefont {Press}\ \emph {et~al.}(2007)\citenamefont {Press}, \citenamefont {Teukolsky}, \citenamefont {Vetterling},\ and\ \citenamefont {Flannery}}]{press2007}%
  \BibitemOpen
  \bibfield  {author} {\bibinfo {author} {\bibfnamefont {W.~H.}\ \bibnamefont {Press}}, \bibinfo {author} {\bibfnamefont {S.~A.}\ \bibnamefont {Teukolsky}}, \bibinfo {author} {\bibfnamefont {W.~T.}\ \bibnamefont {Vetterling}},\ and\ \bibinfo {author} {\bibfnamefont {B.~P.}\ \bibnamefont {Flannery}},\ }\href@noop {} {\emph {\bibinfo {title} {Numerical recipes: {T}he art of scientific computing}}},\ \bibinfo {edition} {3rd}\ ed.\ (\bibinfo  {publisher} {Cambridge University Press},\ \bibinfo {year} {2007})\BibitemShut {NoStop}%
\bibitem [{\citenamefont {Gluza}\ \emph {et~al.}(2026)\citenamefont {Gluza}, \citenamefont {Son}, \citenamefont {Tiang}, \citenamefont {Zander}, \citenamefont {Seidel}, \citenamefont {Suzuki}, \citenamefont {Holmes},\ and\ \citenamefont {Ng}}]{Gluza2025}%
  \BibitemOpen
  \bibfield  {author} {\bibinfo {author} {\bibfnamefont {M.}~\bibnamefont {Gluza}}, \bibinfo {author} {\bibfnamefont {J.}~\bibnamefont {Son}}, \bibinfo {author} {\bibfnamefont {B.~H.}\ \bibnamefont {Tiang}}, \bibinfo {author} {\bibfnamefont {R.}~\bibnamefont {Zander}}, \bibinfo {author} {\bibfnamefont {R.}~\bibnamefont {Seidel}}, \bibinfo {author} {\bibfnamefont {Y.}~\bibnamefont {Suzuki}}, \bibinfo {author} {\bibfnamefont {Z.}~\bibnamefont {Holmes}},\ and\ \bibinfo {author} {\bibfnamefont {N.~H.~Y.}\ \bibnamefont {Ng}},\ }\bibfield  {title} {\bibinfo {title} {\emph{Double-bracket quantum algorithms for quantum imaginary-time evolution}},\ }\href {https://link.aps.org/doi/10.1103/rw81-k8vk} {\bibfield  {journal} {\bibinfo  {journal} {Phys. Rev. Lett.}\ }\textbf {\bibinfo {volume} {136}},\ \bibinfo {pages} {020601} (\bibinfo {year} {2026})}\BibitemShut {NoStop}%
\bibitem [{\citenamefont {Suzuki}\ \emph {et~al.}(2025)\citenamefont {Suzuki}, \citenamefont {Gluza}, \citenamefont {Son}, \citenamefont {Tiang}, \citenamefont {Ng},\ and\ \citenamefont {Holmes}}]{Suzuki2025}%
  \BibitemOpen
  \bibfield  {author} {\bibinfo {author} {\bibfnamefont {Y.}~\bibnamefont {Suzuki}}, \bibinfo {author} {\bibfnamefont {M.}~\bibnamefont {Gluza}}, \bibinfo {author} {\bibfnamefont {J.}~\bibnamefont {Son}}, \bibinfo {author} {\bibfnamefont {B.~H.}\ \bibnamefont {Tiang}}, \bibinfo {author} {\bibfnamefont {N.~H.~Y.}\ \bibnamefont {Ng}},\ and\ \bibinfo {author} {\bibfnamefont {Z.}~\bibnamefont {Holmes}},\ }\href {https://arxiv.org/abs/2507.15065} {\bibinfo {title} {\emph{Grover's algorithm is an approximation of imaginary-time evolution}}} (\bibinfo {year} {2025}),\ \Eprint {https://arxiv.org/abs/2507.15065} {arXiv:2507.15065 [quant-ph]} \BibitemShut {NoStop}%
\bibitem [{\citenamefont {Yuan}\ \emph {et~al.}(2019)\citenamefont {Yuan}, \citenamefont {Endo}, \citenamefont {Zhao}, \citenamefont {Li},\ and\ \citenamefont {Benjamin}}]{Yuan2019}%
  \BibitemOpen
  \bibfield  {author} {\bibinfo {author} {\bibfnamefont {X.}~\bibnamefont {Yuan}}, \bibinfo {author} {\bibfnamefont {S.}~\bibnamefont {Endo}}, \bibinfo {author} {\bibfnamefont {Q.}~\bibnamefont {Zhao}}, \bibinfo {author} {\bibfnamefont {Y.}~\bibnamefont {Li}},\ and\ \bibinfo {author} {\bibfnamefont {S.~C.}\ \bibnamefont {Benjamin}},\ }\bibfield  {title} {\bibinfo {title} {\emph{Theory of variational quantum simulation}},\ }\href {https://doi.org/10.22331/q-2019-10-07-191} {\bibfield  {journal} {\bibinfo  {journal} {Quantum}\ }\textbf {\bibinfo {volume} {3}},\ \bibinfo {pages} {191} (\bibinfo {year} {2019})}\BibitemShut {NoStop}%
\bibitem [{\citenamefont {McArdle}\ \emph {et~al.}(2019)\citenamefont {McArdle}, \citenamefont {Jones}, \citenamefont {Endo}, \citenamefont {Li}, \citenamefont {Benjamin},\ and\ \citenamefont {Yuan}}]{McArdle2019}%
  \BibitemOpen
  \bibfield  {author} {\bibinfo {author} {\bibfnamefont {S.}~\bibnamefont {McArdle}}, \bibinfo {author} {\bibfnamefont {T.}~\bibnamefont {Jones}}, \bibinfo {author} {\bibfnamefont {S.}~\bibnamefont {Endo}}, \bibinfo {author} {\bibfnamefont {Y.}~\bibnamefont {Li}}, \bibinfo {author} {\bibfnamefont {S.~C.}\ \bibnamefont {Benjamin}},\ and\ \bibinfo {author} {\bibfnamefont {X.}~\bibnamefont {Yuan}},\ }\bibfield  {title} {\bibinfo {title} {\emph{Variational ansatz-based quantum simulation of imaginary time evolution}},\ }\href {https://doi.org/10.1038/s41534-019-0187-2} {\bibfield  {journal} {\bibinfo  {journal} {npj Quantum Information}\ }\textbf {\bibinfo {volume} {5}},\ \bibinfo {pages} {75} (\bibinfo {year} {2019})}\BibitemShut {NoStop}%
\bibitem [{\citenamefont {Yu}\ \emph {et~al.}(2025)\citenamefont {Yu}, \citenamefont {Moreno}, \citenamefont {Iosue}, \citenamefont {Bertels}, \citenamefont {Claudino}, \citenamefont {Fuller}, \citenamefont {Groszkowski}, \citenamefont {Humble}, \citenamefont {Jurcevic}, \citenamefont {Kirby} \emph {et~al.}}]{Yu2025}%
  \BibitemOpen
  \bibfield  {author} {\bibinfo {author} {\bibfnamefont {J.}~\bibnamefont {Yu}}, \bibinfo {author} {\bibfnamefont {J.~R.}\ \bibnamefont {Moreno}}, \bibinfo {author} {\bibfnamefont {J.~T.}\ \bibnamefont {Iosue}}, \bibinfo {author} {\bibfnamefont {L.}~\bibnamefont {Bertels}}, \bibinfo {author} {\bibfnamefont {D.}~\bibnamefont {Claudino}}, \bibinfo {author} {\bibfnamefont {B.}~\bibnamefont {Fuller}}, \bibinfo {author} {\bibfnamefont {P.}~\bibnamefont {Groszkowski}}, \bibinfo {author} {\bibfnamefont {T.~S.}\ \bibnamefont {Humble}}, \bibinfo {author} {\bibfnamefont {P.}~\bibnamefont {Jurcevic}}, \bibinfo {author} {\bibfnamefont {W.}~\bibnamefont {Kirby}}, \emph {et~al.},\ }\href {https://arxiv.org/abs/2501.09702} {\bibinfo {title} {\emph{Quantum-centric algorithm for sample-based {K}rylov diagonalization}}} (\bibinfo {year} {2025}),\ \Eprint {https://arxiv.org/abs/2501.09702} {arXiv:2501.09702 [quant-ph]} \BibitemShut {NoStop}%
\bibitem [{\citenamefont {Raj}\ and\ \citenamefont {Coyle}(2025)}]{Raj2025}%
  \BibitemOpen
  \bibfield  {author} {\bibinfo {author} {\bibfnamefont {S.}~\bibnamefont {Raj}}\ and\ \bibinfo {author} {\bibfnamefont {B.}~\bibnamefont {Coyle}},\ }\href {https://arxiv.org/abs/2502.06916} {\bibinfo {title} {\emph{Hyper compressed fine-tuning of large foundation models with quantum-inspired adapters}}} (\bibinfo {year} {2025}),\ \Eprint {https://arxiv.org/abs/2502.06916} {arXiv:2502.06916 [cs.LG]} \BibitemShut {NoStop}%
\bibitem [{\citenamefont {Datta}(2010)}]{Datta2010}%
  \BibitemOpen
  \bibfield  {author} {\bibinfo {author} {\bibfnamefont {B.~N.}\ \bibnamefont {Datta}},\ }\href {https://doi.org/10.1137/1.9780898717655} {\emph {\bibinfo {title} {Numerical Linear Algebra and Applications, 2nd Edition}}},\ \bibinfo {edition} {2nd}\ ed.\ (\bibinfo  {publisher} {Society for Industrial and Applied Mathematics},\ \bibinfo {address} {Philadelphia, PA},\ \bibinfo {year} {2010})\BibitemShut {NoStop}%
\bibitem [{\citenamefont {Abbas}\ \emph {et~al.}(2023)\citenamefont {Abbas}, \citenamefont {King}, \citenamefont {Huang}, \citenamefont {Huggins}, \citenamefont {Movassagh}, \citenamefont {Gilboa},\ and\ \citenamefont {McClean}}]{Abbas2023}%
  \BibitemOpen
  \bibfield  {author} {\bibinfo {author} {\bibfnamefont {A.}~\bibnamefont {Abbas}}, \bibinfo {author} {\bibfnamefont {R.}~\bibnamefont {King}}, \bibinfo {author} {\bibfnamefont {H.-Y.}\ \bibnamefont {Huang}}, \bibinfo {author} {\bibfnamefont {W.~J.}\ \bibnamefont {Huggins}}, \bibinfo {author} {\bibfnamefont {R.}~\bibnamefont {Movassagh}}, \bibinfo {author} {\bibfnamefont {D.}~\bibnamefont {Gilboa}},\ and\ \bibinfo {author} {\bibfnamefont {J.}~\bibnamefont {McClean}},\ }\bibfield  {title} {\bibinfo {title} {\emph{On quantum backpropagation, information reuse, and cheating measurement collapse}},\ }in\ \href {https://proceedings.neurips.cc/paper_files/paper/2023/file/8c3caae2f725c8e2a55ecd600563d172-Paper-Conference.pdf} {\emph {\bibinfo {booktitle} {Advances in Neural Information Processing Systems}}},\ Vol.~\bibinfo {volume} {36}\ (\bibinfo  {publisher} {Curran Associates, Inc.},\ \bibinfo {year} {2023})\ pp.\ \bibinfo {pages} {44792--44819}\BibitemShut {NoStop}%
\end{thebibliography}%

\onecolumngrid
\appendix

\section{The differential geometry of hyperspheres}
\label{sec:differential_geometry}

Here, we introduce the main concepts of differential geometry used in this paper, as they apply to the study of hyperspheres as subspaces of Euclidean spaces. We give an elementary account of the subject, as given in Ref.~\cite{Oneill2006}, and introduce also a few sophisticated results that we need.

\textbf{Additional notation:} Besides the conventions given in the main text, we define $\mathbf{e}_{\ell}\coloneq(\delta_{\ell,j})^{d}_{j=1}\in\euclidean$, where $\delta_{\ell,j}$ is the Kronecker delta symbol, for the canonical basis vectors in $\euclidean$, use the letters $\gamma$ and $\tau$ for functions defining differentiable curves on manifolds, and write $\gamma^{\prime},\gamma^{\prime\prime}$ for the first and second derivatives of a real function of a single variable $\gamma$. 
We use the symbol $\langle\,\,;\,\rangle$ to denote the canonical Euclidean inner product function $\langle\,\,;\,\rangle:\euclidean\cross\euclidean\rightarrow\mathbb{R}$ given by $\langle\vecv;\vecu\rangle\coloneq\sum^{d}_{\ell=1}v_{\ell}u_{\ell}$ with $v_{\ell},u_{\ell}$ being the components of the basis vector $\mathbf{e}_{\ell}$ in $\vecv,\vecu$ respectively, and we use the symbol $\norm{\,\,}$ for the induced Euclidean two-norm on $\euclidean$. 

\vspace{.2cm}

\textbf{Tangent spaces:} Consider the vector space $\euclidean$ with its standard Euclidean inner product and recall that the unit hypersphere is the set $\sphere$ of points $\vecx\in\euclidean$ that satisfy the quadratic equation $\norm{\vecx}^{2}=1$. 
Our approach is to define the geometric properties of $\sphere$ by studying the properties of sets of differentiable curves $\gamma: [0, \, 1] \rightarrow \sphere$ on them. 
But first, let us build intuition by considering sets of differentiable curves in $\euclidean$. 
Let $\veca\in\euclidean$ be some point and let $\tau:[0,1]\rightarrow\euclidean$ be an arbitrary differentiable curve with $\tau(0)=\veca$. 
Since $\tau$ is differentiable, there exists a unique tangent vector to $\tau$ at the point $\veca$ given by its derivative $\tau^{\prime}(0)$. 
We call the set of all possible tangent vectors to curves of this form the \emph{tangent space of $\euclidean$ at $\veca$}, denoted by $T_{\veca}\euclidean$.
Since a tangent vector to a curve in $\euclidean$ defines a unique tangent line to that curve, every element of $T_{\veca}\euclidean$ can be obtained by computing tangent vectors to lines of the form $\veca+t\vecv$ for some $\vecv\in\euclidean$. 
This makes $T_{\veca}\euclidean$ an isomorphic copy of $\euclidean$ but centered at the point $\veca$. 
Therefore, the points in $\euclidean$ that represent canonical basis vectors form a basis for the tangent spaces of $\euclidean$ at every point. 
This justifies using the same notation $\mathbf{e}_{j}$ for the $j$th canonical basis vector of any tangent space of $\euclidean$.
Similarly, let $T_{\veca}\sphere$ be the set of all the tangent vectors to the curves $\gamma:[0,1]\rightarrow\sphere$ with $\gamma(0)=\veca$ at the point $\veca\in\sphere$. 
Since for every $t\in[0,1]$ we must have $\langle\gamma(t);\gamma(t)\rangle=1$, we get by differentiation and the properties of the Euclidean inner product that $\langle\gamma(t);\gamma^{\prime}(t)\rangle=0$ for all $t\in[0,1]$, and in particular for $t=0$ we have $\langle\veca;\gamma^{\prime}(0)\rangle=0$, \ie the tangent space $T_{\veca}\sphere$ is the subspace of all $\vecv\in T_{\veca}\euclidean$ such that $\langle\veca;\vecv\rangle=0$. 

\vspace{.2cm}

\textbf{Natural gradient:} Tangent spaces allow us to define linear maps that approximate the effect of smooth maps at a point. 
For example, let $\loss:\euclidean\rightarrow\mathbb{R}$ be a smooth function. We know from vector calculus that near a point $\veca$, we have $\loss(\veca+\vecv)\approx \loss(\veca)+\sum^{n}_{j=1}(\partial_{x_{j}}\loss)_{\veca}v_{j}$. 
The second term of this sum can be interpreted as the action of the linear functional $\text{d}\loss_{\veca}:T_{\veca}\euclidean\rightarrow\mathbb{R}$ on $\vecv\in T_{\veca}\euclidean$ defined by $\text{d}\loss_{\veca}(\vecv)\coloneq\langle(\partialbf_{\vecx}\loss)_{a};\vecv\rangle=\sum^{n}_{j=1}(\partial_{x_{j}}\loss)_{\veca}v_{j}$, where the last equation is the expression of the inner product in the canonical basis of $T_{\veca}\euclidean$.
The linear functional $\text{d}\loss_{\veca}$ is called the \textit{differential} of $\loss$ at $\veca$, and its generalizations play a major role in differential geometry.
The vector $\Grad^{\text{euc}}_{\veca}(\loss)\coloneq(\partialbf_{\vecx}\loss)_{a}\in T_{\veca}\euclidean$ is called the (euclidean) \textit{gradient} of $\loss$ at $\veca$, and gives the direction of steepest ascent of $\loss$ at $\veca$. If $\veca\in\sphere$, we can define the \quotes{spherical} gradient of $\loss$ at $\veca$, also called \emph{natural gradient}, as the vector of steepest ascent of $\loss$ at $\veca$ among the vectors tangent to $\sphere$. 
We can calculate it by computing the differential $\text{d}\loss_{\veca}$ on the vectors in $T_{\veca}\sphere$.
Since $T_{\veca}\sphere$ consists of vectors in $T_{\veca}\euclidean$ that are orthogonal to $\veca$, for any $\vecv\in T_{\veca}\euclidean$ the vector $\Pi_{\veca}\vecv$ is in $T_{\veca}\sphere$, where $\Pi_{\veca}$ is the projection operator on $T_{\veca}\sphere$ with a matrix representation given by $\identity-[\veca][\veca]^{T}$. 
Then, since projection operators are symmetric, we have $\dd \loss_{\veca}(\Pi_{a}\vecv)=\langle(\partialbf_{\vecx}\loss)_{\veca};\Pi_{a}\vecv\rangle=\langle(\Pi_{a}(\partialbf_{\vecx}\loss)_{\veca});\vecv\rangle$, from which we get
\begin{align}
    \Grad^{\text{sph}}_{\veca}(\loss) = \Pi_{\veca}(\partialbf_{\vecx} \loss)_{\veca}\,.
    \label{eq:gradient_sphere}
\end{align}

\vspace{.2cm}

\textbf{Metric and geodesics:} So far, we have been assuming that the tangent spaces $T_{\veca}\euclidean$ have linear and Euclidean structures, in the sense that each $T_{\veca}\euclidean$ naturally inherits the same inner product function as $\euclidean$ for every point $\veca$ under its isomorphism with $\euclidean$. 
It turns out that this fact is the distinguishing feature of Euclidean geometry. A function that associates a positive definite inner product with each point in $\euclidean$ is called a \emph{metric}, and the metric that chooses the standard Euclidean inner product for every $T_{\veca}\euclidean$ is called the \emph{standard metric of }$\euclidean$. 
Fixing a choice of metric for $\euclidean$ transforms it into a \emph{Riemannian manifold}. 
Among other things, the standard metric allows us to define the (standard) length of a differentiable curve $\tau: [0, \, 1] \rightarrow \euclidean$ by integrating $\norm{\tau^{\prime}}$ over its domain, where the norm is defined using the inner product function given above for the tangent space of every point in $\tau$.
Since each $T_{\veca}\sphere$ is a linear subspace of $T_{\veca}\euclidean$ for every $\veca\in\sphere$, it inherits the inner product function of $T_{\veca}\euclidean$, allowing us to define a metric for $\sphere$ called the \emph{round metric}. 
Therefore, the length of a curve $\gamma:[0,1]\rightarrow\sphere$ in the round metric has the same form as if we computed it as a curve in $\euclidean$, given by
\begin{align}
L(\gamma) \coloneqq \int_{0}^{1} \dd{\eta} \, \sqrt{\langle\gamma^{\prime}(\eta);\gamma^{\prime}(\eta)\rangle}\,.
\label{eq:length_curve}
\end{align}
Using this definition of length, we can define a curve that is the analog of a straight line but in $\sphere$, called a \emph{geodesic}, which is a curve on $\sphere$ having minimal length between any pair of its points. 
For each point $\veca\in\sphere$ and each tangent vector $\vecv\in T_{\veca}\sphere$, there is a unique geodesic curve $\gamma:[0,1]\rightarrow\sphere$ with $\gamma(0)=\veca$ and $\gamma^{\prime}(0)=\vecv$, and it is obtained by minimizing $L(\gamma)$ over the set of all possible curves that satisfy these conditions. 
It turns out that this problem is equivalent to finding the minimum of the functional
\begin{align}
E(\gamma) \coloneq \int_{0}^{1} \dd{\eta} \, \left(\frac{1}{2}\langle\gamma^{\prime}(\eta);\gamma^{\prime}(\eta)\rangle + \lambda \left(\langle\gamma(\eta);\gamma(\eta)\rangle - 1\right)\right),
\end{align}
where the curve $\gamma$ represents the trajectory of a free particle, with the first term in parentheses being the particle's energy and the second term being a constraint term imposing that this particle should stay on $\sphere$ at all times, with $\lambda$ being a Lagrange multiplier. 
The solution to this problem is given by solving the Euler-Lagrange equation $\gamma^{\prime\prime} = \lambda \, \gamma$ with the initial conditions $\gamma(0) = \veca$ and $\gamma^{\prime}(0)=\vecv$, resulting in
\begin{align}
\gamma(\eta) = \cos(\eta \, \norm{\vecv}) \, \veca + \sin(\eta \, \norm{\vecv}) \, \frac{\vecv}{\norm{\vecv}} \, .
\label{eq:gamma_solution}
\end{align}
Equation \eqref{eq:gamma_solution} shows that the geodesics of $\sphere$ are great circles, \ie the intersection of the hypersphere with a hyperplane that contains the initial point and the center of the hypersphere.

\vspace{.2cm}

\textbf{Exponential map and parallel transport:} As we have already said, the geodesics of the Riemannian manifold are analogs of straight lines in Euclidean spaces. 
This allows us to define the analog of the straight motion of geometric objects on the manifold. 
Formally, when the transported object is a point this is described by the \emph{Riemannian exponential map} $\expmap_{\veca}: T_{\veca}\sphere \rightarrow \sphere$, that takes a velocity vector $\vecv \in T_{\veca}\sphere$ at the point $\veca$ and maps it to the endpoint $\gamma(1) \in \sphere$ of the geodesic defined by the initial conditions $\gamma(0) = \veca$ and $\gamma^{\prime}(0) = \vecv$. 
To see how to transport tangent vectors, note that the equation of the geodesic curve corresponds to the rotation formula for an arbitrary point in $\euclidean$ on the 2-plane generated by the unit vectors $\veca$ and $\vecv/\norm{\vecv}$. 
If we look at the effect of this rotation on the hyperplane $T_{\veca}\sphere$, as a subset of $\euclidean$, we see that it not only changes the tangency point but also changes the relative orientation of the hyperplane so that it remains tangent to $\sphere$. 
Under this rotation, the straight line generated by a vector $\vecu\in T_{\veca}\sphere$ that is orthogonal to the vector $\vecv$ is mapped to a parallel line. 
This observation motivates the definition of a \emph{parallel transport map} $\mathcal{T}_{\eta \vecv}:T_{\veca}\sphere\rightarrow T_{\gamma(\eta)}\sphere$ given by the 
\begin{align}
\mathcal{T}_{\eta \vecv}(\vecu)  &\coloneq \left(\vecu - \frac{\langle\vecv;\vecu\rangle}{\norm{\vecv}^{2}} \, \vecv \right) + \frac{\langle\vecv;\vecu\rangle}{\norm{\vecv}^{2}} \, \gamma^{\prime}(\eta) \notag \\
&= \vecu - \sin(\eta \, \norm{\vecv}) \, \frac{\langle\vecv;\vecu\rangle}{\norm{\vecv}} \, \veca + (\cos(\eta \, \norm{\vecv}) - 1) \, \frac{\langle\vecv;\vecu\rangle}{\norm{\vecv}^{2}} \, \vecv \, ,
\label{eq:parallel_transport_definition}
\end{align}
where the expression in parentheses is the component of $\vecu$ that is invariant under this rotation. This map is called parallel transport because it changes the tangent vectors in such a way that their relative orientation w.r.t. the tangent hyperplane is the same before and after moving along a geodesic.

\vspace{.2cm}

\textbf{The differential:} To use these results in our method, we need to describe how to connect the description of $\sphere$ as a subspace of $\euclidean$ and its parameterization in hyperspherical coordinates. 
As we will show shortly, a parameterization of a $d_{1}$-dimensional surface $M$ in $\euclidean$ can be realized as a bijection between it and a space of parameters, which is usually some subset of $\mathbb{R}^{d_{1}}$. 
We can use this function to translate back and forth between parameters and points in our surface, so it would be convenient if we could also use this function to somehow translate geometric objects defined in points of our surface into geometric objects in the space of parameters.
This is done using the differential, which we will now proceed to define.
Let $\mathbf{f}:\mathbb{R}^{d_{1}}\rightarrow\mathbb{R}^{d_{2}}$ be a smooth map and $\veca\in\mathbb{R}^{d_{1}}$ be a point, the \emph{differential of $\mathbf{f}$} at $\veca$ is the linear map $\text{d}\mathbf{f}_{\veca}:T_{\veca}\mathbb{R}^{d_{1}}\rightarrow T_{\mathbf{f}(\veca)}\mathbb{R}^{d_{2}}$ defined by $\text{d}\mathbf{f}_{\veca}(\vecv)\coloneq(\partial_{t}\mathbf{f}(\veca+t\vecv))_{0}$ (in case the meaning of this expression is unclear, recall the notational conventions for derivatives given in the main text).
By the chain rule, we get $\text{d}\mathbf{f}_{\veca}(\vecv)=\langle(\partialbf_{\vecx}\mathbf{f})_{\veca};\vecv\rangle$. 
Expanding both $\mathbf{f}$ and $\vecv$ in their respective canonical bases, we obtain the vector $(\sum^{n}_{\ell=1}(\partial_{x_{\ell}}f_{j})_{\veca}v_{\ell})^{n}_{j=1}$, from which we can read the matrix representation $\jacobian_{\mathbf{f}}(\veca)$ of $\text{d}\mathbf{f}_{\veca}$ with coefficients $(j,\ell)$ given by $(\partial_{x_{\ell}}f_{j})_{\veca}$ called the \emph{Jacobian matrix of $\mathbf{f}$} at $\veca$.

\vspace{.2cm}
 
\textbf{Hyperspherical coordinates:} To show how the differential does the translation of geometric objects between a surface and its parameter space, we will use as an example the hyperspherical coordinate system on $\sphere$.
These coordinates are defined by choosing the ordered sequence $(\mathbf{e}_{1},\dots,\mathbf{e}_{d})$ of basis vectors and applying the iterative process described in Ref.~\cite{Blumenson1960}, producing a map $\vecx: D \subset \mathbb{R}^{d-1}\rightarrow\sphere$ where
\begin{align}
\vecx(\thetabf) \coloneqq \cos(\theta_{1}) \, \mathbf{e}_{1} + \sum_{j=1}^{d-2} \, \left(\prod_{\ell=1}^{j} \, \sin(\theta_{\ell})\right) \cos(\theta_{j+1}) \, \mathbf{e}_{j+1} + \left(\prod_{j=1}^{d-1} \, \sin(\theta_{j})\right) \, \mathbf{e}_{d} \, ,
\label{eq:hyperspherical_coordinates}
\end{align}
such that $D\coloneq[0, \, \pi]^{d-2} \cross [0, \, 2\pi)$ with $\thetabf = (\theta_{j})_{j=1}^{d-1}$ being a vector in $D$ representing hyperspherical angles. 
By construction, this map is smooth and has a continuous inverse, but it does not have a smooth inverse over all of $D$. To see this, we can compute the differential $\text{d}\vecx_{\boldsymbol{\phi}}$ at a generic point $\phi\in D\subset\mathbb{R}^{d-1}$ and see if the Jacobian matrix $\jacobian_{\vecx}(\boldsymbol{\phi})$ \footnote{which is the same Jacobian matrix as the one found in the main text just below Eq.\ref{eq:update_exact}, but without the index $\vecx$, which is dropped for convenience. For explicit formulas, see Eq.\eqref{eq:jacob_matrix} in App.\ref{app:jacobian_regularization}} has full rank there. 

\vspace{.2cm}

\textbf{Tangent vectors and pushforwards:} From now on, we will only consider the points of $\boldsymbol{\phi}\in \mathbb{R}^{d-1}$ for which the Jacobian has full rank.
For these points, the action of the differential on the canonical basis vectors of $T_{\boldsymbol{\phi}}\mathbb{R}^{d-1}$ defines a basis on the tangent spaces $T_{\vecx(\boldsymbol{\phi})}\sphere$ called the \emph{hyperspherical coordinate basis}. 
Therefore, the differential \quotes{pushes} of any vector in $T_{\boldsymbol{\phi}}\mathbb{R}^{d-1}$ to some vector in $T_{\vecx(\boldsymbol{\phi})}\sphere$, which is why $\text{d}\vecx_{\boldsymbol{\phi}}$ it is also called the \emph{pushforward} of $\vecx$ at the point $\boldsymbol{\phi}$, where \quotes{forward} refers to the choice of initial and final spaces being the same as the ones for the function $\vecx$.

\vspace{.2cm}

\textbf{Pullback of the metric:} To see what the differential does to the metric on $\sphere$, let $\veca\in\sphere$ be a point for which there is some $\boldsymbol{\phi}$ with $\boldsymbol{\phi}\coloneq\vecx^{-1}(\veca)$. 
Since we can push any pair of vectors $\vecv,\vecu\in T_{\boldsymbol{\phi}}\mathbb{R}^{d-1}$ to a pair of vectors $\text{d}\vecx_{\boldsymbol{\phi}}(\vecv),\text{d}\vecx_{\boldsymbol{\phi}}(\vecu)\in T_{\veca}\sphere$, we can define an inner product function $\langle\,\,;\,\rangle_{\boldsymbol{\phi}}:(T_{\boldsymbol{\phi}}\mathbb{R}^{d-1})^{2}\rightarrow\mathbb{R}$ by $\langle\vecv;\vecu\rangle_{\boldsymbol{\phi}}\coloneq\langle\text{d}\vecx_{\boldsymbol{\phi}}(\vecv);\text{d}\vecx_{\boldsymbol{\phi}}(\vecu)\rangle_{\veca}$ for all $\vecv,\vecu\in T_{\boldsymbol{\phi}}\mathbb{R}^{d-1}$.
This inner product is called the \textit{pullback} of the Euclidean inner product give in $T_{\veca}\sphere$, since it is defined by \quotes{pulling back} the Euclidean inner product from $T_{\veca}\sphere$ to $T_{\boldsymbol{\phi}}\mathbb{R}^{d-1}$.
The matrix representation of this inner product is the \emph{ metric tensor for the hyperspherical coordinates} and is given by $\metric(\boldsymbol{\phi})\coloneq\jacobian_{\vecx}(\boldsymbol{\phi})^{T}\jacobian_{\vecx}(\boldsymbol{\phi})$, which is the same matrix as the one given in Eq.\ref{eq:metric_components} when evaluated on the fixed HW ansatz state.

\vspace{.2cm}

\textbf{Pullback of the natural gradient:} With the pullback of the metric, we can also express the natural gradient of a function $h$ on $\sphere$ at a point $\mathbf{a}$ as a vector in $T_{\boldsymbol{\phi}}\mathbb{R}^{d-1}$, which amounts to defining a \quotes{pullback} for the natural gradient.
First, let $\Tilde{\loss}(\thetabf)\coloneq(\Tilde{\loss}\circ\vecx)(\thetabf)=\loss(\vecx(\thetabf))$, which means that $\Tilde{\loss}:\mathbb{R}^{d-1}\rightarrow\mathbb{R}$ is a smooth function between Euclidean spaces. 
By the definition of the differential we have $\text{d}\Tilde{\loss}_{\boldsymbol{\phi}}(\vecv)=\langle(\partialbf_{\thetabf}\Tilde{\loss})_{\boldsymbol{\phi}};\vecv\rangle$ for all $\vecv\in T_{\boldsymbol{\phi}}\mathbb{R}^{d-1}$.
It is tempting to read the vector $(\partialbf_{\thetabf}\Tilde{\loss})_{\boldsymbol{\phi}}$ as the expression of the natural gradient on $\sphere$ as a vector in $T_{\boldsymbol{\phi}}\mathbb{R}^{d-1}$.
However, since the Euclidean norm in $T_{\boldsymbol{\phi}}\mathbb{R}^{d-1}$ is not the same as the pullback over $\vecx$ of the Euclidean inner product on $\sphere$, the correct expression for the natural gradient must be a vector $\vecu$ that satisfies the equation $\langle\vecu;\vecv\rangle_{\boldsymbol{\phi}}=\langle(\partialbf_{\thetabf}\Tilde{\loss})_{\boldsymbol{\phi}};\vecv\rangle$.
This guarantees that the image of $\vecu$ under the action of $\text{d}\vecx_{\boldsymbol{\phi}}$ is exactly given by $\Grad^{sph}_{\veca}(\loss)$, or in other words, that $\vecu$ is the \textit{pullback} of $\Grad^{sph}_{\veca}(\loss)$ under $\vecx$.
Notice that it is a pullback instead of a pushforward.
This is a consequence of the fact that the natural gradient is the dual vector of a linear functional, a.k.a. a pseudovector, and as such, it is transformed via a pullback under $\vecx$, not a pushforward.

\vspace{.2cm}

\textbf{Natural gradient in different basis:} Using the metric tensor $\metric(\boldsymbol{\phi})$, we can see that the matrix used to express the natural gradient on $\sphere$ as an element in $T_{\boldsymbol{\phi}}\mathbb{R}^{d-1}$ is given by $\metric^{-1}(\boldsymbol{\phi})(\partialbf_{\thetabf}\Tilde{\loss})_{\boldsymbol{\phi}}$, where we do a slight abuse of notation by identifying a vector with its matrix representation. 
This allows us to write the matrix representation of the natural gradient on $\sphere$, now again as a vector in $T_{\veca}\sphere$, but expressed in terms of hyperspherical coordinates, by applying the pushforward under $\vecx$ giving us $\jacobian_{\vecx}(\boldsymbol{\phi})\metric^{-1}(\boldsymbol{\phi})(\partialbf_{\thetabf}\Tilde{\loss})_{\boldsymbol{\phi}}$ which is the same expression found in the main text just below Eq.\ref{eq:update_exact} where the index $\vecx$ was dropped and the substitution $\boldsymbol{\phi}\rightarrow\thetabf_{t}$ was made.

\vspace{.2cm}

This finishes our discussion of all of the necessary background information to understand how our approach to VQA works. 
Next, we will briefly discuss how to extend our method to states with complex amplitudes. 
It will be slightly more technical and will invoke some results without stating them.

\vspace{.2cm}

\textbf{Spheres and the space of quantum states:} It is well known that the space of pure quantum states over $d$ levels is the \emph{complex projective space} $\mathbb{CP}^{d-1}$, which is a Riemannian manifold with the Fubini-Study metric~\cite{Bengtsson2017}. 
It is also a well known result that $\mathbb{CP}^{d-1}=\mathbb{S}^{2d-1}/U(1)$, where $U(1)$ acts as multiplication by a global phase. 
The Fubini-Study metric is the quotient metric associated with the round metric on $\mathbb{S}^{2d-1}$ \cite[Example 2.30]{Lee2018}. 
If we use the general HWE encoder as our ansatz~\cite{Farias2025} to prepare a parametrized quantum state on $\mathbb{CP}^{d-1}$, we can use the fact that a global phase $\varphi_{global}$ is not an observable to prove that for any loss function $\loss(\thetabf,\boldsymbol{\varphi})$, $\partial_{\varphi_{global}}\loss(\thetabf, \, \boldsymbol{\varphi})=0$. This implies that the natural gradient of $\loss$ in $\mathbb{S}^{2d-1}$ does not have any component in the direction of the global phase, so the natural gradient on $\mathbb{CP}^{d-1}$ is given by the same vector \cite[Section 3.6.2]{absil2008}, and a geodesic on $\mathbb{S}^{2d-1}$ that is defined by some natural gradient in $\loss$ is mapped into a geodesic on $\mathbb{CP}^{d-1}$ since the differential of the quotient map from $\mathbb{S}^{2d-1}$ to $\mathbb{CP}^{d-1}$ is a Riemannian isometry~\cite{Bengtsson2017}. Therefore, we can do our optimization completely by embedding the state vectors in the parameterization of $\mathbb{S}^{2d-1}$ given by the ansatz after removing a global phase from the initial data vector.

\section{Barren-plateau analysis}
\label{app:barren-plateau}

In this appendix, we show that the amplitude encoder in hyperspherical coordinates allows exact computation of the variance of a loss function of the type $\loss = \bra{\psi} H \ket{\psi}$. 
The result is summarized in the following theorem: 
\begin{theorem}
    Let $H$ be a real-valued Hamiltonian, $\ket{\psi}$ an $n$-qubit quantum state amplitude-encoded in hyperspherical coordinates and with support in a $d$-dimensional (sub)space, and $\loss = \bra{\psi} \, H \, \ket{\psi}$ a loss function. 
    Denote by $\vecv$ the Riemannian gradient of $\loss$, and by $v_{j}$ its $j$th component. 
    Then, 
    \begin{gather}\label{eq:varL_result}
    \begin{aligned}
        \Var[\loss] = \frac{d \, \trace(H^{2}) - \trace(H)^{2}}{d^{2} \, (d + 2) / 2} \, ,
    \end{aligned}
    \end{gather}    
    and
    \begin{gather}\label{eq:var_result}
    \begin{aligned}
        \Var[v_{j}] = \frac{d \, (d + 2) \, \norm{h_{j}}_{2}^{2} - 2 \, (d + 2) \, H_{jj} \, \trace(H) + \trace(H)^{2}+ 2 \, \trace(H^{2})}{d \, (d + 2) \, (d + 4) / 4} \, ,
    \end{aligned}
    \end{gather}
    where $h_{j}$ is the $j$-th row of $H$, and $\norm{\cdot}_{{2}}$ is the Euclidean norm.
\end{theorem}
\begin{proof}
    In terms of the amplitudes $\vecx \in \sphere$, the loss function takes the quadratic form $\loss(\vecx) = \vecx^{T} \, H \, \vecx$ and the Riemannian gradient is $\vecv = \Pi_{\vecx} \, \partialbf_{\vecx} \loss = 2 \, (\identity - \vecx \vecx^{T}) \, H \, \vecx$, where $\partialbf_{\vecx} \loss = 2 \, H \, \vecx$ is the Euclidean gradient and $\Pi_{\vecx} \coloneqq \identity - \vecx \vecx^{T}$ is the orthogonal projector onto the tangent space to the sphere at $\vecx$. 
    The variance of $\loss$ and $v_{j}$, are, respectively, $\Var[\loss] = \mathbb{E}[\loss^{2}] - \mathbb{E}[\loss]^{2}$ and $\Var[v_{j}] = \mathbb{E}[v_{j}^{2}] - \mathbb{E}[v_{j}]^{2}$.
    The expectation values $\mathbb{E}[f(\vecx)] \coloneqq \int_{\sphere} \, \dd\mu(\vecx) \, f(\vecx)$ are w.r.t. the uniform measure $\dd\mu(x)$ on $\sphere$~\cite{Bengtsson2017}. 
    All the expectation values of interest can be expressed as linear combinations of moments of the random variables $x_j$ (\eg, see Ref.~\cite{fang2018}). 
    All \emph{odd} moments of $\vecx$ vanish, \ie
    \begin{align}\label{eq:odd_moments}
        \mathbb{E}[x_{i_1}x_{i_2}\cdots x_{i_{\ell}}] = 0 \quad \forall \, \ell\text{ odd}\,,
    \end{align}
    while all the even moments have a closed-form expression. 
    In particular,
    \begin{spreadlines}{2ex}
        \begin{align}\label{eq:even_moments}
            \mathbb{E}[x_i \, x_j] &= \frac{\delta_{ij}}{d} \notag\\
            \mathbb{E}[x_i\, x_j \, x_k \, x_\ell] &= \frac{\delta_{ij} \, \delta_{k\ell} + \delta_{ik} \, \delta_{j\ell} + \delta_{i\ell} \, \delta_{jk}}{d \, (d + 2)} \\
            \mathbb{E}[x_i \, x_j\,x_k\,x_\ell \, x_p \, x_q] &= \frac{\delta_{ij} \, \delta_{k\ell} \, \delta_{pq} + \text{permutations}}{d \, (d + 2) \, (d + 4)} \, ,\notag
        \end{align}
    \end{spreadlines}
    where $\delta_{ij}$ is the Kronecker symbol and the permutations in the last line denote $14$ similar terms corresponding to all the ways of partitioning the indices $(i,j,k,\ell,p,q)$ into three unordered and disjoint pairs of indices. 
    The calculation of $\Var[\loss]$ is then straightforward:
    \begin{align}
    \Var[\loss] = \sum_{i, j, k, \ell} \, H_{ij} \, H_{k\ell} \, \mathbb{E}[x_i \, x_j \, x_k \, x_\ell] - \left(\sum_{ij} \, H_{ij}\,\mathbb{E}[x_i \, x_j]\right)^{2} = \frac{2 \, \trace(H^{2}) + \trace(H)^{2}}{d \, (d + 2)} - \frac{\trace(H)^{2}}{d^{2}} \, ,
    \end{align}
    from which Eq. \eqref{eq:varL_result} follows.

    For the $j$-th component of the Riemannian gradient, $v_{j}$, we first notice that $\mathbb{E}\big[v_{j}\big]=0$ follows immediately from (\ref{eq:odd_moments}), while
    $\mathbb{E}\big[v_{j}^{2}\big]$ can be computed using (\ref{eq:even_moments}) after noticing that
    $\frac{1}{4}v_{j}^2=(H\vecx)_j^2-2(H\vecx)_j\,x_j\,(\vecx^T H \vecx)+ x_j^2\,(\vecx^T H \vecx)^2$. 
    Namely, the first term is
    \begin{align}\label{eq:var-grad-1}
    \mathbb{E}\big[(H\vecx)_j^2\big] = \sum_{ik}H_{ji}\,H_{jk}\,\mathbb{E}[x_i\,x_k] = \frac{1}{d}\sum_i H_{ji}^2 = \frac{\norm{h_j}_2^2}{d} \, .
    \end{align}
    The second term is
    \begin{gather}\label{eq:var-grad-2}
    \begin{aligned}
    \mathbb{E}\big[(H\vecx)_j \, x_j \, (\vecx^{T} H \vecx)\big] = \sum_{i, k, \ell} H_{ji} \, H_{k \ell}\,\mathbb{E}[x_j \, x_i \, x_k \, x_\ell] =\frac{H_{jj} \, \trace(H) + 2 \, \norm{h_{j}}_2^{2}}{d \, (d + 2)} \, ,
    \end{aligned}
    \end{gather}
    and the last term reads
    \begin{align*}
        \mathbb{E}[x_j^{2}(\vecx^{T} H \vecx)^{2}] &= \sum_{i, k, p, q} \, H_{ik} \, H_{pq} \, \mathbb{E}\big[x_j \, x_j \,x_i \,x_k \,x_p \,x_q\big] =\frac{1}{d \, (d + 2) \, (d + 4)} \sum_{i, k, p, q} \, H_{ik} \, H_{pq} \big(\delta_{ik} \, \delta_{k\ell} \, \delta_{pq} + \text{permutations}\big) \, .
    \end{align*}
    There are $15$ different pairings of the indices $(j_1, j_2, i, k, p, q)$ into triples of unordered disjoint pairs (here we use subscripts $j_{1}, j_{2}$ to distinguish the two occurrences of the same symbol $j$ for counting purposes). 
    Of these, $3$ have $j_1, \, j_2$ paired together: one of the type $(j_1, \, j_2)(i, \, k)(p, \, q)$, leading to $\sum_{i, k, p, q} \, H_{ik} \, H_{pq} \, \delta_{ik} \, \delta_{pq} = \trace(H)^{2}$, 
    and $2$ of the type $(j_1, j_2) \, (i, p) \, (k, q)$, which give $\sum_{i, k, p, q} \, H_{ik} \, H_{pq} \, \delta_{ip} \, \delta_{kq} = \sum_{i,k} \, H_{ik} \, H_{ik} = \trace(H^{2})$; 
    $4$ pairings have each $j$ paired with a different index and the leftovers pair together, \eg $(j_1, i) \, (j_2, \, k)(p, \, q)$, which yield $\sum_{i, k, p, q} \, H_{ik} \, H_{pq} \, \delta_{ji} \, \delta_{jk} \, \delta_{pq} = H_{jj} \, \trace(H)$; 
    and, finally, $8$ pairings have each $j$ paired with a different index with the left-over pair linking the two $H$ factors, \eg $(j_{1}, \, i) \, (j_{2}, \, p) \, (k, \, q)$, which give $\sum_{i, k, p, q} \, H_{ik} \, H_{pq} \, \delta_{ji} \, \delta_{jp} \, \delta_{kq} = \sum_{k} \, H_{jk} \, H_{jk} = \norm{h_{j}}_{{2}}^{2}$.
    It follows that
    \begin{align}\label{eq:var-grad-3}
        \mathbb{E}[x_{j}^{2} \, (\vecx^{T} H \vecx)^{2}] = \frac{\trace(H)^{2} + 2 \, \trace(H^{2}) + 4 \, H_{jj} \, \trace(H) + 8 \, \norm{h_j}_{{2}}^2}{d \, (d + 2) \, (d + 4)}\,,
    \end{align}
    and putting all together leads to $\Var[v_{j}]$ given by Eq. \eqref{eq:var_result}.
\end{proof}

\section{Riemannian optimization in the hypersphere}
\label{app:riemannian_opt_wolfe}

In optimization tasks, the gradient-descent method is the simplest algorithm to find a minimum of sufficiently smooth functions. 
We first recall that our circuit ansatze~\cite{Farias2025} parametrize quantum states $\ket{\psi(\thetabf(\vecx))}$ as points in a manifold that is the surface of a $(d-1)$-dimensional unit sphere $\sphere$.
We also recall that geodesics that connect two points $\vecx_{t}$ and $\vecx_{t+1}$ are not a straight line.
One way to move along a geodesic is to consider the great circle $\gamma(\etat)$ such that $\gamma(0) = \vecx_{t}$ and $\left(\frac{d}{d\etat}\gamma(\vecv_{t})\right)_{\etat = 0} = \zeta_{t}$.
Thus, we can use the great circle given by the \emph{exponential map}~\cite{absil2008,sato2021}: 
\begin{align}
    \vecx_{t+1} = \expmap_{\vecx_{t}}(\etat \, \vecv_{t}) = \cos(\etat \, \norm{\vecv_{t}}) \, \vecx_{t} + \sin(\etat \, \norm{\vecv_{t}}) \, \frac{\vecv_{t}}{\norm{\vecv_{t}}} \, ,
    \label{eq:update_exact_bla}
\end{align}
where $\norm{\vecv_{t}} \coloneqq \sqrt{\langle \vecv_{t} , \vecv_{t} \rangle_{\vecx_{t}}}$, and, in this context, the scaling factor $\etat$ can be understood as a learning rate. 
Equation \eqref{eq:update_exact} defines a closed line on $\sphere$, with the geodesic coming back to the original point when $\etat \, \norm{\vecv_{t}} = 2  \pi  j$, $\forall \, j \in \mathbb{Z}$.

\begin{algorithm}[!t]
\label{alg:conjugate}
\DontPrintSemicolon
{\bf Input:} Loss function $\loss$, transport function $\mathscr{T}^{(t)}$, tolerance $\varepsilon > 0$, maximum number of iterations $N_{\max}$ \;
{\bf Output:} Critical point $\thetabf_{t}^\star$ and $\loss(\thetabf_{t}^\star)$ \;
\caption{\label{alg:r-cg} Riemannian Conjugate Gradient}
\SetKwFunction{OptStepSize}{\text{\sc{OptStepSize}}}
{\;
    $\vecv_0 \leftarrow -\jacobian(\thetabf_{0}) \, \metric^{-1}(\thetabf_{0}) \, (\partialbf_{\thetabf}\loss(\thetabf))_{\thetabf_{0}}$ \;
    $\vecu_0 \leftarrow \vecv_0$ \;
    \For{$t \leftarrow 0, \, 1, \, \ldots, N_{\max}$} {
        $\mathrm{res}\leftarrow \braket{-\vecv_{t}, \, \vecu_{t}}_{\vecx_{t}}^{2} / \norm{\vecu_{t}}^2$ \;
        \If{$\abs{\mathrm{res}} < \varepsilon$}{
            \textbf{break} \;
        }
        $\etat \leftarrow $ \OptStepSize{$\loss, \, \mathscr{T}^{(t)}, \, \vecx_{t}, \, \vecv_{t}, \, \vecu_{t}$} \;
        $\thetabf_{t+1} \leftarrow \thetabf \left(\expmap_{\vecx_{t}}(\etat \, \vecu_{t}) \right)$ (see Eq. \eqref{eq:update_exact}) \;
        $\vecv_{t+1} \leftarrow -\jacobian(\thetabf_{t+1}) \, \metric^{-1}(\thetabf_{t+1}) \, (\partialbf_{\thetabf}\loss(\thetabf))_{\thetabf_{t+1}}$ \;
        $\beta_{t+1} \leftarrow$  from Eqs. \eqref{eq:beta} \;
        $\vecu_{t+1}\leftarrow$ from Eq. \eqref{eq:r-cgd} \;
        $\thetabf_{t}\leftarrow \thetabf_{t+1}; \,\vecv_{t} \leftarrow \vecv_{t+1}; \,\vecu_{t} \leftarrow \vecu_{t+1}; \, \vecx_{t} \leftarrow \vecx_{t+1}$ \;
    }
}   
\KwRet{$\thetabf_{t}$, $\loss(\thetabf_{t})$}
\end{algorithm}

We enhance the optimization scheme in Eq. \eqref{eq:update_exact_bla} using a Riemannian conjugate gradient (CG) method with convergence guarantees for lower-bounded and sufficiently smooth objective functions.
This addition distances the above optimizer even further from the first-~\cite{Stokes2020} and second-order~\cite{Halla2025secondorder} approximations in the literature.
In flat geometry, we formulate the CG as follows.
For $t = 0$, the direction of descent is given by the natural gradient, $\vecv_{0} = -\jacobian(\thetabf_{0}) \, \metric^{-1}(\thetabf_{0})(\partialbf_{\thetabf}\loss(\thetabf))_{\thetabf_{0}}$.
For $t > 0$, CG uses the natural gradient $\vecv_{t}$ and a memory from the previous direction, $\vecu_{t}$, \ie
\begin{align}
    \vecu_{t+1} = \vecv_{t+1} + \beta_{t+1} \, \vecu_{t},
    \label{eq:cgd}
\end{align}
for some choice of $\beta_{t+1} \geq 0$.
Since the memory direction $\vecu_{t}$ in  Eq. \eqref{eq:cgd} was calculated in the previous iteration, the equivalent update in curved space requires the transport of the vector $\vecu_{t}$ along the geodesic connecting $\vecx_{t}$ and $\vecx_{t+1}$ to properly perform the addition of these two vectors. 
Reference~\cite{sato2022} considers a general framework of Riemannian CG methods where they allow for a general map $\mathscr{T}^{(t)}: T_{\vecx_{t}}\mathcal{M} \rightarrow T_{\vecx_{t+1}}\mathcal{M}$ to transport the direction $\vecu_{t} \in  T_{\vecx_{t}}\mathcal{M}$ to $T_{\vecx_{t+1}}\mathcal{M}$,
\begin{align}
    \vecu_{t+1} =  \vecv_{t+1}+\beta_{t+1}s_{t}\mathscr{T}^{(t)}(\vecu_{t}),
    \label{eq:r-cgd}
\end{align}
where $s_{t}$ is a scaling parameter satisfying
\begin{align}
    0 < s_{t} \leq \min \left\{ 1, \frac{\sqrt{\braket{\vecu_{t}, \vecu_{t}}_{\vecx_{t}}}}{\sqrt{\braket{\mathscr{T}^{(t)}(\vecu_{t}), \, \mathscr{T}^{(t)}(\vecu_{t})}_{\vecx_{t+1}}}} \right\} \, .
    \label{eq:st}
\end{align}
While any choice of $s_{t}$ in the above interval suffices, here we choose $s_{t}$ to saturate the inequality on the right-hand side of Eq. \eqref{eq:st}.
In this work, we generate the sequence $\{\vecx_{t}\}$ on the manifold $\mathcal{M}$ from an initial point $\vecx_0 \in \mathcal{M}$ using the exponential map in Eq. \eqref{eq:update_exact}, and the search direction $\vecu_{t}$ lies in the tangent space $\mathrm{T}_{\vecx_{t}}\mathcal{M}$. 
 Here, we use the exact geodesic transport of an arbitrary vector $\zeta_{t} \in \mathrm{T}_{\vecx_{t}}\mathcal{M}$ on $\sphere$ along the geodesic (see Eq. \eqref{eq:parallel_transport_definition}),
\begin{gather}
\begin{aligned}
    \mathscr{T}^{(t)}(\zeta_{t}) \coloneq \mathcal{T}_{\etat \vecu_{t}}(\zeta_{t})\coloneq \zeta_{t} - \sin(\etat \, \norm{\vecu_{t}}) \, \frac{\braket{\vecu_{t}, \, \zeta_{t}}_{\vecx_{t}}}{\norm{\vecu_{t}}} \, \vecx_{t} + (\cos(\etat \, \norm{\vecu_{t}}) - 1) \, \frac{\braket{\vecu_{t}, \, \zeta_{t}}_{\vecx_{t}}}{\norm{\vecu_{t}}^{2}} \, \vecu_{t} \, .
    \label{eq:parallel-transport}
\end{aligned}
\end{gather}
In addition, we need to choose $\beta_{t+1}$ accordingly. Here, we consider one hybrid Riemannian CG method from Ref.~\cite{sato2022} (initially proposed for Euclidean optimization in Ref.~\cite{dai2001}), using the choices $\beta_{t+1}$ from Dai and Yuan~\cite{dai1999} as well as Hestenes and Stiefel~\cite{hestenes1952},
\begin{spreadlines}{3ex}
\begin{gather}
\begin{aligned}
    &\beta_{t+1} \coloneq \max\big\{0, \, \min\{\beta_{t+1}^{\text{DY}}, \, \beta_{t+1}^{\text{HS}}\}\big\} \, ;\\
    &\beta_{t+1}^{\text{DY}} \coloneqq \frac{\braket{-\vecv_{t+1}, \, -\vecv_{t+1}}_{\vecx_{t+1}}}{\braket{-\vecv_{t+1}, \, s_{t} \, \mathcal{T}_{\etat \vecu_{t}}(\vecu_{t})}_{\vecx_{t+1}} -\braket{-\vecv_{t}, \, \vecu_{t}}_{\vecx_{t}}} \, ; \\ 
    &\beta_{t+1}^{\text{HS}} \coloneqq \frac{\braket{-\vecv_{t+1}, \, -\vecv_{t+1}}_{\vecx_{t+1}} - \braket{-\vecv_{t+1}, \, \ell_{t} \, \mathcal{T}_{\etat \vecu_{t}}(-\vecv_{t})}_{\vecx_{t+1}}}{\braket{-\vecv_{t+1}, \, s_{t} \, \mathcal{T}_{\etat \vecu_{t}}(\vecu_{t})}_{\vecx_{t+1}} - \braket{-\vecv_{t}, \, \vecu_{t}}_{\vecx_{t}}} \, ,
    \label{eq:beta}
\end{aligned}
\end{gather}
\end{spreadlines}
where $\ell_{t}$ has a similar role to Eq. \eqref{eq:st}, but for the transport of the gradient.

So far, we assumed we have an appropriate step size $\etat >0$ --- also known as the learning rate.
However, to guarantee the convergence properties at each step, we need a step size that satisfies the $\mathscr{T}^{(t)}$-Wolfe conditions (\emph{cf.} Ref.~\cite{sato2022}),
\begin{align}
    \loss\left(\thetabf(\vecx_{t+1})\right) - \loss\left(\thetabf(\vecx_{t})\right) &\leq c_{1} \, \etat \, \braket{-\vecv_{t}, \, \vecu_{t}}_{\vecx_{t}} \, , 
    \label{eq:wolfe-1} 
    \\
    \big|\!\braket{-\vecv_{t+1}, \, \mathcal{T}_{\etat \, \vecu_{t}}(\vecu_{t})}_{\vecx_{t+1}}\!\big| & \leq c_{2} \, \big| \braket{-\vecv_{t}, \, \vecu_{t}}_{\vecx_{t}}\big| \, , 
    \label{eq:wolfe-2}
\end{align}
where $0 < c_{1} < c_{2} < 1$. 
In our numerical runs, we chose $c_{1} = 0.485$ and $c_{2} = 0.999$, and the subroutine {\sc OptStepSize} in Alg. \ref{alg:r-cg} can be seen as a standard backtracking algorithm ~\cite{truong2021}, outputting a learning rate that satisfies the strong Wolfe conditions in Eqs. \eqref{eq:wolfe-1}-\eqref{eq:wolfe-2}. 
We can modify this subroutine in problem-specific cases. We considered three heuristics: 
($i$) search for a learning rate in a region given by the power method, \ie, we start with
\begin{align}
\etat^{PM} = c_3 \norm{\vecu_{t}}^{-1} \, \arccos\left[\left(1+\left(\frac{\norm{\vecu_{t}}}{2\, \loss(\theta_{t})}\right)^{2}\right)^{-1/2}\right],
\end{align}
and then backtrack until Eqs. \eqref{eq:wolfe-1}-\eqref{eq:wolfe-2} are satisfied, with $c_3 > 1$ being problem-dependent --- for the spin chains, $c_3=1$ behaved well;
($ii$) in the case of molecules, a heuristic choice for the backtrack starting point that proved to work was $\max \left( \etat^{PM}, \frac{\pi}{4 \norm{\vecu_{t}}} \right)$;
($iii$) perform a \emph{golden-section search}~\cite{press2007} on $\loss\left(\thetabf(\vecx_{t+1})\right)$ to guarantee Eq. \eqref{eq:wolfe-1} and then backtrack until both conditions are satisfied.
Although all the strategies work, the amount of resources necessary for each one may vary in a problem-dependent way, and we reported the results that yielded the best results with the minimum resources.
A description of the entire procedure is presented in Alg. \ref{alg:conjugate}.

\section{Relation between EGT and QITE}
\label{app:relation}

Here, we show how the EGT, at each epoch, is equivalent to a first-order approximation of quantum imaginary-time evolution (QITE).
This result holds for the choices of loss function and Hamiltonians in the main text.
If the Hamiltonian is a projector, then each step of the EGT is \emph{exact} QITE.

The \quotes{classical} version of the loss function presented in the main text is the so-called \emph{Rayleigh quotient}, defined as
\begin{align}
    \operatorname{R}(\vecx) = \frac{\vecx^{T} \, H \, \vecx}{\vecx^{T}\vecx} \, .
    \label{eq:rayleigh_quotient}
\end{align}
The gradient of the Rayleigh quotient can be expressed as
\begin{align}
    \partialbf_{\vecx} \operatorname{R}(\vecx) &= \frac{2}{\normx^{2}} \, (H - \operatorname{R}(\vecx)) \, \vecx = 2 \, (H - \operatorname{R}(\vecx)) \, \vecx \, ,
    \label{eq:rayleigh_quotient_gradient}
\end{align}
where the second equality is valid when $\normx = 1\, , \, \forall \, \vecx$, and the norm was defined in the main text.
However, we can see that the loss function $\loss$ defined in Sec. \ref{sec:results} is
\begin{align}
    \loss(\vecx) &= \bra{\psi(\vecx)} \, H \, \ket{\psi(\vecx)} \nonumber \\
    &= \sum_{\ell, \, m} \, \frac{x_{\ell} \, x_{m}}{\normx^{2}} \, \bra{b_{\ell}} \, H \, \ket{b_{m}} \nonumber \\
    &= \sum_{\ell, \, m} \, \frac{x_{\ell} \, H_{\ell m} \, x_{m}}{\normx^{2}} \nonumber \\
    &\equiv \frac{\vecx^{T} \, H \, \vecx}{\vecx^{T}\, \vecx} = \operatorname{R}(\vecx)
\end{align}
Since the \quotes{quantum} loss $\loss(\vecx)$ and \quotes{classical} Rayleigh quotient $\operatorname{R}(\vecx)$ are effectively the same function, their gradients also match, with their difference being the fact that now the vector $\vecx$ is encoded in the amplitudes of the ($\ell_{2}$-normalized) quantum state $\ket{\psi(\vecx)}$.
Hence, 
\begin{align}
    \partialbf_{\vecx} \, \loss(\vecx) = 2 \, (H - \loss(\vecx)) \, \ket{\psi(\vecx)}
    \label{eq:loss_gradient}
\end{align}
Using \eqref{eq:loss_gradient}, we can calculate the norm of the gradient above as
\begin{align}
    \norm{\partialbf_{\vecx} \, \loss(\vecx)}^{2} = 4 \, \bra{\psi(\vecx)} \left(H^{2} - \loss^{2}(\vecx)\right) \, \ket{\psi(\vecx)} \, .
    \label{eq:norm_gradient}
\end{align}
Now, based on Refs.~\cite{Gluza2025, Suzuki2025}, we define the commutator $W_{H}(\vecx)$ as follows:
\begin{align}
    W_{H}(\vecx) \coloneqq 2 \, \left[\, H, \, \ketbra{\psi(\vecx)} \, \right] \, ,
\end{align}
where $[A, \, B] = A\, B - B \, A$ is the standard commutator.
By its definition above and by Eq. \eqref{eq:loss_gradient}, we can show that the action of $W_{H}(\vecx)$ on $\ket{\psi(\vecx)}$ is equal to the gradient of $\loss$, \ie
\begin{align}
    W_{H}(\vecx) \, \ket{\psi(\vecx)} &= 2 \, \left[\, H, \, \ketbra{\psi(\vecx)} \, \right] \, \ket{\psi(\vecx)} \nonumber \\
    &= 2 \, \left( H \, \ketbra{\psi(\vecx)} - \ketbra{\psi(\vecx)} \, H \right) \, \ket{\psi(\vecx)} \nonumber \\
    &= 2 \, \left( H - \loss(\vecx) \right) \ket{\psi(\vecx)} \nonumber \\
    &= \partialbf_{\vecx} \, \loss(\vecx) \, .
    \label{eq:commutator_on_state}
\end{align}
Using Eqs. \eqref{eq:norm_gradient} and \eqref{eq:commutator_on_state}, we can see that $\ket{\psi(\vecx)}$ is an eigenstate of $W_{H}^{2}(\vecx)$ such that
\begin{align}
    W_{H}^{2}(\vecx) \, \ket{\psi(\vecx)} &= 4 \, \left[\, H, \, \ketbra{\psi(\vecx)} \, \right] \, \left( H - \loss(\vecx)\right) \, \ket{\psi(\vecx)} \nonumber \\
    &= 4 \, \bigg( \loss(\vecx) \, H - \bra{\psi(\vecx)} \, H^{2} \, \ket{\psi(\vecx)} - \loss(\vecx) \, H + \loss^{2}(\vecx) \bigg) \, \ket{\psi(\vecx)} \nonumber \\
    &= - 4 \, \left( \bra{\psi(\vecx)} \, H^{2} \, \ket{\psi(\vecx)} - \loss^{2}(\vecx) \right) \, \ket{\psi(\vecx)} \nonumber \\
    &= - \norm{\partialbf_{\vecx}\,\loss(\vecx)}^{2} \, \ket{\psi(\vecx)} \, .
    \label{eq:eigenvalue_equation}
\end{align}
It immediately follows from Eq. \eqref{eq:eigenvalue_equation} that
\begin{gather}
    \begin{aligned}
        W_{H}^{2\ell}(\vecx) \, \ket{\psi(\vecx)} &= (-1)^{\ell} \, \norm{\partialbf_{\vecx}\,\loss(\vecx)}^{2\ell} \, \ket{\psi(\vecx)} \\
        W_{H}^{2\ell+1}(\vecx) \, \ket{\psi(\vecx)} &= (-1)^{\ell} \, \norm{\partialbf_{\vecx}\,\loss(\vecx)}^{2\ell + 1} \, \frac{W_{H}(\vecx)}{\norm{\partialbf_{\vecx}\,\loss(\vecx)}} \, \ket{\psi(\vecx)} \, .
        \label{eq:powers_of_commutator}
    \end{aligned}
\end{gather}
For $\eta \in \mathbb{R}$, Ref. \cite[Prop. 11 of App. III B]{Gluza2025} has shown that $e^{\eta \, W_{H}(\vecx)} \, \ket{\psi(\vecx)}$ is a first-order approximation of QITE up to corrections of order $\order{\eta^{2}}$.
Reference \cite[Lemma 2]{Suzuki2025} has shown that the same evolution exactly equates QITE if $H$ is either a projector onto a pure state (\ie $H \rightarrow \rho_{t}$, with $\rho_{t}$ being the density matrix of a target state, and $\loss$ becomes the pure state fidelity) or a projection onto a subspace spanned by a set of pure states.
Lastly, we use Eqs. \eqref{eq:commutator_on_state} and \eqref{eq:powers_of_commutator} to show below how, in our case, $e^{\eta \, W_{H}(\vecx)} \, \ket{\psi(\vecx)}$ is equivalent to the exponential map defined in Eq. \eqref{eq:expmap}:
\begin{align}
    e^{\eta \, W_{H}(\vecx)} \, \ket{\psi(\vecx)} &= \sum_{\ell=0}^{\infty} \frac{\eta^{\ell}}{\ell!} \, W_{H}^{\ell}(\vecx) \, \ket{\psi(\vecx)} \label{eq:expansion} \\
    &= \sum_{\ell=0}^{\infty} \, \frac{\eta^{2\ell}}{(2\ell)!} \, W_{H}^{2\ell}(\vecx) \, \ket{\psi(\vecx)} + \sum_{\ell=0}^{\infty} \, \frac{\eta^{2\ell+1}}{(2\ell+1)!} \, W_{H}^{2\ell+1}(\vecx) \, \ket{\psi(\vecx)} \nonumber \\
    &= \left(\sum_{\ell=0}^{\infty} \, \frac{(-1)^{\ell}}{(2\ell)!} \, \left(\eta \, \norm{\partialbf_{\vecx} \, \loss(\vecx)}\right)^{2\ell} \right) \, \ket{\psi(\vecx)} + \left(\sum_{\ell=0}^{\infty} \, \frac{(-1)^{\ell}}{(2\ell+1)!} \, \left(\eta \, \norm{\partialbf_{\vecx} \, \loss(\vecx)}\right)^{2\ell+1} \, \right) \, \frac{W_{H}(\vecx)}{\norm{\partialbf_{\vecx} \, \loss(\vecx)}} \, \ket{\psi(\vecx)} \nonumber \\
    &= \cos(\eta \, \norm{\partialbf_{\vecx} \, \loss(\vecx)}) \, \ket{\psi(\vecx)} + \sin(\eta \, \norm{\partialbf_{\vecx} \, \loss(\vecx)}) \, \frac{\partialbf_{\vecx} \, \loss(\vecx)}{\norm{\partialbf_{\vecx} \, \loss(\vecx)}} \, .
    \label{eq:qite_expmap}
\end{align}
As a consequence of Eq. \eqref{eq:qite_expmap}, for a generic Hamiltonian $H$, each iteration of the EGT optimizer can be interpreted as a first-order approximation of a QITE from the state $\ket{\psi(\vecx)}$ to a linear combination of $\ket{\psi(\vecx)}$ and $\partialbf_{\vecx} \, \loss(\vecx) / \norm{\partialbf_{\vecx} \, \loss(\vecx)}$, with their respective weights characterized by \quotes{imaginary time step} $\eta$.
If the Hamiltonian is a projector, then the QITE is exact.

Note that other methods for implementing \quotes{variational QITE} based on the first-order natural gradient, \eg $\operatorname{VQS}$~\cite{Yuan2019, Gacon2023} and $\operatorname{varQITE}$~\cite{McArdle2019}, can be recovered from Eq. \eqref{eq:qite_expmap}.
Let $\thetabf : \sphere \rightarrow \mathbb{R}^{d-1}$ be the standard hyperspherical coordinate transformation.
The particular form of this map is not relevant for this argument, but for the sake of completeness, see Ref.~\cite{Farias2025}).
For any smooth curve $\gamma: [0,1]\rightarrow\sphere$, the composition $\thetabf \circ \gamma : [0, \ 1] \rightarrow \mathbb{R}^{d}$ is a smooth curve in Euclidean space.
Therefore, we can apply Taylor's theorem to $\gamma$ and build the following first-order approximation:
\begin{align}
\thetabf \circ \gamma(\eta) \approx \thetabf \circ \gamma(0) + \eta \ (\thetabf \circ \gamma)^{\prime}(0) \ .
\end{align}
From Eq. \eqref{eq:qite_expmap}, we have  
\begin{equation}
\label{eq:geodesic_curve}
\gamma(\eta)\coloneq\cos(\eta \, \norm{\partialbf_{\vecx} \, \loss(\vecx)}) \, \vecx_{0} + \sin(\eta \, \norm{\partialbf_{\vecx} \, \loss(\vecx)}) \, \frac{\partialbf_{\vecx} \, \loss(\vecx)}{\norm{\partialbf_{\vecx} \, \loss(\vecx)}} \, 
\end{equation}
where we used the notation $\ket{\psi(\vecx)} \rightarrow \vecx_{0}$ and $\ket{\psi(\thetabf)} \rightarrow \thetabf_{0}$ (\ie, we reverted from Dirac notation to classical vector notation).
This expression implies that $\gamma(0) = \vecx_{0}$, and therefore $\thetabf \circ \gamma(0) = \thetabf(\vecx_{0}) = \thetabf_{0}$.
Also, by the chain rule, we have that
\begin{align}
(\thetabf\circ\gamma)^{\prime}(0) = \jacobian_{\vecx_{0}}^{-1}(\thetabf) \, \gamma^{\prime}(0) = \jacobian_{\vecx_{0}}^{-1}(\thetabf) \, \left(\partialbf_{\vecx}\loss(\vecx)\right)_{\vecx_{0}} \, ,
\end{align}
where we used Eq. \eqref{eq:geodesic_curve} to calculate $\gamma^{\prime}(0)$, and used the definition of $\jacobian^{-1}_{\vecx_{0}}(\thetabf)$ as the coordinate expression of the differential of the map $\thetabf:\sphere \rightarrow \mathbb{R}^{d-1}$.
In other words, we used that $\thetabf : \sphere \rightarrow \mathbb{R}^{d-1}$ is the inverse map of $\vecx : D \subset \mathbb{R}^{d-1} \rightarrow \sphere$, defined by Eq. \eqref{eq:hyperspherical_coordinates}.
Now, from the definition $\metric(\boldsymbol{\phi}) \coloneqq \jacobian_{\vecx}^{T}(\boldsymbol{\phi}) \, \jacobian_{\vecx}(\boldsymbol{\phi})$, we get that
\begin{align}
\jacobian_{\vecx_{0}}^{-1}(\thetabf) = \metric^{-1}(\thetabf) \, \jacobian_{\vecx_{0}}^{T}(\thetabf) \, .
\end{align}
We can now use the identity $\jacobian^{T}_{\vecx_{0}}(\thetabf) \, \left(\partialbf_{\vecx}\loss(\vecx)\right)_{\vecx_{0}} = \partialbf_{\thetabf}\Tilde{\loss}(\thetabf_{0})$, where $\Tilde{\loss}(\thetabf) \coloneqq (\loss \circ \vecx)(\thetabf)$ (see App. \ref{sec:differential_geometry}), to show that
\begin{align}
\thetabf \circ \gamma(\eta) \approx \thetabf_{0} + \eta \, \metric^{-1}(\thetabf) \, \partialbf_{\thetabf} \Tilde{\loss}(\thetabf_{0}) \, .
\label{eq:first_order_theta}
\end{align}
Equation \eqref{eq:first_order_theta} is exactly the quantum natural gradient in $\thetabf$ coordinates, as defined in Ref.~\cite{Stokes2020}, and is equivalent to a first-order approximation of QITE~\cite{Gacon2023, Yuan2019, McArdle2019}.

\section{Global convergence guarantee \emph{vs.} heuristic method}
\label{app:global_convergence}

In this appendix, we analyze the numerical performance of the EGT-CG defined in Eqs. \eqref{eq:expmap} - \eqref{eq:egt-cg-update} when the learning rate(s) $\etat$ and memory scaling factor $\beta_{t + 1}$ are chosen using ($i$) the strong Wolfe conditions described in App. \ref{app:riemannian_opt_wolfe}; or ($ii$) Bayesian optimization techniques~\cite{Movckus1975}.
We tested the Bayesian approach with $25$ and $50$ calls to the loss function to determine the learning rate $\etat$ per epoch $t$.

\begin{figure*}[!t]
    \includegraphics[width=\textwidth]{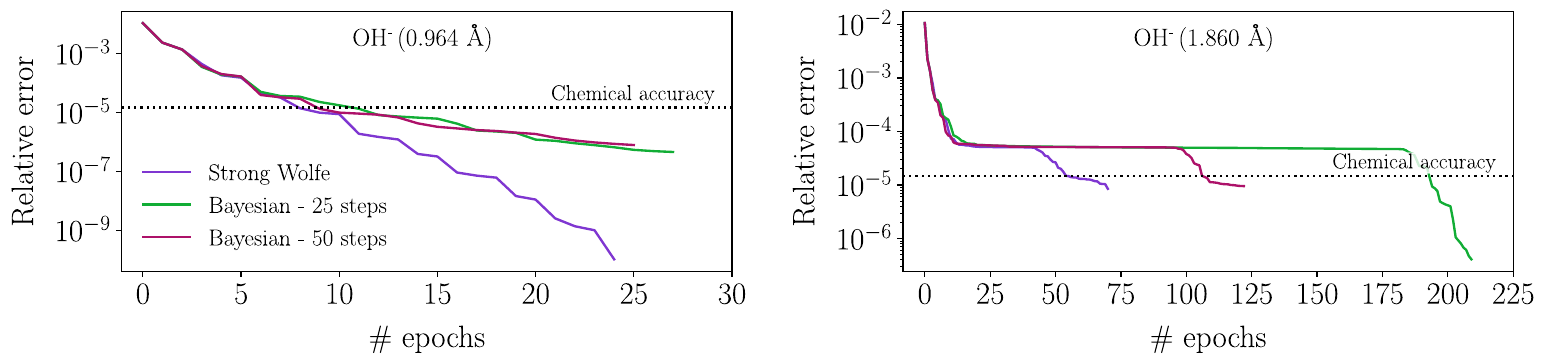}
    \caption{
        \textbf{Relaxation of the global convergence guarantees.}
        Relative ground-state energy estimation error for the \chem{OH^{-}} molecule as a function of the number of epochs, at two different bond lengths.
        Using the Jordan-Wigner transformation and the STO-3G basis set, the ground state is represented by $n = 12$ spin-orbitals as the active space with $k = 10$ electrons.
        Quantum circuit ansatz is the $\operatorname{HWE}_{k}$.
        The optimizer is the EGT-CG with learning-rate scheduling based on the strong Wolfe conditions (purple) and Bayesian optimization with $25$ (green) or $50$ (red) steps.
        Optimization is halted $15$ epochs after chemical accuracy is reached.
        (\emph{left}) Energy estimation at bond length of $0.964 \text{\AA}$. 
        In this case (bond length close to equilibrium), the ground state is not degenerate, and the spectral gap is $0.49\%$ of the ground-state energy.
        The three schedulings perform well, with the strong-Wolfe scheduling slightly outperforming both Bayesian approaches.
        (\emph{right}) Energy estimation at bond length of $1.860 \text{\AA}$.
        In this case, the ground state has degeneracy 6, and the spectral gap is $0.005\%$ of the ground-state energy.
        In this scenario, the strong-Wolfe scheduling reaches chemical accuracy almost $3$ times faster than both scheduling that use Bayesian optimization.
    }
    \label{fig:ohminus_wolfe_bayes}
\end{figure*}

In Fig. \ref{fig:ohminus_wolfe_bayes}(\emph{left}), we show the relative energy-estimation error as a function of the number of epochs in the training procedure.
The \chem{OH^{-}} molecule is at equilibrium bond length ($0.964 \text{\AA}$), and the spectral gap in this case is $\sim 0.49\%$ of the ground energy.
We see that the EGT-CG optimizer reaches chemical accuracy around the same number of epochs ($\sim 10$) for all three methodologies, even though the training method using the strong Wolfe conditions reaches relative energy errors more than $3$ orders of magnitude smaller by the end of training.
We also highlight that for this non-degenerate ground state, the strong Wolfe conditions perform better both in terms of calls to loss per epoch as well as calls to loss overall, since it reaches better precision in the same number of epochs.
The number of calls to the gradient function is similar in all three methods.
The numbers are displayed in Table \ref{tab:resources_oh_minus}.

\begin{table}[!h]
\begin{tabular}{ccccc}
\hhline{=====}
\begin{tabular}[c]{@{}c@{}}Bond\\ length ($\text{\AA}$)\end{tabular} &
  \begin{tabular}[c]{@{}c@{}}Learning rate\\ scheduling\end{tabular} &
  \begin{tabular}[c]{@{}c@{}}Calls\\ to loss\\ (per epoch)\end{tabular} &
  \begin{tabular}[c]{@{}c@{}}Calls\\ to loss\\ (total)\end{tabular} &
  \begin{tabular}[c]{@{}c@{}}Gradient \\ evaluations\end{tabular} \\
\hhline{=====}
\multirow{3}{*}{$0.964$} & Bayesian & $25$       & $669$  & $27$  \\
                         & Bayesian & $50$       & $1198$ & $25$  \\
                         & Strong Wolfe        & $\approx 18$ & $461$ & $24$  \\
\hhline{=====}
\multirow{3}{*}{$1.860$} & Bayesian & $25$       & $5115$ & $209$ \\
                         & Bayesian & $50$       & $5665$ & $122$ \\
                         & Strong Wolfe        & $\approx 17$ & $1245$ & $70$  \\
\hhline{=====}
\end{tabular}
\caption{Resource comparison for \chem{OH^{-}} at two bond lengths — degenerate and non-degenerate cases.  
The total counts reflect values for 15 consecutive epochs after chemical accuracy was achieved.  
In all cases, only one gradient evaluation per epoch was performed (\ie, the second Wolfe condition was always satisfied when the first was).
Here, \quotes{calls to loss} refers to the loss evaluations necessary to test the strong Wolfe conditions or the queries to perform the Bayesian trials only, \ie, it does not include the loss calls necessary for gradient estimation.
}
\label{tab:resources_oh_minus}
\end{table}

In Fig. \ref{fig:ohminus_wolfe_bayes}(\emph{right}), we show the relative energy-estimation error \emph{vs.} the number of epochs for the same molecule but a different bond length, $1.86\text{\AA}$.
At this bond length, the spectral gap is only $0.005\%$ of the ground-state energy.
The proximity between ground and first-excited states makes training more difficult.
We see that all methods display a plateau-like behavior after approximately $10$ epochs.
However, the method that uses the strong Wolfe conditions significantly improves the relative error before the $50$th epoch.
Meanwhile, the Bayesian optimization scheme with $50$ calls to the loss function to estimate the best $\etat$ takes $2\text{x}$ more epochs to start training significantly.
The Bayesian optimization with $25$ calls to the loss function takes approximately $3.5\text{x}$ more epochs than the strong Wolfe conditions.
In Table \ref{tab:resources_oh_minus}, we also see that this increases the number of calls to the loss function for the Bayesian methods by a factor bigger than $4$, while the strong Wolfe conditions require fewer than $3\text{x}$ more calls to the loss function.

\section{\chem{CH_{2}} optimization curves}
\label{app:opt_curves_ch2}

As discussed in the main text, the proposed warm start (Eq. \eqref{eq:warm_start}) is not the most adequate choice for \chem{CH_{2}} --- indeed, as one can verify via the numerical diagonalization, the Hartree state $\ket{1^{8} 0^{6}}$ is not the major contribution in none of its three degenerate ground-states ---, as seen in Fig. \ref{fig:molecules_comparison}(d), the overlap of $\ket{\psi_\text{warm}}$ with the degenerate ground-states subspace is of $\sim 10^{-5}$. 
Despite this, as can be seen in Fig. \ref{fig:ch2}, the conjugate methods, as well as Adam with a relatively large learning rate, were able to achieve chemical accuracy, although even EGT-CG took more epochs than usual when compared to the other molecules (see Fig. \ref{fig:molecules_comparison}(c)). 
This indicates that using the simple proposed warm start may not be the optimal choice for all molecules.
However, we stress that if there is any other initialization inspired by domain knowledge, preparing it with \hwe$_k$ remains trivial, whilst impractical for other ansatze. 

\begin{figure}[!h]
    \includegraphics[width=.5\columnwidth]{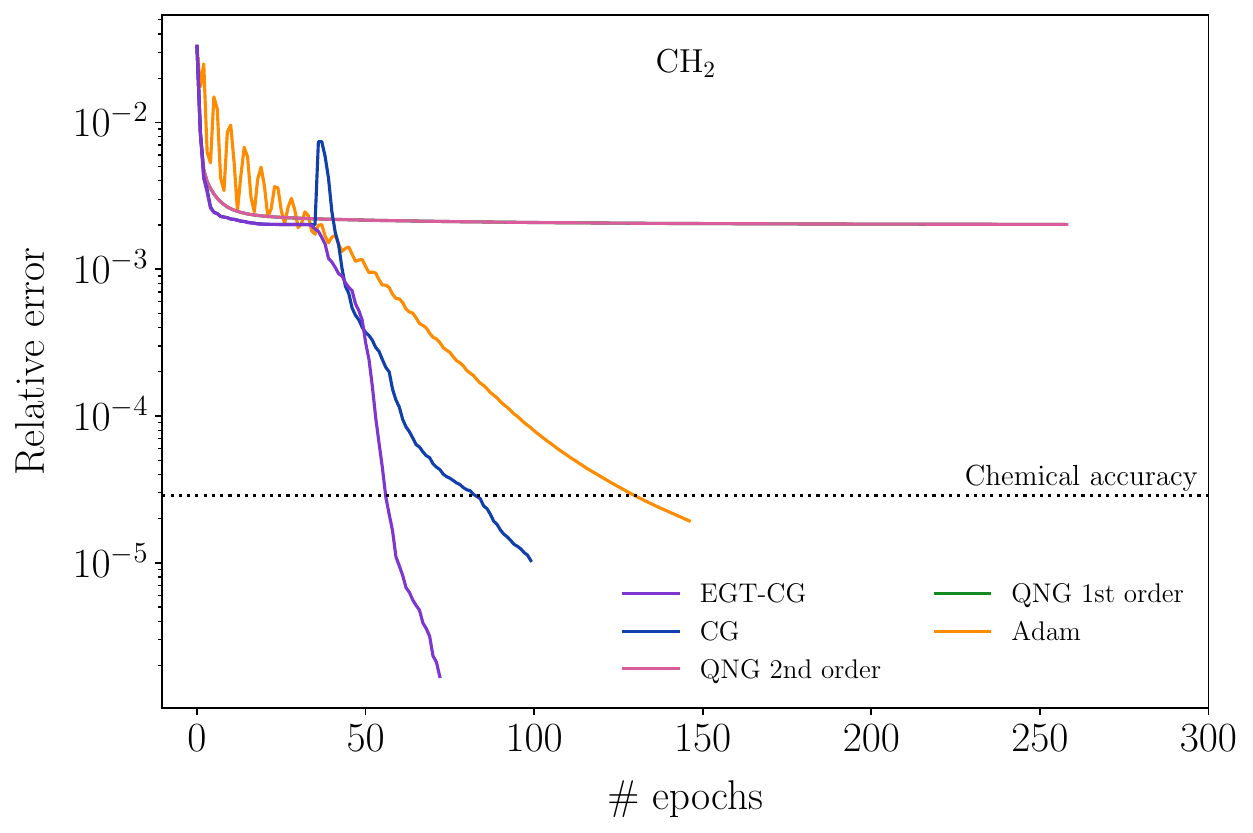}
    \caption{
            Relative ground-state energy estimation error for the \chem{CH_{2}} molecule as a function of the number of epochs.
            Using the Jordan-Wigner transformation and the STO-3G basis set, the ground state is represented by $n = 14$ spin-orbitals as the active space with $k = 8$ electrons. 
            The same warm start $\ket{\psi_\text{warm}}$ ($\alpha = 0.9$) was used, despite it not having a large overlap with the true degenerate ground-space.
            The ground-state has degeneracy $3$, and a spectral gap of $0.2\%$ of the ground-state energy, which favors most optimizer schemes to plateau at the first excited-state energy, whilst the EGT-CG escapes the plateau and converges to chemical accuracy monotonically, given the guarantees of the Strong Wolfe conditions. 
            Despite the non-monotonic path, Adam ($\eta = 0.01)$ escapes the plateau and manages to achieve chemical accuracy.
    }
    \label{fig:ch2}
\end{figure}

As observed for other molecules as well, the optimization curve of Adam is non-monotonic, exhibiting an oscillating behavior which, although may be an indication of an exceedingly large initial learning rate that could lead to divergence, has been shown to lead to faster convergence to chemical accuracy in all the performed simulations --- here, for \chem{CH_{2}}, $\eta = 0.01$, and the initial oscillatory behavior was essential for the optimizer to scape the plateau in the first excited states, which the QNG methods were not able to do. 
Also, we note that the CG optimization scheme, despite achieving chemical accuracy, follows a non-monotonic path, while the loss for EGT-CG is monotonously decreasing.

\section{Quantum resource estimation}
\label{app:quantum_resources}

Here, we present an estimation of the quantum resources necessary to run the ground-state optimization demonstrations for XXZ chains described in the main text. 
We focus on the number of calls to the loss function required to determine the learning rate $\etat$ per epoch $t$ per optimization scheme (Table \ref{tab:resources_xxz}).
We also provide explicit \cnot and parameter counts in each one of the circuit ansatze considered in the numerical demonstrations (Table \ref{tab:ncnots_xxz_ansatze}).

\begin{table}[!h]
\setlength{\tabcolsep}{7.5pt}
\begin{tabular}{c|cccc}
\begin{tabular}[c]{@{}c@{}}Number of \\ qubits ($n$)\end{tabular} &
\begin{tabular}[c]{@{}c@{}}HWE\\(EGT-CG)\end{tabular} &
\begin{tabular}[c]{@{}c@{}}HWE\\(CG)\end{tabular} &
\begin{tabular}[c]{@{}c@{}}HWE\\(QNG 1st)\end{tabular} &
\begin{tabular}[c]{@{}c@{}}HWE\\(QNG 2nd)\end{tabular} \\
\hhline{=====}
\\[-1.2em]
& \multicolumn{4}{c}{\textbf{Calls to loss (average per epoch)}} \\
\hline
4  & 6.4 & 6.4 & 50.0 & 50.0 \\
6  & 5.5 & 5.5 & 50.0 & 50.0 \\
8  & 6.1 & 5.5 & 50.0 & 50.0 \\
10 & 5.4 & 4.4 & 50.0 & 50.0 \\
12 & 5.2 & 3.4 & 50.0 & 50.0 \\
14 & 5.0 & 3.0 & 50.0 & 50.0 \\
16 & 4.2 & 3.1 & 50.0 & 50.0 \\
\hhline{=====}
\end{tabular}
\caption{
Resource comparison for \xxz ground state estimation, in terms of calls to the loss in each epoch necessary to determine the learning rate of each optimization scheme: 
The satisfaction of the strong Wolfe conditions for EGT-CG and CG; 
per-epoch greedy Bayesian optimization for QNG 1st and 2nd; 
constant initial learning rate for Adam (which does not require any additional calls to the loss).
The values are averages over $50$ different Haar-random initializations.  
For the CG methods, the testing of Wolfe conditions requires a variable number of calls to loss per epoch, so the reported number is the average number of calls throughout the optimization.
}
\label{tab:resources_xxz}
\end{table}

\renewcommand{\arraystretch}{1.4}
\begin{table}[!h]
\setlength{\tabcolsep}{7pt}
\centering
\begin{tabular}{c|cccc|cccc}
\multirow{2}{*}{\begin{tabular}[c]{@{}c@{}}Number\\ of qubits ($n$)\end{tabular}} &
\multicolumn{4}{c|}{\textbf{\# of CNOTs in ansatz}} &
\multicolumn{4}{c}{\textbf{\# of parameters in ansatz}} \\
& HWE & QOL & HE OP & HE & HWE & QOL & HE OP & HE \\
\hhline{=========}
4  & 2     & 10     & 28     & 12     & 1     & 5     & 30     & 15     \\
6  & 8     & 38     & 126    & 60     & 2     & 19    & 126    & 63     \\
8  & 102   & 138    & 504    & 248    & 7     & 69    & 510    & 255    \\
10 & 390   & 502    & 2040   & 1020   & 15    & 251   & 2046   & 1023   \\
12 & 3230  & 1846    & 8184    & 4092    & 49    & 923   & 8190    & 4095    \\
14 & 11372 & 6862    & 32760    & 16380    & 132   & 3431   & 32766    & 16383    \\
16 & 45738 & 25738    & 131056    & 65520    & 439   & 12869   & 131070    & 65535    \\
\hhline{=========}
\end{tabular}
\caption{
Number of CNOTs and parameters in each ansatz used for simulations of the XXZ ground state.
Although we did not run experiments for ansatze other than \hwe{}  for $n > 10$, we plot the CNOT counts for these circuits for a scaling comparison.
        Numerical details can be found in Tables \ref{tab:resources_xxz} and \ref{tab:ncnots_xxz_ansatze} in App. \ref{app:quantum_resources}.
}
\label{tab:ncnots_xxz_ansatze}
\end{table}

\section{Ground-state optimization for the Ising model}
\label{app:tfim}

As an additional example of ground states belonging to mixed particle-number subspaces, we consider the one-dimensional $n$-qubit transverse-field Ising model (\tfim), with Hamiltonian $H_{\operatorname{TFIM}} \coloneqq -\sum_{j=0}^{n-1} \, \left( Z_{j} \, Z_{j+1} + h \, X_{j} \right)$ and closed-boundary condition, where $h$ is the transverse field strength. 
As demonstrated in Ref.~\cite{Yu2025}, if $h = \order{(k / n)^{a}}$ with $a > 1/2$ and sufficiently large $n$, the ground state of \tfim{} in the computational basis is fully supported on states of HW at most $k$. 
Thus, the ground state can be searched using a circuit ansatz given by the union of $\operatorname{HWE}_{k}$ ansatze with all the Hamming weights up to $k$, \ie $\operatorname{HWE}_{\leq k} \coloneqq \bigcup_{j=0}^{k} \, \operatorname{HWE}_{j}$~\cite{Farias2025, Raj2025}.

\begin{figure}[!h]
    \includegraphics[width=0.5\columnwidth]{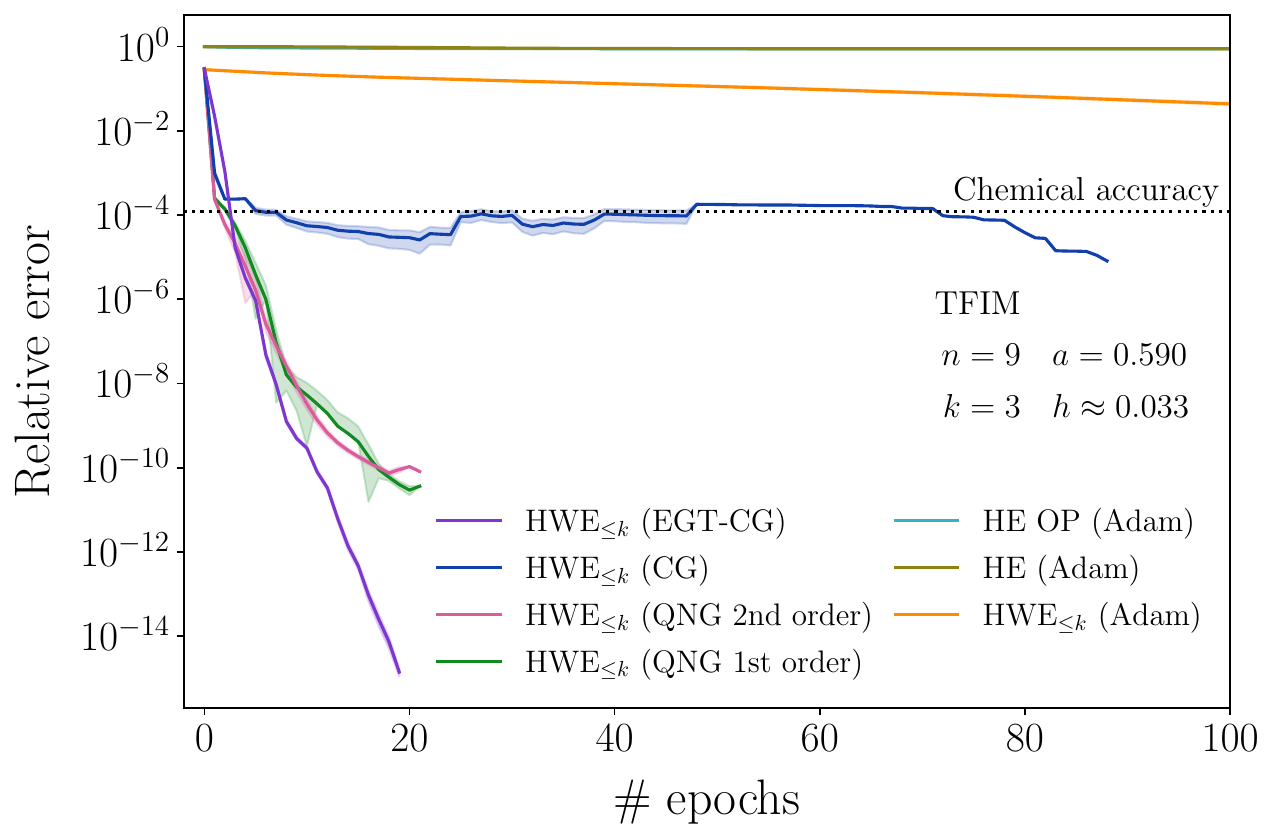}
    \caption{
        \textbf{Ground-state optimization for the TFIM Hamiltonian.}
        Optimization curves comparing different circuit ansatze and optimization schemes.
        With $\operatorname{HWE}_{\leq k}$, we used all the considered optimization schemes.
        We used Adam with QOL, HE, and HE OP.
        See \ref{ssec:details_numerical} for details.
        Each point is the average over $50$ Haar-random initializations, with error shades indicating the $95\%$ confidence interval.
        Chemical accuracy is plotted as a dotted black line.
    }
    \label{fig:tfim}
\end{figure}

Figure \ref{fig:tfim} shows the average relative error in ground-state energy estimation for a chain of $n = 9$ qubits as a function of the number of epochs.
We chose the parameters $a = 0.59$ and $h \approx 0.033$, with the ground state having support on computational basis states of HW up to $k = 3$.
We use the same $50$ Haar-random states as starting points of the optimizations for all combinations of circuits and optimization schemes.
We see that all versions of the natural gradient method, as well as the two conjugate gradient methods, achieve relative chemical accuracy, on average, in $3$ to $5$ epochs. 
While the geometry-aware methods continue to improve relative error estimation, the flat-space CG optimizer, on average, does not provide any improvement beyond the chemical accuracy level until around the $80$th epoch.
Meanwhile, all three geometry-aware methods continuously improve accuracy until being stopped at around the $20$th epoch.
The EGT-CG again displays the best performance at the end of the training, reaching an average relative energy-estimation error lower than $10^{-14}$.
The first- and second-order QNGs reach an average relative error of $\sim 10^{-10}$, which is $5$ orders of magnitude better than flat-space CG at the end of training but still $4$ orders of magnitude worse than the results using the EGT-CG optimizer.
As in the example of the \chem{H_{5}} molecule in the previous section, the EGT-CG performs better than flat-space CG around small-gradient landscape regions.
Even though it is not displayed in Fig. \ref{fig:tfim} for convenience, the combination of $\operatorname{HWE}_{\leq k}$ ansatz and Adam optimizer reached chemical accuracy around the $311$th epoch on average.
That is approximately $60$ times more epochs than the geometry-aware optimizers and flat-space CG.
The \he{} and \heop{} ansatze with Adam optimizer did not achieve chemical accuracy, being early-stopped, on average, at the $128$th and $159$th epochs, respectively.

\section{Efficient implementations of coordinate transformation and regularization}
\label{app:jacobian_regularization}

As mentioned in the main text, there are two steps required to estimate the natural gradient $\vecv_{t}$ in $\vecx$ coordinates.
The first is the estimation of the Euclidean gradient $(\partialbf_{\thetabf}\loss)_{\thetabf_{t}}$ obtained using quantum hardware measurements, while the second involves the computation of the product $\jacobian(\thetabf_{t}) \, \metric^{-1}(\thetabf_{t}) \, (\partialbf_{\thetabf}\loss)_{\thetabf_{t}}$. 
We now calculate the cost of the two remaining products to numerically calculate $\vecv_{t}$.
The cost to classically calculate the natural gradient $\metric^{-1}(\thetabf_{t}) \, (\partialbf_{\thetabf}\loss)_{\thetabf_{t}}$ is of $d$ floating-point multiplications.
Since the metric $\metric$ is a diagonal matrix in the coordinate basis, it can be element-wise inverted efficiently.
Moreover, the aforementioned product has the same numerical cost as an element-wise multiplication between the two $(d-1)$-dimensional arrays $(\partialbf_{\thetabf}\loss)_{\thetabf_{t}}$ and $\operatorname{diag}(\metric^{-1})$.

We now estimate the numerical cost of applying the coordinate transformation using the Jacobian matrix $\jacobian$.
Under the assumption of $\vecx \in \sphere$, the Jacobian matrix for the hyperspherical coordinate transformation can be straightforwardly calculated using its definition as
\begin{spreadlines}{3ex}
    \begin{align}
        \jacobian(\thetabf) &= \begin{pmatrix}
        \partial_{\theta_{1}} \! \cos(\theta_{1}) & \partial_{\theta_{2}} \! \cos(\theta_{1}) & \cdots & \partial_{\theta_{d-1}} \! \cos(\theta_{1}) \\
        \partial_{\theta_{1}} \! \sin(\theta_{1})\cos(\theta_{2}) & \partial_{\theta_{2}} \! \sin(\theta_{1})\cos(\theta_{2}) & \cdots & \partial_{\theta_{d-1}} \! \sin(\theta_{1})\cos(\theta_{2}) \\
        \vdots & \vdots & \ddots  & \vdots \\
        \partial_{\theta_{1}} \! \prod_{\ell=1}^{d-2}\sin(\theta_{\ell}) \! \cos(\theta_{d-1}) & \partial_{\theta_{2}} \! \prod_{\ell=1}^{d-2}\sin(\theta_{\ell}) \! \cos(\theta_{d-1}) & \cdots & \partial_{\theta_{d-1}} \! \prod_{\ell=1}^{d-2}\sin(\theta_{\ell}) \! \cos(\theta_{d-1}) \\
        \partial_{\theta_{1}} \! \prod_{\ell=1}^{d-1}\sin(\theta_{\ell}) & \partial_{\theta_{2}} \! \prod_{\ell=1}^{d-1}\sin(\theta_{\ell}) & \cdots & \partial_{\theta_{d-1}} \! \prod_{\ell=1}^{d-1}\sin(\theta_{\ell}) \\
        \end{pmatrix} 
        \notag \\
        &=\begin{pmatrix}
        -\sin(\theta_{1}) & 0 & \cdots & 0 \\
        \cos(\theta_{1})\cos(\theta_{2}) & -\sin(\theta_{1})\sin(\theta_{2}) & \cdots & 0 \\
        \vdots & \vdots & \ddots & \vdots \\
        \cos(\theta_{1})\!\prod_{l=2}^{d-2}\sin(\theta_{l})\!\cos(\theta_{d-1}) & \sin(\theta_{1})\!\cos(\theta_{2})\!\prod_{l=3}^{d-2}\sin(\theta_{l})\!\cos(\theta_{d-1}) & \cdots & -\prod_{l=1}^{d-2}\!\sin(\theta_l)\!\sin(\theta_{d-1}) \\
        \cos(\theta_{1})\!\prod_{l=2}^{d-1}\!\sin(\theta_{l}) & \sin(\theta_{1})\!\cos(\theta_{2})\!\prod_{l=3}^{d-1}\sin(\theta_{l}) & \cdots & \prod_{l=1}^{d-2}\!\sin(\theta_{l})\cos(\theta_{d-1}) \\
        \end{pmatrix} \, .
    \label{eq:jacob_matrix}
    \end{align}
\end{spreadlines}

Equation \eqref{eq:jacob_matrix} shows that $\jacobian(\thetabf) \in \mathbb{R}^{d \times d-1}$ is a \emph{tall lower Hessenberg} matrix, \ie a tall matrix that has nonzero elements only in the lower triangle and the first diagonal above the main diagonal (known as the first \emph{superdiagonal}) \cite[Chapter 2]{Datta2010}.
Then, the cost of the product between the Jacobian and the natural gradient array can be implemented efficiently by element-wise multiplication of the nonzero elements of each row of $\jacobian$ with the corresponding element(s) of the natural gradient.
Since the number of nonzero elements in $\jacobian$ goes from $1$ in the first row to $d-1$ in the last two rows, the cost of the matrix multiplication is $(d + 2)(d-1)/ 2$ multiplications and $(d-2)(d+1) / 2$ additions.
Hence, the overall time complexity of numerically calculating $\vecv_{t}$ is $\order{d^{2}}$.
Additionally, the nonzero elements of the $j$th Jacobian row can be reused in the calculation of the $(j+1)$th row.
Therefore, the entire matrix product can be performed while only one Jacobian row is stored in memory at a time, giving a space complexity of $\order{d}$.

\vspace{0.2cm}

For the aforementioned procedure to work, it is necessary to guarantee the invertibility of the metric tensor $\metric$ during training. 
Here, we describe how this can be efficiently guaranteed for the EGT optimizer.
As stated in the main text, the metric of $\sphere$ is diagonal in the coordinate basis, with components $g_{11} = 1$ and $g_{jj} = \prod_{\ell = 1}^{j-1} \, \sin^{2}(\theta_{\ell})$ for $j \in [2, \, d - 1]$.
The eigenvalues of $\metric$ are a function of $\thetabf$ and, more importantly, each eigenvalue $g_{jj}$ is associated with all angles $\{\theta_{\ell}\}_{\ell \in [j]}$.
If any of the angles decrease in magnitude such that $\theta_{\ell} \ll 1$ during training, then all eigenvalues $g_{jj} \rightarrow 0 \, , \forall \, j \geq \ell$.
Such a metric then becomes degenerate, and all directions $\theta_{j} \, , \forall \, j \geq \ell$ become unexplored during training even under the use of pseudo-inversion.
Instead of resorting to standard techniques, \eg Tikhonov regularization, to solve this issue, we use the connection between the eigenvalues $g_{jj}$ and angles $\theta_{j}$ to create the following heuristic rule:
if $\sin(\theta_{j}) \approx \theta_{j} \lesssim \tau$, then $\theta_{j} \rightarrow \pi / 2$.
The constant $\tau$ is a threshold used to reset training in the $j$th direction and maintain the invertibility of $\metric$.
We empirically verified that $\tau = 10^{-3}$ yielded stable training for all numerical demonstrations presented in Sec. \ref{sec:results}.

\section{Estimating gradients with \hwe$_k$ ansatz without using the parameter-shift rule}
\label{app:hwe_gradient}

In this appendix, we derive an explicit expression for the gradient of the loss $\loss_{\psi}(\thetabf) = \bra{\psi(\thetabf)} H \ket{\psi(\thetabf)}$ on the state prepared by the ansatze used in this manuscript, presenting an improvement in terms of quantum resources when compared to the standard parameter-shift rule (PSR).
For this appendix, we introduce the subscript $\psi$ to $\loss$, indicating that the loss function is calculated w.r.t. state $\ket{\psi(\thetabf)}$.
This distinction will become relevant later.

Given $n$ qubits, the ansatze prepare quantum states with support on $d$-dimensional (sub)spaces.
The ansatz explores the effective Hilbert space using $M = d-1$ angles represented by the vector $\thetabf \in \mathbb{R}^{d-1}$.
We can write the prepared state $\ket{\psi(\thetabf)}$ as
\begin{align}
\ket{\psi(\thetabf)} = \cos(\theta_{1}) \, \ket{b_{1}} + \sum_{j=2}^{d-1} \left(\prod_{m=1}^{j-1}\sin(\theta_{m})\right) \, \cos(\theta_{j}) \, \ket{b_{j}} + \left(\prod_{j=1}^{d-1} \, \sin(\theta_{j})\right) \ket{b_{d}} \, .
\label{eq:state_expanded}
\end{align}
The $\ell$th component of the gradient of the state in Eq. \eqref{eq:state_expanded} is then given by
\begin{align}
\ket{\partial_{\theta_{\ell}}\psi(\thetabf)} &= \delta_{\ell 1} \! \cos\left(\theta_{\ell} + \frac{\pi}{2}\right) \! \ket{b_{1}} + \sum_{j=2}^{d-1} \! \omega_{\ell j} \! \prod_{m=1}^{j-1} \! \sin\left(\theta_{m} + \delta_{\ell m} \! \frac{\pi}{2}\right) \! \cos\left(\theta_{j} + \delta_{\ell j} \! \frac{\pi}{2}\right) \! \ket{b_{j}} + \prod_{j=1}^{d-1} \! \sin\!\left(\theta_{j} + \delta_{\ell j} \! \frac{\pi}{2}\right) \! \ket{b_{d}}, 
\label{eq:del_state_expanded}
\end{align}
where $\delta_{lj}$ is the Kronecker delta and $\omega_{\ell j}$ is the discrete step function $\omega_{\ell j} = 1$ if $\ell \leq j$, and zero otherwise.
From the definition of the metric $\metric$ for $\sphere$ in Eq. \eqref{eq:metric_components}, we see that $\ket{\varphi_{\ell}(\thetabf)} \coloneqq g_{\ell\ell}^{-1/2}(\thetabf) \, \ket{\partial_{\theta_{\ell}}\psi(\thetabf)}$ is a normalized quantum state $\forall \, \ell \in [d - 1]$, with nonzero $j$th components in $j \geq \ell$.
Defining the quantum state $\ket{\phi_{\ell}(\thetabf)} \coloneqq \left(\ket{\psi(\thetabf)} + \ket{\varphi_{\ell}(\thetabf)} \right) / \sqrt{2}$, we obtain
\begin{align}
(\partial_{\theta_{l}} \loss)_{\thetabf_{t}} = g_{\ell\ell}^{1/2} \, \left(2 \, \loss_{\phi_{\ell}}(\thetabf_{t}) - \loss_{\varphi_{\ell}}(\thetabf_{t}) -  \loss_{\psi}(\thetabf_{t})\right) \, ,
\label{eq:grad_loss_comps_final}
\end{align}
where, as introduced at the beginning of this appendix, the terms $\loss_{\phi_{\ell}}$, $\loss_{\varphi_{\ell}}$, and $\loss_{\psi}$ indicate the loss function calculated w.r.t. the quantum states $\ket{\phi_{\ell}(\thetabf_{t})}$, $\ket{\varphi_{\ell}(\thetabf_{t})}$, and $\ket{\psi(\thetabf_{t})}$.
Despite the similarity with PSR due to the parameters $\theta_{\ell}$ being shifted by $\frac{\pi}{2}$, Eq. \ref{eq:grad_loss_comps_final} is not a proper PSR since it lacks opposing, negative shifts, \ie shifts by $-\frac{\pi}{2}$.

\vspace{0.2cm}

Now, we follow Ref.~\cite{Abbas2023} in the way we analyze the time complexity of estimating $(\partialbf_{\theta}\loss)_{\thetabf_{t}}$ using quantum hardware measurements.
There, the time complexity to execute the loss function once is denoted as $\Time(\loss)$.
Meanwhile, Refs.~\cite{Anselmetti2021, Kottmann2021} derived parameter-shift rules for quantum circuits composed of $M$ Givens rotations and showed that the time complexity of estimating the entire gradient vector is $\Time(\text{PSR}) = 4 \, M \, \Time(\loss)$.
With that in mind, we now analyze the time complexity for Eq. \eqref{eq:grad_loss_comps_final}.
Since the last term on the right-hand side (r.h.s.) of Eq. \eqref{eq:grad_loss_comps_final}, $\loss_{\psi}$, is constant, it can be estimated only once and therefore has a time complexity of $\Time(\loss)$ independent of $\ell$.
The first term on the r.h.s. of Eq. \eqref{eq:grad_loss_comps_final} is the most interesting of the three.
As mentioned above, the state $\ket{\varphi_{\ell}(\thetabf)}$ only has nonzero amplitudes in the components $j \geq \ell$.
This means that this state does not require all $M$ parameters to be prepared, instead requiring $d-\ell$ parametrized gates.
Hence, if the number of parametrized rotations needed to estimate each $\loss_{\varphi_{\ell}}$ is $m$, then the complexity of estimating all $\ell$ elements is $\sum_{m = 1}^{M} \, m = \frac{M + 1}{2} \, \Time(\loss)$.
The first term on the r.h.s. of Eq. \eqref{eq:grad_loss_comps_final} depends on the preparation of a superposition between $\ket{\psi(\thetabf)}$ and $\ket{\varphi_{\ell}(\thetabf)}$.
Even though $\ket{\varphi_{l}(\thetabf)}$ can be, in general, individually prepared with a reduced number of parametrized gates, the same cannot be said about $\ket{\psi(\thetabf)}$.
Therefore, the time complexity of estimating all components $\loss_{\phi_{l}}$ is $M \, \Time(\loss)$.
Finally, the sum of the three components gives us the time complexity of estimating the full gradient vector, \ie $\Time((\partialbf_{\thetabf}\loss)_{\thetabf_{t}}) = \frac{3 \, (M + 1)}{2} \, \Time(\loss)$.
For sufficiently large $M$, we have
\begin{align}
    \Time((\partialbf_{\thetabf}\loss)_{\thetabf_{t}}) \approx \frac{3}{8} \, \Time(\text{PSR}) \, .
    \label{eq:time_complexity}
\end{align}
Equation \eqref{eq:time_complexity} demonstrates that estimating gradient vectors using Eq. \eqref{eq:grad_loss_comps_final} instead of PSR provides a $62.5\%$ reduction in time complexity and quantum resources used.
We also highlight that, even though significant, the resource reduction presented in Eq. \eqref{eq:time_complexity} is constant, and therefore it is still in line with the complexity results obtained in Ref.~\cite{Abbas2023}.

\end{document}